\documentclass[sigconf, nonacm]{acmart}

\usepackage[textsize=small]{todonotes}
\usepackage{forest}
\usepackage{xspace}
\usepackage{multirow}
\usepackage[ruled,noend,linesnumbered]{algorithm2e}
\usepackage{booktabs}
\usepackage{tikz}
\usepackage{forest}
\usepackage{changepage}

\DeclareFontEncoding{LS1}{}{}
\DeclareFontSubstitution{LS1}{stix}{m}{n}
\DeclareSymbolFont{STIXsymbols}{LS1}{stixscr}{m}{n}
\DeclareMathSymbol{\bbowtie}{\mathrel}{STIXsymbols}{"0E}
\DeclareMathSymbol{\lltimes}{\mathbin}{STIXsymbols}{"0F}

\usepackage{pgfplots,pgfplotstable}
\usepackage{siunitx}
\usepackage{wasysym}
\usepackage{listings}
\usepackage{fancyvrb}

\newif\ifArxivVersion

\ArxivVersiontrue

\usepackage[para,online,flushleft]{threeparttable}
\lstdefinestyle{mystyle}{
    basicstyle=\ttfamily
}
\usepgfplotslibrary{colorbrewer}
\pgfplotsset{compat=1.17}
\usepgfplotslibrary{statistics}
\usetikzlibrary{positioning}
\usetikzlibrary{shapes,arrows,calc,matrix,patterns}

\usetikzlibrary{shadows}

\newcommand{\nop}[1]{}

\newcommand{\lbl}{\operatorname{label}}

\newcommand{\queryname}{zero-materialisation answerable\xspace}
\newcommand{\qnameshort}{0MA\xspace}

\newcommand{\YA}{Yannakakis' algorithm\xspace}

\newcommand{\Att}{\mathit{Att}}
\newcommand{\ghw}{\mathit{ghw}}

\newcommand{\figARXIV}{\url{https://figshare.com/s/b9ba4b798760cf6af3a4}}

\forestset{qtree/.style={for tree={parent anchor=south, 
      child anchor=north,align=center,inner sep=0pt}}
}
\forestset{
  sn edges/.style={for tree={parent anchor=south, child anchor=north}},
  red subtree/.style={for tree={text=red},for descendants={edge=red}}
}

\AtBeginDocument{%
  \providecommand\BibTeX{{%
    \normalfont B\kern-0.5em{\scshape i\kern-0.25em b}\kern-0.8em\TeX}}}

\tikzset{
  my node style/.style={
    font=\scriptsize,
    top color=white,
    bottom color=white,
    rectangle,
    minimum size=6mm,
    draw=black,
    thick,
    drop shadow,
    align=center,
  }
}
\forestset{
  my tree style/.style={
    for tree={
      parent anchor=south,
      child anchor=north,
      l sep-=1pt,
      my node style,
      edge={draw=black, thick},
      edge path={
        \noexpand\path [draw, \forestoption{edge}] (!u.parent anchor) -- +(0,-7.5pt) -| (.child anchor)\forestoption{edge label};
      },
      if n children=3{
        for children={
          if n=2{calign with current}{}
        }
      }{},
      delay={if content={}{shape=coordinate}{}}
    }
  }
}

\pgfplotstableread{%
Author lw uw lq uq median
Duck-Star 0.01 4.35 0.12 1.955 0.39
Duck-Agg.  0.00 2.91 0.07 1.345  0.29
PG-Star 0.01 23.27  0.24  9.66 0.775
PG-Agg. 0.01 13.15  0.23 5.73  0.77
Duck-Star+Lydia  0.09 14.23   0.9   6.27 2.03
Duck-Agg.+Lydia  0.09 6.3 0.68 3.055  1.54
PG-Star+Lydia 0.01 107.56  1.49 45.17  8.49
PG-Agg.+Lydia 0.01 54.31 0.57 22.35 5.72
}\datatable

\definecolor{pastelBlue}{HTML}{B4CFEC }
\definecolor{pastelGreen}{HTML}{88ddc8}
\definecolor{pastelReddish}{HTML}{EDB4D0}
\definecolor{pastelThistle}{HTML}{E0BBE4}
\definecolor{pastelLumber}{HTML}{e89a7d}
\definecolor{pastelPeach}{HTML}{fecba5}
\definecolor{pastelYellow}{HTML}{f0db75}
\definecolor{pastelIvory}{HTML}{ded5ba}

\definecolor{DuckDB}{HTML}{88ddc8}
\definecolor{Postgres}{HTML}{e89a7d}
\definecolor{Spark}{HTML}{f0db75}

\pgfplotsset{every tick label/.append style={font=\tiny}}

\definecolor{darkgreen}{rgb}{0.0, 0.5, 0.0}

\newcommand{\QueryPlanNode}[3]{
  \begin{tabular}{l r}
        #1  \ &  \color{red} #2{s} \\
     \multicolumn{2}{c}{\color{darkgreen} \# #3}
  \end{tabular}
}
\newcommand{\QueryPlanNodePlain}[1]{
  \begin{tabular}{l r}
     \multicolumn{2}{c}{ #1}
  \end{tabular}
}

\newcommand{\ourSystem}{\textsc{YanRe}\xspace}

\begin{document}
\title{
Structure-Guided Query Evaluation: \\
Towards Bridging the Gap from Theory to Practice}

\ifArxivVersion
    \author{Georg Gottlob}
    \orcid{0000-0002-2353-5230}
    \author{Matthias Lanzinger}
    \orcid{0000-0002-7601-3727}
    \affiliation{%
      \institution{University of Oxford}
      \country{United Kingdom}
    }
    \email{georg.gottlob@cs.ox.ac.uk}
    \email{matthias.lanzinger@cs.ox.ac.uk}
    
    \author{Davide Mario Longo}
    \orcid{0000-0003-4018-4994}
    \author{Reinhard Pichler}
    \orcid{0000-0002-1760-122X}
    \author{Alexander Selzer}
    \affiliation{%
      \institution{TU Wien}
      \country{Austria}
    }
    
    \email{firstname.lastname@tuwien.ac.at}

    \author{Cem Okulmus}
    \orcid{0000-0002-7742-0439}
    \affiliation{%
      \institution{Ume\aa{} University}
      \country{Sweden}
    }
    \email{okulmus@cs.umu.se}
\fi

\begin{abstract}
Join queries involving many relations pose a severe challenge to 
today's query optimisation techniques. To some extent, this is due 
to the fact that these techniques do not pay sufficient attention
to structural properties of the query. 
In stark contrast, the Database Theory community has intensively studied structural properties of queries
(such as acyclicity and various notions of width) and proposed efficient query evaluation
techniques through variants of Yannakakis' algorithm. However, although most queries in practice actually are acyclic or have low width, structure-guided query evaluation techniques based on Yannakakis' algorithm have not found their
way into mainstream database technology yet.

The goal of this work is to address this gap between theory and practice and to demonstrate that the consideration of query structure can improve query evaluation performance on modern DBMSs significantly in cases that have been traditionally challenging. In particular, we study the performance of structure-guided query evaluation in three  architecturally distinct DBMSs by rewriting SQL queries into a sequence of SQL statements that express an execution of Yannakakis' algorithm. 
Moreover, we identify a class of queries that is particularly well suited for our approach and allows query answering in a variety of common scenarios without materializing any join. Through empirical evaluation we show that structure-guided query evaluation can make the evaluation of many difficult join queries feasible whereas their evaluation requires a prohibitive amount of time and memory on current DBMSs. %
\end{abstract}

\maketitle

\section{Introduction}
\label{sect:Introduction}

Query processing lies at the very heart of database applications and
systems -- with join queries arguably being the most fundamental and basic form of queries. A lot of research 
spanning over several decades has gone into optimising queries in general and join queries in particular. Consequently, in many practical cases, 
Database Management Systems (DBMSs) perform really well -- even as the amounts of data to be handled get bigger and bigger. However, there still remain queries where 
today's DBMSs struggle or simply fail. This is especially the case with queries that involve the join of many 
(say 10, 50 or even hundreds of) relations. 
Large join queries remain challenging even when all joins are made along foreign key constraints, one of the most common and basic cases for relational DBMSs. We will summarily refer to these kinds of queries as \emph{typical yet challenging}. 
Such queries are becoming more and more common. For example, queries automatically
generated by business intelligence tools may easily reach this size~\cite{DBLP:conf/sigmod/NeumannR18}.
It, therefore, is a requirement for DBMSs today to cope with~such~queries.

The traditional approach to evaluating a join query is to split it into a sequence of two-way joins. One of the main tasks of query optimisation is then to determine 
the optimal or at least a good join order. In particular, part of finding a good join order is avoiding the costly computation of large intermediate results as far as possible.
However, typical systems rely on some combination of heuristics and optimisation procedures to determine the join order for given queries. Hence, even for moderately large queries, the resulting optimisation problems become too difficult to solve exactly and the quality of the resulting join orders degrades quickly. For instance, PostgreSQL 14 by default performs a full search for the optimal plan only up to 11 joins before falling back on heuristic optimisation techniques. Sophisticated pruning methods and parallelisation have been shown to push this threshold higher~\cite{DBLP:conf/sigmod/ManciniKCMA22,%
DBLP:conf/icde/MageirakosMKCA22}, but the task still remains fundamentally challenging.
Moreover, the problem of huge intermediate results 
is not restricted to the choice of a bad join order. 
As has been shown in recent work on worst-case optimal 
joins~\cite{DBLP:journals/jacm/NgoPRR18,DBLP:journals/sigmod/NgoRR13}, 
it is an intrinsic deficiency of splitting join queries into a sequence of two-way joins. %
For queries of particular structure 
(typically, small queries with joins that do not 
follow
foreign key relationships)
and heavily skewed data, worst-case optimal joins 
may indeed help to avoid the generation of intermediate result tuples 
that do not contribute to the final result. But 
empirical studies of database queries 
have shown that most queries in practice are acyclic or almost acyclic,
involving mostly joins along foreign 
keys~\cite{DBLP:journals/vldb/BonifatiMT20,DBLP:journals/jea/FischlGLP21} and thus call for a different solution.

From a theoretical perspective, the problem of avoiding large intermediate results in join queries has long been considered as essentially solved. For
acyclic queries, Yannakakis'
algorithm~\cite{DBLP:conf/vldb/Yannakakis81} is well known to
guarantee query answering without any unnecessary intermediate results
by following the inherent tree-like structure of acyclic queries in
the evaluation of the query. 
From there, a rich theory of structural
decompositions and related notions of width was developed~\cite{DBLP:journals/jcss/GottlobLS02,2014grohemarx} that generalises the
acyclic case to general queries 
with guaranteed 
bounds (relative to some notion of width) 
on the intermediate results.
Furthermore, it is considered highly unlikely that stronger bounds can be given on the intermediate results through other methods, see~\cite{DBLP:journals/siamcomp/AtseriasGM13,DBLP:journals/jacm/Grohe07}.

A small number of research systems have indeed adopted structural decomposition methods and worst-case optimal join algorithms with highly promising results~\cite{DBLP:journals/tods/AbergerLTNOR17, 
DBLP:journals/corr/abs-2302-13140,%
DBLP:conf/sigmod/IdrisUV17,%
DBLP:journals/vldb/IdrisUVVL20,%
DBLP:conf/sigmod/PerelmanR15,%
DBLP:conf/sigmod/TuR15,%
DBLP:journals/corr/abs-2301-04003,%
DBLP:conf/sigmod/0001022,%
DBLP:conf/sigmod/WangY20,%
DBLP:conf/sigmod/Wang021}.
The work in 
\cite{DBLP:conf/sigmod/IdrisUV17,%
DBLP:journals/vldb/IdrisUVVL20,%
DBLP:conf/sigmod/WangY20,%
DBLP:journals/corr/abs-2301-04003} 
focuses specifically on the problem of dynamic query answering, i.e., updating the answers to a query as modifications are made to the database. 
Important extensions of Yannakakis's algorithm beyond 
pure conjunctive queries are studied for instance 
in~\cite{DBLP:journals/corr/abs-2302-13140,%
DBLP:journals/vldb/IdrisUVVL20,%
DBLP:conf/sigmod/0001022}, where 
acyclic conjunctive queries are extended by 
set difference, theta-joins, and comparisons spanning several
relations, respectively. 
The methods and systems presented 
in~\cite{DBLP:journals/tods/AbergerLTNOR17,%
DBLP:conf/sigmod/PerelmanR15,%
DBLP:conf/sigmod/TuR15,DBLP:conf/sigmod/Wang021}
have focused on queries that are more of graph theoretical interest rather than typical relational database queries. In particular, their performance evaluation focuses on tasks like counting triangles, barbell graph queries (i.e., two disjoint cliques that are connected by a single edge), and the like, where worst-case optimal join techniques are expected to have a larger effect on the performance than structure-guided evaluation. These works motivate the use of theoretical results in practical systems in their own right, but tell us little about the real-world feasibility for typical yet challenging queries. Moreover, it is difficult to judge how performance improvements translate from purpose-built research systems to the established systems that are popular in industry. 
Ultimately, the principal question remains unanswered:
\begin{center}\it
  Can structure-guided evaluation improve real-world performance of standard database systems for typical yet challenging queries?
\end{center}

The goal of this paper is to study this question and bridge this gap 
between the theory and systems communities. Our results show that structure-guided evaluation brings large performance gains on a variety of mainstream DBMSs and can thus alleviate some of the most critical pain-points of modern DBMSs.

Our affirmative answer is primarily based on a broad experimental evaluation 
on a recent benchmark by~\citet{DBLP:conf/sigmod/ManciniKCMA22} that is representative of the typical yet challenging queries we are interested in.
Such an evaluation on mainstream DBMSs has traditionally been difficult due to an apparent mismatch in paradigms between Yannakakis' algorithm, which operates in multiple phases, and the Volcano iterator
 model~\cite{DBLP:conf/icde/GraefeM93} commonly adopted by modern DBMSs. A direct integration of such methods is therefore laborious and shifts the performance question towards a matter of effective implementation and integration, rather than a study of the general viability of the method. Moreover, such an integration would, in the first place,  
 be restricted to a {\em single} DBMS based on a {\em single} architectural type.

We therefore base our experimental evaluation on a DBMS-agnostic rewriting
to control 
a Yannakakis-style evaluation from ``outside'' the DBMS  by 
submitting to the DBMS appropriate SQL statements that correspond to the 
operations performed by Yannakakis' algorithm.
Using these rewritings, we
 compare the performance of a structure-guided approach to the standard query execution strategies in three DBMSs: PostgreSQL, DuckDB, and Spark SQL, that were selected as popular representatives of distinct types of DBMS architecture.

In addition to our empirical results, we also observe that certain common query patterns are particularly amenable to structure-guided evaluation.
 For these queries, even a partial execution of Yannakakis' algorithm 
 is sufficient to answer the query. 
 More specifically, it is possible to leave out the materialisation of any joins
 and to check consistency via semi-joins only. We shall therefore refer to 
 these  queries as 
 \emph{zero-materialisation answerable (0MA)}. 

\paragraph{Contributions}
Our main contributions are as follows.
\begin{itemize}
\item
We develop a flexible methodology for Yanna\-kakis-style query processing 
on top of a DBMS without requiring any modification to the DBMS itself. 
Our rewriting-based approach is 
completely DBMS-agnostic and could, in principle, be applied to any DBMS that 
adheres to the SQL standard. This will facilitate wider experimental investigation of the benefits of structure-guided query evaluation in different settings and systems without the need for deep integration in a DBMS, which -- at this stage -- would require 
an unjustifiably high effort.
\item 
We introduce and formally define the class of 
zero-mate\-rialisation answerable (0MA) queries and show that all 
0MA queries can be evaluated without materialising any joins.
We illustrate through various examples that 
this class indeed covers a wide variety of common query patterns. Moreover, we show how the beneficial properties of 0MA queries can be generalised to allow for highly efficient structure-guided query answering beyond the class of 0MA queries.
\item 
We experimentally verify that structure-guided query processing can indeed provide significant improvements %
for difficult queries.
Our experiments are carried out on 
three concrete, quite different DBMSs: (1) PostgreSQL -- a well-established
row-oriented DBMS, (2) DuckDB -- a recent, in-process, columnar DBMS that implements many modern techniques for query optimisation and execution, and
(3) Spark SQL, which is part of a distributed computing framework for a wide range of 
data analytics tasks.
For all three systems, we show that our rewritings drastically reduce (and in some cases 
even completely eliminate) the number of timeouts on over 300 
challenging queries of 
a recent benchmark 
from~\cite{DBLP:conf/sigmod/ManciniKCMA22} over the 
MusicBrainz dataset~\cite{musicbrainz}.
\end{itemize}

\paragraph{Related Work} 
Yannakakis' algorithm has received a lot of 
attention in the Database Theory community in the context of identifying classes of queries that allow for particularly efficient enumeration of query result, namely linear-time pre-processing and constant delay. 
This line of research was initiated by Bagan et al.~\cite{DBLP:conf/csl/BaganDG07} 
and has triggered a lot of follow-up work 
such as \cite{DBLP:journals/mst/CarmeliK20,DBLP:journals/tods/CarmeliK21,DBLP:conf/pods/CarmeliTGKR21,DBLP:journals/tods/CarmeliZBCKS22,DBLP:conf/icdt/GeckKSS22,DBLP:conf/pods/LutzP22} since then. 

On the Database Systems side, 
a combination of traditional query optimisation with 
Yannakakis-style query evaluation was first presented in \cite{DBLP:conf/icde/GhionnaGGS07}, 
building upon results 
from~\cite{DBLP:conf/pods/ScarcelloGL04}. In \cite{DBLP:conf/icde/GhionnaGGS07}, 
the authors 
present the integration of structure-guided query processing, 
based on hypertree decompositions, into a traditional 
query engine. The integration works via rewriting to not further specified
``nested SQL subqueries''.
The approach, which is tightly coupled with a concrete DBMS,  (namely PostgreSQL) aims at improving the performance on 
queries from the TPC-H benchmark. The performance gain reported in the
paper reaches up to 20\% for two concrete TPC-H queries (queries $Q_5$ and $Q_8$). 
Large join queries or the 
identification of particularly well-suited classes of queries for a structure-guided approach 
were not in the scope of the paper. 

As mentioned above, several successful research
prototypes based 
on Yannakakis-style query evaluation have been presented recently: 
The DunceCap query compiler presented in \cite{DBLP:conf/sigmod/PerelmanR15,DBLP:conf/sigmod/TuR15, DBLP:journals/tods/AbergerLTNOR17} 
combines Yannakakis-style query evaluation 
with worst-case optimal join techniques -- 
primarily targeting small, cyclic graph queries.
Similarly, \citet{DBLP:journals/tods/AbergerLTNOR17} use generalized hypertree decompositions as a form of query plans 
in combination with multi-way joins and further advanced techniques to obtain highly promising results in a graph database setting.
Further applications and extensions of 
Yannakakis' algorithm comprise dynamic query evaluation 
\cite{DBLP:conf/sigmod/IdrisUV17,%
DBLP:journals/vldb/IdrisUVVL20,DBLP:journals/corr/abs-2301-04003}, 
comparisons spanning several 
relations~\cite{DBLP:conf/sigmod/0001022},
queries involving theta-joins~\cite{DBLP:journals/vldb/IdrisUVVL20},
and privacy protecting query 
processing~\cite{DBLP:conf/sigmod/Wang021}.
Finally, we note that Yannakakis-style query evaluation is also well suited for distributed data processing. The theoretical foundation was already laid in \cite{DBLP:journals/jacm/GottlobLS01} by
showing that acyclic conjunctive query evaluation falls into the highly parallelisable complexity 
class LogCFL. 
This result was later generalised to hypertree decompositions in \cite{DBLP:journals/jcss/GottlobLS02}.
An actual prototype system implementing Yannakakis' algorithm in MapReduce was 
presented in \cite{DBLP:conf/icdt/AfratiJRSU17}.

To conclude, there are many theoretical studies and also concrete
implementations that underline the great potential of Yannakakis-style 
query evaluation. However, all these successful implementations were 
achieved by standalone research systems. None of them  
studied the viability of Yannakakis' algorithm in standard DBMSs.  
Concurrently and independently of this work, 
\citet{DBLP:journals/corr/abs-2302-13140}
actually did query rewriting on top of several standard DBMSs (including
PostgreSQL, DuckDB and Spark SQL) in case
of acyclic conjunctive queries. However, their work specifically aims at 
 efficient computation of the {\em difference}
 between (small) conjunctive queries   
 rather than the basic query evaluation (in particular, of large queries) 
 considered here.

\paragraph{Structure of the paper}
In Section~\ref{sect:Preliminaries}, we recall some basic definitions and results relevant to 
our work. In Section \ref{sect:theory},
we introduce the class of {\em zero-materialisation answerable} queries 
(0MA queries, for short), 
which can be evaluated by semi-joins only. More general queries will be 
discussed in Section \ref{sect:generalqueries}.
The general idea of our rewriting-based approach for 
combining structure-guided query evaluation with traditional DBMS technology 
and the experimental results thus obtained 
are presented in Section~\ref{sec:expeval}. 
A conclusion and a detailed discussion of directions for future work are given in Section~\ref{sec:conclusion}.
\ifArxivVersion

In the appendix, we provide further details on various aspects of our experiments. 
More specifically, in Appendix \ref{sect:Implementation},
we provide details of the system designed for our experiments. Further experimental results
(in particular, on memory and communication cost) are given 
in Appendix~\ref{sect:spark:appendix}. And, finally, some more details on 
cyclic queries are presented in Appendix~\ref{sec:cycles}. 
Moreover, the whole collection of results produced in our experiments (query rewritings, hypergraphs, output data, performance measurements) as well as instructions for reproducing our experiments are available on 
Figshare:~\figARXIV.
\fi

\section{Preliminaries}
\label{sect:Preliminaries}

\noindent
\paragraph{Conjunctive Queries}
The basic form of queries studied here are Conjunctive Queries (CQs). 
We will later introduce also more general forms of queries. 
It is convenient to 
consider CQs as Relational Algebra expressions of the form 
$\pi_U (R_1 \bowtie  \dots \bowtie R_n ),$
where $R_1, \dots, R_n$ are pairwise distinct relations and 
the projection list $U$ consists of attributes occurring in the $R_i$'s.
This restriction of CQs is without loss of generality. 
Indeed, by applying appropriate renaming operations, 
we may always assume that the $R_i$'s are pairwise distinct 
and that equi-joins are replaced by natural joins. 
Moreover, we do not need to make selections explicit since equality conditions on attributes from 
different relations are taken care of by the natural joins and 
equality conditions on attributes of a single relation
can be pushed immediately in front of the corresponding relation
and carried out when the relation is first accessed. By slight abuse of notation, 
we shall use the same symbol $R_i$ to refer also to the relational schema 
(i.e., the set of attributes) of 
a relation $R_i$.

\paragraph{Acyclicity}
Several forms of acyclicity have been studied in the
literature \cite{DBLP:journals/jacm/Fagin83,DBLP:journals/csur/Brault-Baron16}. 
Our notion of acyclicity is the so-called $\alpha$-acyclicity. 
An {\em acyclic conjunctive query} (an ACQ, for short) is a 
CQ 
$Q = \pi_U (R_1 \bowtie  \dots \bowtie R_n )$ 
that has a {\em join tree}, i.e., 
a rooted, 
labelled tree $\langle T,r,\lambda\rangle$ with root $r$, 
such that 
(1) $\lambda$ is a bijection that assigns to each node of $T$ one of the relations
    in $\{R_1,   \dots,  R_n\}$ and
(2) $\lambda$  satisfies the so-called {\em connectedness condition\/}, i.e., 
if some attribute $A$ occurs in both relations $\lambda(u_i)$ and $\lambda(u_j)$
for two nodes $u_i$ and $u_j$, then 
$A$ occurs in the  relation $\lambda(u)$ for every node $u$ along the path between $u_i$ and $u_j$.
Deciding if a CQ is acyclic and, in the positive case, constructing a join tree 
can be done very efficiently by the GYO-algorithm
(named after the authors 
of  
\cite{report/toronto/Gra79,DBLP:conf/compsac/YuO79}).

It is convenient to introduce the following notation: for a node $u$ of $T$, 
we write $T_u$ to denote the subtree of $T$ rooted at $u$. 
Moreover, for every node $u$ of $T$ and every subtree $T'$ of $T$, 
we write $\Att(u)$ for the attributes of the relation $\lambda(u)$ 
and $\Att(T')$  for the attributes occurring in any of the relations $\lambda(u')$ 
for all nodes $u'$ in  $T'$. 

\paragraph{Yannakakis' algorithm.}
In  \cite{DBLP:conf/vldb/Yannakakis81},
Yannakakis showed that ACQs can be evaluated 
in time $O( (||D|| + ||Q(D)||) \cdot ||Q||)$, i.e., 
linear w.r.t.\ the size of the input and output data and 
w.r.t.\ the size of the query. 
This bound applies to both, set and bag semantics.
Let us ignore the projection for a while and consider 
an ACQ $Q$ of the form  $R_1 \bowtie \dots \bowtie R_n$ with join tree 
$\langle T,r, \lambda \rangle$. 
Yannakakis' algorithm  (no matter whether 
we consider set or bag semantics) 
consists of a preparatory step followed by 
3 traversals~of~$T$: 

In the {\em preparatory step} (also referred to as {\em Setup stage} in the sequel), 
we associate with each node $u$ in the join tree $T$ the relation 
$\lambda(u)$. If the CQ originally contained selection conditions on attributes of 
relation $\lambda(u)$, then we can now apply this selection. 
The 3 traversals of $T$ consist of 
(1) a bottom-up traversal of semi-joins, (2) a top-down traversal of semi-joins,
and (3) a bottom-up traversal of joins. 
Formally, let $u$ be a node in $T$ with  
child nodes $u_1, \dots, u_k$  of $u$ and let 
relations $R$, $R_{i_1}, \dots, R_{i_k}$ be associated with 
the nodes  $u$, $u_1, \dots, u_k$ at some stage of the computation. 
Then we set 
(1) $R = (((R  \ltimes R_{i_1}) \ltimes R_{i_2}) \dots)  \ltimes  R_{i_k}$,
(2) $R_{i_j} = R_{i_j}  \ltimes R$ for every $j \in \{1, \dots, k \}$, and 
(3) $R = (((R  \bowtie R_{i_1}) \bowtie R_{i_2}) \dots)  \bowtie  R_{i_k}$ in the 3 traversals (1), (2), and (3), respectively.
The final result of the query is the resulting relation associated with 
the root node $r$ of $T$.
\nop{********************************
If this relation is empty after the first bottom-up traversal, then so is the final result 
and the traversals (2) and (3) can be omitted.

If the first bottom-up traversal produces the empty relation 
at the root node $r$ of $T$, then 
so is the final result,
and the traversals (2) and (3) can be omitted.

In words, the first bottom-up traversal restricts the relation associated with each node $u$ in $T$ 
to those tuples that join with the relations at the nodes in the subtree below $u$. 
The top-down traversal further restricts the relation at each node $u$ to those tuples that 
join with {\em all} relations of $Q$, i.e., these tuples can be extended to answer tuples of $Q$. 
The second bottom-up traversal computes for each node $u$ the join of all relations in the subtree rooted at $u$.
********************************}

We can now easily integrate the  projection $\pi_U$ into
this algorithm by projecting out in the second bottom-up traversal all attributes that neither occur
in $U$ nor further up in $T$. Of course, attributes neither occurring in $U$ nor in any join condition can 
already be projected out as part of the 
preparatory  step.

\paragraph{Beyond Conjunctive Queries}
For queries beyond CQs, we shall use the (extended) Relational 
Algebra notation from \cite{DBLP:books/daglib/0010423}. 
To be consistent with SQL, 
we consider Relational Algebra with bag semantics
throughout this paper.
In addition to the operators $\pi$ (projection), $\sigma$ (selection), 
and $\bowtie$ (join), we also allow 
$\gamma_U$ (group-by) and $\delta$
(duplicate elimination). In case of the group-by operator, the subscript $U$ 
is a list of attributes $A$ and aggregate expressions of the form $g(A)$  with 
$g \in \{$\,MIN, MAX, COUNT, SUM, AVG\,$\}$. 
Moreover, 
we tacitly assume
that relations and attributes may be renamed.

\section{0MA  Queries}
\label{sect:theory}

It is well known \cite{DBLP:journals/jacm/GottlobLS01} that for Boolean ACQs (i.e., queries where we are only interested if the answer is non-empty), 
Yannakakis' algorithm can be stopped after the first bottom-up traversal. 
Indeed, if at that stage the relation associated with the root node of the join tree
is non-empty, then so is the query result. Most importantly, for queries of this type, the most expensive part of the evaluation (i.e., the joins in the second bottom-up traversal) can be completely omitted. %
The next example illustrates that such favourable behaviour is by no
means restricted to Boolean queries.

\begin{example}
  \label{ex:sca}
  Consider an excerpt of a university schema with 
  relations 
  $\mathsf{exams}(\mathrm{cid}$, $\mathrm{student}$, $\mathrm{grade})$
  and
  $\mathsf{courses}(\mathrm{cid}, \mathrm{faculty})$. Querying each student's lowest grade in courses of the Biology faculty is naturally stated in SQL as follows.
  {\small
  \begin{lstlisting}[language=SQL]
    SELECT   exams.student, MIN(exams.grade) 
    FROM     exams,courses
    WHERE    exams.cid=courses.cid
             AND courses.faculty='Biology'
    GROUP BY exams.student;
  \end{lstlisting}
  }

\noindent
Ignoring the GROUP BY clause for a while, the query 
involving only two relations is
trivially acyclic. In the join tree consisting of 2 nodes, 
we choose as root the node labelled by the \textsf{exams}-relation. 
After the first bottom-up traversal, this relation contains 
all \textsf{exams}-tuples that join with the \textsf{courses}-relation restricted to those tuples with faculty = 'Biology'. Hence, if we now also take the GROUP BY clause 
into account, answering the query is 
possible by only looking at the \textsf{exams}-relation 
-- without the need for the remaining two traversals of Yannakakis' algorithm. 
\hfill $\diamond$
\end{example}

In this section,  we want to identify a whole family of 
queries whose evaluation only requires the first bottom-up traversal 
of Yannakakis' algorithm. 
To this end, we introduce the class of \emph{\queryname (\qnameshort)} queries and we will illustrate the usefulness of this class 
by various examples. In particular, Boolean queries and the 
query from Example~\ref{ex:sca} are contained in this class. 
The performance gain attainable when answering \qnameshort queries
will be demonstrated experimentally in Section~\ref{sec:expeval}.

\begin{definition}
\label{def:SCA}
\mbox{}
\begin{itemize}
    \item A query $Q$ is in \emph{aggregation normal form}\footnote{Note that the $\gamma$ operator also implicitly projects to some subset of attributes. The projection $\pi_S$ is thus not strictly necessary and is only added for clarity.}  if it is of the form 
    $\gamma_U(\pi_S(Q'))$, where   
    $Q'$ is a query consisting only of natural joins and selection. 
    \item For a query $Q$ in aggregation normal form, we say that $Q$ is \emph{guarded},
    if $Q'$ mentions a relation $R$ with 
    $\Att(S) \subseteq \Att(R)$, i.e., $R$ contains all attributes 
    occurring in the GROUP BY clause (aggregate or not). If this is the case, we say that $R$ \emph{guards} 
    query $Q$ or, equivalently, 
    $R$ is a \emph{guard} of $Q$.
    \item We say that a query $Q = \gamma_U(\pi_S(Q'))$ is \emph{set-safe} if it is equivalent to $\gamma_U(\delta(\pi_S(Q')))$, i.e., duplicate elimination before the GROUP BY does not change the 
    meaning of the query.
    \item A query $Q$ in aggregation normal form is called 
    \emph{zero-mate\-rialisation answerable} (\emph{0MA}) if it is guarded and set-safe.
\end{itemize}
\end{definition}

\noindent
As far as the notation is concerned, 
recall from Section~\ref{sect:Preliminaries}
that the restriction to natural joins and top-level projection is without loss of generality and it only serves to simplify 
the notation. This is also the case in the above
definition. In our examples, we may 
freely lift this restriction if it is convenient. 
In contrast to the restricted notation of CQs 
in Section~\ref{sect:Preliminaries}, we now prefer to make selection explicit in the 
``inner'' query $Q'$ -- 
in addition to the natural joins. 

Despite the various technical constraints, 
\qnameshort queries still cover many common query patterns. 
Clearly, the restriction to aggregation normal form 
matches the standard use of aggregates
in SELECT-FROM-WHERE-GROUP BY statements in SQL. 
Also the further restrictions imposed by 0MA queries 
are met by many common query patterns observed in practice. 
Boolean ACQs mentioned above (e.g., realised by a query of the form  SELECT 1 FROM ....) 
are a special case of 0MA queries, where we simply leave out the grouping, and the projection is to the empty set of attributes.

We next verify that also the query from Example~\ref{ex:sca}
is 0MA. By slightly simplifying the subscripts
(in particular, abbreviating attribute names), the query  
translates to the following Relational Algebra query $Q$:
  \[
    \gamma_{\mathrm{stud}, \mathrm{MIN(grad)}}(\pi_{\mathrm{stud},\mathrm{grad}}
    (\mathsf{exams} \bowtie \sigma_{\mathrm{faculty}=\mathtt{'Biology'}}(\mathsf{courses})))
  \]
Clearly, query $Q$ is \queryname, since relation $\mathsf{exams}$ (containing both attributes 
$\mathrm{student}$ 
and 
$\mathrm{grade}$) is a guard of $Q$ and  aggregation via MIN (or MAX) is always set-safe.

We now formally prove that acyclic 0MA queries may 
indeed be evaluated without the 
join-phase of Yannakakis' algorithm. That is,
these queries can be evaluated 
via aggregate/group processing over a single relation of the database that has been reduced by the semi-joins of the first bottom-up traversal.
\nop{***************
This avoidance of all joins can lead to significant improvements in situations where large join results are unavoidable or finding a good join ordering is difficult by current query planning methods. 
***************}

  \begin{theorem}
    \label{thm:sca}
  Let $Q=\gamma_{U}(\pi_S(Q'))$ be a \qnameshort query in aggregation normal form  such that $Q'$ is an ACQ, and let $D$ be an arbitrary database. 
  Let $\langle T,r, \lambda \rangle$ be a  join tree of $Q'$ such that the root $r$ of $T$ is
  labelled by relation $R$ that guards $Q$. Let $R'$ be the relation associated with node $r$ 
  after the first bottom-up traversal of Yannakakis' algorithm. Then
  the equality $Q(D) = \gamma_{U}(\delta(\pi_S(R')))$ holds.
\end{theorem}

\begin{proof}[Proof Sketch]
After the first bottom-up traversal, all tuples 
in a relation associated with a node in $T$ actually join with all relations in the subtree below. 
Since $R'$ is the relation at the root, every tuple $r \in R'$ extends to a result of $Q'$. 
Since $Q$ is guarded, we have that $S$ is a subset of attributes in $R'$ and thus 
$\pi_{S}(Q'(D)) \supseteq \delta(\pi_S(R'))$ and, therefore,
also 
$\delta(\pi_S(Q'(D))) \supseteq \delta(\pi_S(R'))$. 

On the other hand, since $R$ is part of $Q'$, which consists only of natural joins and selection, any tuple in $Q'(D)$ must be consistent with $R$. Since every tuple in $Q'(D)$ must also be consistent with all other relations mentioned in $Q'$, it must also be consistent with $R'$ and,
therefore,  
$\delta(\pi_S(Q'(D))) \subseteq \delta(\pi_S(R'))$ holds. 
Moreover, as $Q$ is set-safe, we also
have 
$\gamma_{U}(\pi_S(Q'(D)))
=
\gamma_{U}(\delta(\pi_S(Q'(D))))$ and, hence, 
$\gamma_{U}(\delta(\pi_S(Q'(D))))=\gamma_{U}(\delta(\pi_S(R')))$.
\end{proof}

Note that the requirement in Theorem~\ref{thm:sca} 
that the guard $R$ must be the label of the root node of a join tree of $Q'$ does not impose any additional restrictions apart from the conditions that $R$ must be a guard and $Q'$ must
be an ACQ. Some node in the join is guaranteed to be labelled by $R$,
and we can always choose this 
particular node as the root of the join tree.

It is important to note that 
set-safety is not required
due to any technical issues with bag semantics. The restriction to set-safety 
is only needed to identify queries whose answer can be determined without
knowing the 
exact multiplicity of a tuple in the answer.
As far as standard aggregate functions are concerned, this
always holds for MAX and MIN
as we have seen in Example~\ref{ex:sca}.
In contrast, other standard aggregates, such as SUM or COUNT are, in general, not set-safe.
They nevertheless can be answered efficiently when 
knowing the multiplicity of each tuple in the result of the join query $Q'$. In cases where the query is guarded, 
it may indeed be possible to compute these multiplicities without materialising joins by adapting dynamic programming algorithms for counting homomorphisms, see e.g.,~\cite{DBLP:journals/siamcomp/FlumG04,DBLP:journals/jcss/PichlerS13}. 
Moreover, these aggregates may actually be used in patterns that are set-safe, e.g., COUNT(DISTINCT~$\dots$) constructs in SQL. Indeed, it is easy to see that the combination with 
DISTINCT can make any aggregate set-safe. Furthermore, trivial use of $\gamma$ as projection (all attributes are grouping attributes) also covers the enumeration of distinct tuples as a set-safe operation.
In practice, even more cases may be set-safe due to constraints on the data such as, for instance, counting the different values of an attribute with a UNIQUE constraint. In light of the promising 
results for \qnameshort queries in Section \ref{sec:expeval}, we consider the study of ways to relax this restriction -- 
e.g., 
in the presence of common database constraints
as in Example~\ref{ex:tpch:one}
-- as a worthwhile area for future work.

\paragraph{Deciding the 0MA Property}
It is natural to consider the question of deciding whether a given query is 0MA. We give a brief informal discussion of why this is not of particular interest in our case. First, it is clear that deciding whether a query is guarded is trivial, and only deciding the set-safe condition is of any real concern. 
For the inner query $Q'$, the set-safety status boils down to the question if it can return duplicate results or not, which is well understood and easy to check (recall that we are restricting ourselves to CQs with some extensions 
such as GROUP BY, HAVING, aggregates; so undecidability results for FO queries such as non-emptiness do not apply here):
if multiset input relations are allowed, then every query may possibly return duplicate tuples; otherwise any query where some attribute is projected out can return duplicates.

Consequently, only the semantics of the aggregate functions themselves are the important factor for set-safety.
In general, set-safety is a non-trivial property of the aggregate functions and thus expected to be undecidable if we allow {\em arbitrary} 
computable functions as
aggregates. 
However, we are interested in the concrete behaviour of 
current DBMSs, which typically only offer a small fixed vocabulary of aggregation functions. For instance, the ANSI SQL standard specifies 28 possible aggregation functions and they are easy to check for set-safety case by case, 
without the need for a general procedure to check set-safety of arbitrary 
functions. 
\section{More General Queries}
\label{sect:generalqueries}
In this section, 
we inspect several situations in which we are not dealing with 
acyclic CQs and/or not zero-ma\-teria\-lisation answerable
queries, and where the performance gain achieved by a short-cut
in Yannakakis' algorithm is nevertheless attainable.

Recall that we have omitted a HAVING clause from our 
aggre\-gation normal form in Definition \ref{def:SCA}. 
If we have a 0MA query
with a HAVING clause on top of it
(see, e.g., Example~\ref{ex:tpch:one} below),
then we can still evaluate the 0MA query without materialising 
any joins and simply filter the result by the HAVING condition  afterwards.

More generally, the optimisation from 
Theorem~\ref{thm:sca} is applicable whenever 
some part of a query satisfies the 0MA condition. 
For instance, subqueries with the EXISTS operator are actually Boolean queries and, as such,
0MA -- provided that they are ACQs.

\noindent
The following example involving a 0MA subquery 
is taken from the TPC-H benchmark: 

\begin{example}
  \label{ex:tpch:two}
TPC-H Query 2 contains the following subquery:
  {\small
    \begin{lstlisting}[language=SQL]
      SELECT MIN(ps_supplycost) 
      FROM partsupp, supplier, nation, region
      WHERE p_partkey = ps_partkey AND ...
    \end{lstlisting}
  }
  
\noindent  
where \texttt{p\_partkey} is an attribute coming from the outer query and 
the rest of the 
WHERE clause are equi-joins and selections. 
This subquery is a standard example of a 0MA query, since
aggregation by MIN is always set-safe and the query is clearly guarded by \texttt{partsupp}.
The subquery is correlated inside  the TPC-H Query 2
due to the attribute  \texttt{p\_partkey} from the outer query, 
but it allows for effective decorrelation.
Notably, 
if we consider magic decorrelation~\cite{DBLP:conf/icde/SeshadriPL96}, then 
we would change the select clause of the subquery to \texttt{ps\_partkey, min(ps\_supplycost)}, add a grouping over \texttt{ps\_partkey}, and remove the correlated join with \texttt{p\_partkey}. This transformation preserves guardedness and set-safety 
and we could, in this case, 
combine decorrelation with the efficient evaluation of the decorrelated 0MA subquery according to Theorem~\ref{thm:sca}.
\hfill $\diamond$
\end{example}

Below we see a more complex TPC-H query, 
where optimised evaluation based on 0MA-parts is even possible twice -- once for the subquery and once
for the outer query.

\begin{example}
 \label{ex:tpch:one}
  TPC-H Query 11 is of the following form
  {\small
  \begin{lstlisting}[language=SQL]
    SELECT ps_partkey, 
           SUM(ps_supplycost*ps_availqty)
    FROM partsupp, supplier, nation
    WHERE ps_suppkey = s_suppkey 
      AND s_nationkey = n_nationkey
      AND n_name = 'GERMANY'
    GROUP BY ps_partkey
    HAVING SUM(ps_supplycost*ps_availqty) >
      (SELECT SUM(ps_supplycost*ps_availqty) 
                  * 0.0001
       FROM ...)
\end{lstlisting}
}

\noindent
where the omitted FROM clause of the subquery 
is the same as the FROM clause of the outer query. 
That is, the subquery is almost the same as the 
outer query: we just leave out the grouping by 
$\mathtt{ps\_partkey}$, and the sum over 
$\mathtt{ps\_supplycost*ps\_availqty}$
is now taken over all 
$\mathtt{ps\_partkey}$'s 
and is multiplied by 0.0001.

At its core, this SQL query can be evaluated via a 0MA query of the
form $Q=\gamma_{U}(\pi_S(Q'))$, where $Q'$ represents the
join query on the three relations, and
$$U =\mathtt{ps\_partkey},\mathtt{ps\_suppkey},\mathtt{ps\_supplycost*ps\_availqty}.$$ Note that keeping \texttt{ps\_suppkey} in the grouping at this step is important to observe that the essence of this query is set-safe.
The result of both, the outer query and the subquery in the HAVING clause, can be directly obtained from $Q$, 
leaving only a final filtering step. 

We  now analyse why $Q$ is 0MA.  While \texttt{partsupp} clearly guards the
query, observing set-safety requires a small but natural step beyond
the technical definition above. In TPC-H, there are constraints on the
database that require that \texttt{s\_suppkey} and \texttt{n\_nationkey}
be keys of \texttt{supplier} and \texttt{nation},
respectively. Therefore, every tuple in \texttt{partsupp} can have
only one join partner in \texttt{supplier}, and the result has only
one join partner in \texttt{nation}. Furthermore, the projection on
$Q'$ retains the key (\texttt{ps\_partkey}, \texttt{ps\_suppkey}) letting us observe overall
that every tuple in the result of
$Q'$ is in fact distinct and, as a consequence, $Q$ is also
set-safe.
\hfill $\diamond$
\end{example}

We conclude this section by briefly discussing CQs to which 
Theorem~\ref{thm:sca} is not applicable. That is, either acyclicity or 
the 0MA property is violated. In case of cyclic queries, we may apply 
decomposition methods~\cite{DBLP:journals/jcss/GottlobLS02,2014grohemarx} to turn 
a given CQ into an acyclic one. Since CQs in practice tend to be 
acyclic or almost acyclic~\cite{DBLP:journals/vldb/BonifatiMT20,DBLP:journals/jea/FischlGLP21}, this transformation into an ACQ is 
feasible at the expense of a polynomial blow-up (where the degree of 
the polynomial is bounded by the corresponding notion of width). 
Actually, also the queries from the benchmark 
of~\cite{DBLP:conf/sigmod/ManciniKCMA22}, which we use for our 
experimental evaluation, follow this pattern: 
the vast majority of the queries is acyclic and the rest have
low generalized hypertree width (ghw). 
First preliminary 
experiments with queries of low ghw  (see Section \ref{sec:expeval}) 
suggest that 
the extension 
of structure-guided evaluation 
to cyclic queries is a worthwhile target for 
future research. 

If the 0MA property (in particular, the guardedness) is violated,
then the materialisation of {\em some} joins is usually unavoidable. 
However, this does not mean that {\em all} joins have to be materialised. 
Instead, 
for a query of the form $\gamma_U(\pi_S(Q'))$,
the joins in the second bottom-up traversal of Yannakakis' algorithm 
(and also the semi-joins in the top-down traversal) 
can still be
restricted to a subtree whose relations contain all the attributes in $U$.

\section{Experimental Evaluation}
\label{sec:expeval}

In this section, we  detail the results of our experiments, which demonstrate that 
structure-guided query evaluation can indeed greatly improve performance on challenging join queries.

\subsection{Methodology}
Our goal is to shed light on the benefit of realizing structure-guided query evaluation by common database systems. 
We thus do not want to restrict ourselves to a single system nor to a single 
architecture or a single query planning and execution strategy. We have therefore chosen three DBMSs based on 
different technologies: 
PostgreSQL 13.4~\cite{DBLP:journals/cacm/StonebrakerK91} as a ``classical'' row-oriented relational DBMS, 
DuckDB 0.4~\cite{DBLP:conf/sigmod/RaasveldtM19} as
a column-oriented, embedded database, and 
Spark SQL~3.3~\cite{DBLP:journals/cacm/ZahariaXWDADMRV16} 
as a database engine specifically designed for distributed data processing in a cluster.  

We have implemented a proof-of-concept system, referred to as \ourSystem in the sequel, 
that works by rewriting a query into a sequence of SQL statements which 
express %
Yannakakis' algorithm. 
This makes our approach easily portable and we can apply it to the three chosen DBMSs 
with almost no change to our rewriting method (apart from some minor differences in SQL syntax). 
The huge effort of a  full integration into any of the three systems (let alone, into all of them) 
does not seem to be justified before gathering further information on the potential benefit
of such an integration. 
Moreover, our rewriting-based approach is also applicable to commercial DBMSs, 
where large internal modifications without convincing justification 
are inconceivable.
In our experiments, we compare the performance of join queries in each DBMS 
with the performance of the \ourSystem rewriting, executed by the same system.

The \ourSystem system proceeds in several steps: we first extract the CQ from the given SQL query  and transform it into a hypergraph. From this we compute a join tree by applying a variant of the 
GYO-algorithm~\cite{report/toronto/Gra79,DBLP:conf/compsac/YuO79}. We then generate the SQL statements that correspond to the 
semi-joins and joins of \YA. 
These SQL statements involve the creation of a couple of temporary tables. 
If the original query contains GROUP BY and HAVING clauses or more general selections (beyond equalities), then these can be integrated into the SQL-statement referring to the root node in the final traversal of the join tree. 
The whole rewriting is rather straightforward. 
Further details on the implementation of
\ourSystem 
are provided 
in Appendix~\ref{sect:Implementation}.

\subsection{Experimental Setup}
We perform experiments using a recent benchmark by 
Mancini et al.~\cite{DBLP:conf/sigmod/ManciniKCMA22}, which consists of 435 challenging synthetic join queries over the MusicBrainz dataset~\cite{musicbrainz}. Classic benchmark datasets, such as TPC-H or TPC-DS, are less interesting for our purposes since their focus is not on the complexity of evaluating queries with a large number of potentially expensive joins. The join-order-benchmark (JOB) 
\cite{DBLP:journals/vldb/LeisRGMBKN18}
focuses on the effectiveness of cardinality estimations to produce optimal query plans, but even the worst query plans still require only a single digit
number of minutes for query evaluation on standard systems. In contrast, the benchmark from~\cite{DBLP:conf/sigmod/ManciniKCMA22} that we consider here contains queries with as many as 30 relations and, in many cases, the join processing (as well as planning, see~\cite{DBLP:conf/sigmod/ManciniKCMA22}) is very challenging for modern DBMSs. The queries in this benchmark were created over the MusicBrainz dataset~\cite{musicbrainz} by randomly joining tables along foreign key relationships. This makes the generated queries similar to real-world queries and particularly interesting for our experiments since one would normally expect classical DBMSs to perform particularly well on this kind of queries. Moreover, the large number of generated queries protects against the very significant variance in the evaluation of large queries. Another important reason for choosing the queries from this benchmark is that they 
operate on a {\em publicly available}
dataset, which makes our results fully reproducible. 
This is in sharp contrast to big join queries mentioned in 
other works such as 
\cite{DBLP:journals/pvldb/DieuDFLS09,DBLP:conf/sigmod/NeumannR18}.

We will report on two types of experiments. One set of tests will be referred to as \emph{full enumeration} queries. 
For these, we essentially use the original queries of~\cite{DBLP:conf/sigmod/ManciniKCMA22}. 
However, since these queries contain no projection, we adapt the queries to project to only the join attributes (one attribute per equi-join equivalence class, i.e., no redundant columns)
in order to lessen the role of unimportant I/O. 
In a second set of experiments, we explore the effectiveness of computing aggregate queries with the 0MA property from
Definition~\ref{def:SCA}. For this purpose, we transform each query to compute a "MIN" aggregate for an attribute that we randomly choose from those attributes that already occur in the original query. In the following, we refer to these aggregation variants as the \emph{0MA aggregation} queries. In both cases, the queries are always executed on the standard MusicBrainz dataset.  For all experiments in this section, we use a timeout of 20 minutes for the execution of each query. The experiments on DuckDB and PostgreSQL are performed on a machine with an Intel Xeon Bronze 3104 with 6 cores clocked at 1.7 GHz, and 128 GB of RAM and running Debian 11, using the Linux kernel 5.10.0 with all data stored on an SSD. The default settings of PostgreSQL proved unsatisfactory in our system environment. We therefore explicitly configured PostgreSQL to use at most 8 concurrent working threads and 200 concurrent I/O requests, which turned out to be the most suitable configuration
for our system. For DuckDB, we use all default parameters (leading to full utilisation of all cores and concurrent disk I/O). Our experiments with Spark SQL are performed in a cluster environment with two namenodes and 18 datanodes, with each node having two  XeonE5-2650v4 CPUs with 24 cores (48 per node) and 256 GB RAM. %

\ifArxivVersion
In addition to reporting the results in this section, we also provide all raw data of our experiments and instructions for reproducing them on 
Figshare~\figARXIV.
We include there only the rewritten queries, as were produced by \ourSystem{} and detailed logs of their execution. We omit the original queries from~\cite{DBLP:conf/sigmod/ManciniKCMA22} and we hope to make the full data publicly available in the future.
\fi
\setlength{\tabcolsep}{2pt}
\newcommand{\p}[0]{\phantom{0}}

\begin{table}[tb]
	\caption{Overview of the queries from~\cite{DBLP:conf/sigmod/ManciniKCMA22}.}
	\label{tab:musicbrainzQueries}
	\centering
    \vspace{-2mm}
	\begin{tabular}{l|c}
	\toprule
	{Statistic} & {Count} \\
	\midrule
	 Number of queries	 & 435 \\
	 Number of joins & Between 1 and 29 \\ 
	 Number of involved tables & Between 2 and 30 \\ 
	 Number of cyclic queries & 84 \\
	\bottomrule
	\end{tabular}
\end{table}

\begin{table}[tb]
	\caption{DuckDB, PostgreSQL, and Spark SQL with or without
	\ourSystem for ACQs over the MusicBrainz dataset. All times are reported in seconds.}
	\small
	\label{tab:comparisonNoCutoff}
	\centering

	\textbf{0MA Aggregation Queries}  \\

	\begin{threeparttable}
	\begin{tabular}{l|c|rrrr}
	\toprule
	{Method} & {Timeouts} & {Max\tnote{1}} & {Mean\tnote{2}} & {Med.\tnote{2}} & Std.Dev.\tnote{2} \\
	\midrule

DuckDB&
58&1169.38&217.9\p&0.44&447.94 \\
DuckDB+\ourSystem&
\textbf{\p0}&15.57&2.31&1.44&2.38\\

	\midrule
PostgreSQL& 
91&1131.08&342.78&2.82&524.16\\
PostgreSQL+\ourSystem& 
\textbf{\p2}&236.75&24.74&5.83&93.73 \\

	\midrule
SparkSQL& 
91&1082.58&365.76&25.35&518.7\p\\
SparkSQL+\ourSystem& 
\textbf{\p3}&214.04&41.12&16.14&113.24 \\

	\bottomrule
	\end{tabular}
	\medskip
		\centering
	\textbf{Full Enumeration Queries}\\
	\begin{tabular}{l|c|rrrr}
	\toprule
	{Method} & {Timeouts} & {Max\tnote{1}} & {Mean\tnote{2}} & {Med.\tnote{2}} & Std.Dev.\tnote{2} \\
	\midrule

DuckDB &
69&770.55&252.27&0.67&473.87\\
DuckDB+\ourSystem& 
\textbf{29}&801.79&121.39&2.34&335.75\\
\midrule
PostgreSQL&
97&1107.66&364.32&4.02&533.47 \\
PostgreSQL+\ourSystem&
\textbf{70}&786.31&283.2\p&25.71&470.2\p\\
\midrule
SparkSQL& 
87&1164.06&358.28&23.91&513.67 \\
SparkSQL+\ourSystem& 
\textbf{29}&876.74&204.11&59.45&335.47\\

	\bottomrule
	\end{tabular}
	\begin{tablenotes}
\item[1] Excludes timeout values. \item[2] Timeout treated as 1200 seconds.
\end{tablenotes}
\end{threeparttable}
\vspace{-5mm}
\end{table}

\nop{**************************

\begin{table}[tb]
	\caption{Showing the same data as in Table~\ref{tab:comparisonNoCutoff}, but only considering queries which take more than 0.5  seconds in the baseline case (i.e. when not using \ourSystem{})}
	\small
	\label{tab:comparisonWithCutoff}
	\centering
	\vspace{-3mm}
	\textbf{Full Enumeration Queries}\\
	\begin{tabular}{l|rrrrr}
	\toprule
	{Method} &  {Max  (s)} & {Mean  (s)} &  {Med. (s)} & Std.Dev. (s) & {Timeouts}  \\
	\midrule
          DuckDB&770.55&46.14&2.85&111.22&69 \\
DuckDB+\ourSystem&801.79&45.35&5.45&126.59&29\\
        \midrule

          PostgreSQL&1107.66&79.53&6.15&202.21&96\\ 
PostgreSQL+\ourSystem&786.31&77.55&26.22&143.32&69 \\

	\bottomrule
	\end{tabular}

	\vspace{2mm}
	\textbf{0MA Aggregation Queries}  \\
	\begin{tabular}{l|rrrrr}
	\toprule
	{Method} &  {Max  (s)} & {Mean  (s)} &  {Med. (s)} & Std.Dev. (s) & {Timeouts}  \\
	\midrule

DuckDB&1192.98&60.12&2.79&168.92&59\\
DuckDB+\ourSystem&14.85&3.56&2.92&2.80&0 \\

	\midrule

PostgreSQL&1131.08&73.95&4.08&196.94&91 \\
PostgreSQL+\ourSystem&236.75&22.27&11.11&30.18&2 \\

	\bottomrule
	\end{tabular}
\end{table}
}

\usetikzlibrary{calc}
\begin{figure}[tb]

\begin{adjustwidth*}{}{-0.5em}

\begin{tikzpicture} [transform shape, scale=0.7]
\begin{axis}[
        xbar,  
        every tick label/.append style={scale=2},
        xmax=230,
        xmin=0,
        width=6cm,
        height=4cm,
        xlabel={DuckDB, Full Enum.},
        symbolic y coords={1,10,100,1000,TO},
        ytick=data,         
        nodes near coords  align={vertical},
        nodes near coords = \rotatebox{0}{{\pgfmathprintnumber[fixed zerofill, precision=0]
        {\pgfplotspointmeta}}},
        point meta=rawx,
    visualization depends on={meta < 25 \as \valueissmall},
    every node near coord/.append style={
            anchor={\ifdim\valueissmall  pt = 1 pt west\else east\fi},
            color={\ifdim\valueissmall  pt = 1 pt black\else black\fi}
    }
    ]
    \addplot + [black,    fill=DuckDB]  coordinates {
                            (189,1)
                            (58,10)
                            (17,100)
                            (18,1000)
                            (69,TO)
    };

\end{axis}
\end{tikzpicture}\begin{tikzpicture} [transform shape, scale=0.7]
\begin{axis}[
        xbar,  
        every tick label/.append style={scale=2},
        xmax=230,
        xmin=0,
        width=6cm,
        height=4cm,
        xlabel={DuckDB + \ourSystem{},  Full Enum.},
        symbolic y coords={1,10,100,1000,TO},
        ytick=data,         
        nodes near coords  align={vertical},
        nodes near coords = \rotatebox{0}{{\pgfmathprintnumber[fixed zerofill, precision=0]
        {\pgfplotspointmeta}}},
        point meta=rawx,
    visualization depends on={meta < 25 \as \valueissmall},
    every node near coord/.append style={
            anchor={\ifdim\valueissmall  pt = 1 pt west\else east\fi},
            color={\ifdim\valueissmall  pt = 1 pt black\else black\fi}
    }
    ]
    \addplot +[black, fill=DuckDB,    postaction={pattern=north east lines, pattern color=white}] coordinates {
                            (89,1)
                            (180,10)
                            (37,100)
                            (16,1000)
                            (29,TO)
    };

\end{axis}
\end{tikzpicture}

\begin{tikzpicture} [transform shape, scale=0.7]
\begin{axis}[
        xbar,  
        every tick label/.append style={scale=2},
        xmax=230,
        xmin=0,
        width=6cm,
        height=4cm,
        xlabel={DuckDB, 0MA Agg.},
        symbolic y coords={1,10,100,1000,TO},
        ytick=data,         
        nodes near coords  align={vertical},
        nodes near coords = \rotatebox{0}{{\pgfmathprintnumber[fixed zerofill, precision=0]
        {\pgfplotspointmeta}}},
        point meta=rawx,
    visualization depends on={meta < 25 \as \valueissmall},
    every node near coord/.append style={
            anchor={\ifdim\valueissmall  pt = 1 pt west\else east\fi},
            color={\ifdim\valueissmall  pt = 1 pt black\else black\fi}
    }
    ]
    \addplot + [black, fill=DuckDB]  coordinates {
                            (217,1)
                            (43,10)
                            (17,100)
                            (15,1000)
                            (58,TO)
    };

\end{axis}
\end{tikzpicture}\begin{tikzpicture} [transform shape, scale=0.7]
\begin{axis}[
        xbar,  
        every tick label/.append style={scale=2},
        xmax=230,
        xmin=0,
        width=6cm,
        height=4cm,
        xlabel={DuckDB + \ourSystem{}, 0MA Agg.},
        symbolic y coords={1,10,100,1000,TO},
        ytick=data,         
        nodes near coords  align={vertical},
        nodes near coords = \rotatebox{0}{{\pgfmathprintnumber[fixed zerofill, precision=0]
        {\pgfplotspointmeta}}},
        point meta=rawx,
    visualization depends on={meta < 25 \as \valueissmall},
    every node near coord/.append style={
            anchor={\ifdim\valueissmall  pt = 1 pt west\else east\fi},
            color={\ifdim\valueissmall  pt = 1 pt black\else black\fi}
    }
    ]
    \addplot + [black, fill=DuckDB,    postaction={pattern=north east lines, pattern color=white}]  coordinates {
                            (119,1)
                            (226,10)
                            (6,100)
                            (0,1000)
                            (0,TO)
    };

\end{axis}
\end{tikzpicture}

\begin{tikzpicture} [transform shape, scale=0.7]
\begin{axis}[
        xbar,  
        every tick label/.append style={scale=2},
        xmax=230,
        xmin=0,
        width=6cm,
        height=4cm,
        xlabel={PostgreSQL, Full Enum.},
        symbolic y coords={1,10,100,1000,TO},
        ytick=data,         
        nodes near coords  align={vertical},
        nodes near coords = \rotatebox{0}{{\pgfmathprintnumber[fixed zerofill, precision=0]
        {\pgfplotspointmeta}}},
        point meta=rawx,
    visualization depends on={meta < 25 \as \valueissmall},
    every node near coord/.append style={
            anchor={\ifdim\valueissmall  pt = 1 pt west\else east\fi},
            color={\ifdim\valueissmall  pt = 1 pt black\else black\fi}
    }
    ]
    \addplot + [black, fill=Postgres]  coordinates {
                            (135,1)
                            (56,10)
                            (41,100)
                            (20,1000)
                            (97,TO)
    };

\end{axis}
\end{tikzpicture}\begin{tikzpicture} [transform shape, scale=0.7]
\begin{axis}[
        xbar,  
        every tick label/.append style={scale=2},
        xmax=230,
        xmin=0,
        width=6cm,
        height=4cm,
        xlabel={PostgreSQL + \ourSystem{}, Full Enum.},
        symbolic y coords={1,10,100,1000,TO},
        ytick=data,         
        nodes near coords  align={vertical},
        nodes near coords = \rotatebox{0}{{\pgfmathprintnumber[fixed zerofill, precision=0]
        {\pgfplotspointmeta}}},
        point meta=rawx,
    visualization depends on={meta < 25 \as \valueissmall},
    every node near coord/.append style={
            anchor={\ifdim\valueissmall  pt = 1 pt west\else east\fi},
            color={\ifdim\valueissmall  pt = 1 pt black\else black\fi}
    }
    ]
    \addplot + [black, fill=Postgres,    postaction={pattern=north east lines, pattern color=Postgres!40}]  coordinates {
                            (56,1)
                            (89,10)
                            (96,100)
                            (40,1000)
                            (70,TO)
    };

\end{axis}
\end{tikzpicture}

\begin{tikzpicture} [transform shape, scale=0.7]
\begin{axis}[
        xbar,  
        every tick label/.append style={scale=2},
        xmax=230,
        xmin=0,
        width=6cm,
        height=4cm,
        xlabel={PostgreSQL, 0MA Agg.},
        symbolic y coords={1,10,100,1000,TO},
        ytick=data,         
        nodes near coords  align={vertical},
        nodes near coords = \rotatebox{0}{{\pgfmathprintnumber[fixed zerofill, precision=0]
        {\pgfplotspointmeta}}},
        point meta=rawx,
    visualization depends on={meta < 25 \as \valueissmall},
    every node near coord/.append style={
            anchor={\ifdim\valueissmall  pt = 1 pt west\else east\fi},
            color={\ifdim\valueissmall  pt = 1 pt black\else black\fi}
    }
    ]
    \addplot + [black, fill=Postgres]  coordinates {
                            (143,1)
                            (62,10)
                            (34,100)
                            (19,1000)
                            (91,TO)
    };

\end{axis}
\end{tikzpicture}\begin{tikzpicture} [transform shape, scale=0.7]
\begin{axis}[
        xbar,  
        every tick label/.append style={scale=2},
        xmax=230,
        xmin=0,
        width=6cm,
        height=4cm,
        xlabel={PostgreSQL + \ourSystem{}, 0MA Agg.},
        symbolic y coords={1,10,100,1000,TO},
        ytick=data,         
        nodes near coords  align={vertical},
        nodes near coords = \rotatebox{0}{{\pgfmathprintnumber[fixed zerofill, precision=0]
        {\pgfplotspointmeta}}},
        point meta=rawx,
    visualization depends on={meta < 25 \as \valueissmall},
    every node near coord/.append style={
            anchor={\ifdim\valueissmall  pt = 1 pt west\else east\fi},
            color={\ifdim\valueissmall  pt = 1 pt black\else black\fi}
    }
    ]
    \addplot + [black, fill=Postgres,    postaction={pattern=north east lines, pattern color=Postgres!40}]  coordinates {
                            (109,1)
                            (95,10)
                            (136,100)
                            (9,1000)
                            (2,TO)
    };

\end{axis}
\end{tikzpicture}

\begin{tikzpicture} [transform shape, scale=0.7]
\begin{axis}[
        xbar,  
        every tick label/.append style={scale=2},
        xmax=230,
        xmin=0,
        width=6cm,
        height=4cm,
        xlabel={SparkSQL,  Full Enum.},
        symbolic y coords={1,10,100,1000,TO},
        ytick=data,         
        nodes near coords  align={vertical},
        nodes near coords = \rotatebox{0}{{\pgfmathprintnumber[fixed zerofill, precision=0]
        {\pgfplotspointmeta}}},
        point meta=rawx,
    visualization depends on={meta < 25 \as \valueissmall},
    every node near coord/.append style={
            anchor={\ifdim\valueissmall  pt = 1 pt west\else east\fi},
            color={\ifdim\valueissmall  pt = 1 pt black\else black\fi}
    }
    ]
    \addplot + [black, fill=Spark]  coordinates {
                            (12,1)
                            (114,10)
                            (101,100)
                            (33,1000)
                            (87,TO)
    };

\end{axis}
\end{tikzpicture}\begin{tikzpicture} [transform shape, scale=0.7]
\begin{axis}[
        xbar,  
        every tick label/.append style={scale=2},
        xmax=230,
        xmin=0,
        width=6cm,
        height=4cm,
        xlabel={SparkSQL + \ourSystem{},  Full Enum.},
        symbolic y coords={1,10,100,1000,TO},
        ytick=data,         
        nodes near coords  align={vertical},
        nodes near coords = \rotatebox{0}{{\pgfmathprintnumber[fixed zerofill, precision=0]
        {\pgfplotspointmeta}}},
        point meta=rawx,
    visualization depends on={meta < 25 \as \valueissmall},
    every node near coord/.append style={
            anchor={\ifdim\valueissmall  pt = 1 pt west\else east\fi},
            color={\ifdim\valueissmall  pt = 1 pt black\else black\fi}
    }
    ]
    \addplot + [black, fill=Spark,    postaction={pattern=north east lines, pattern color=white}]  coordinates {
                            (0,1)
                            (57,10)
                            (163,100)
                            (102,1000)
                            (29,TO)
    };

\end{axis}
\end{tikzpicture}

\begin{tikzpicture} [transform shape, scale=0.7]
\begin{axis}[
        xbar,  
        every tick label/.append style={scale=2},
        xmax=230,
        xmin=0,
        width=6cm,
        height=4cm,
        xlabel={SparkSQL, 0MA Agg.},
        symbolic y coords={1,10,100,1000,TO},
        ytick=data,         
        nodes near coords  align={vertical},
        nodes near coords = \rotatebox{0}{{\pgfmathprintnumber[fixed zerofill, precision=0]
        {\pgfplotspointmeta}}},
        point meta=rawx,
    visualization depends on={meta < 25 \as \valueissmall},
    every node near coord/.append style={
            anchor={\ifdim\valueissmall  pt = 1 pt west\else east\fi},
            color={\ifdim\valueissmall  pt = 1 pt black\else black\fi}
    }
    ]
    \addplot + [black, fill=Spark]  coordinates {
                            (11,1)
                            (117,10)
                            (98,100)
                            (33,1000)
                            (91,TO)
    };

\end{axis}
\end{tikzpicture}\begin{tikzpicture} [transform shape, scale=0.7]
\begin{axis}[
        xbar,  
        every tick label/.append style={scale=2},
        xmax=230,
        xmin=0,
        width=6cm,
        height=4cm,
        xlabel={SparkSQL + \ourSystem{}, 0MA Agg.},
        symbolic y coords={1,10,100,1000,TO},
        ytick=data,         
        nodes near coords  align={vertical},
        nodes near coords = \rotatebox{0}{{\pgfmathprintnumber[fixed zerofill, precision=0]
        {\pgfplotspointmeta}}},
        point meta=rawx,
    visualization depends on={meta < 25 \as \valueissmall},
    every node near coord/.append style={
            anchor={\ifdim\valueissmall  pt = 1 pt west\else east\fi},
            color={\ifdim\valueissmall  pt = 1 pt black\else black\fi}
    }
    ]
    \addplot + [black, fill=Spark,    postaction={pattern=north east lines, pattern color=white}]  coordinates {
                            (10,1)
                            (116,10)
                            (200,100)
                            (22,1000)
                            (3,TO)
    };

\end{axis}
\end{tikzpicture}

\end{adjustwidth*}
\caption{Histograms showing how many instances were solved in each time range, with or without \ourSystem{}, for the three systems studied. }
\label{fig:histogram}

\end{figure}
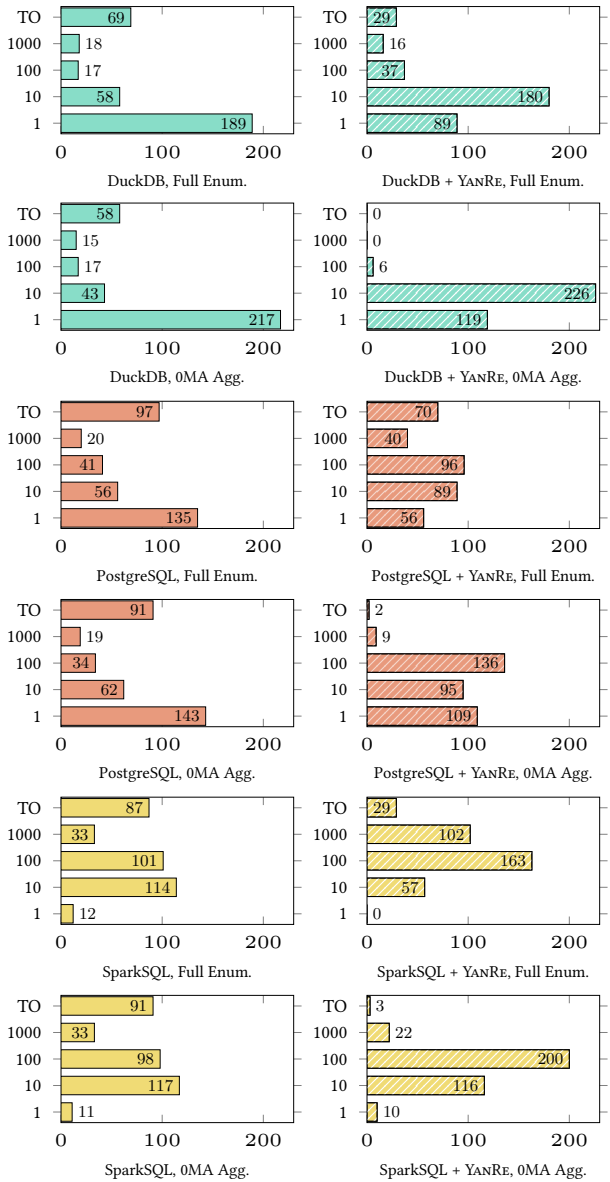

\subsection{Experimental Results}

We primarily concentrate on acyclic queries from the benchmark 
of~\cite{DBLP:conf/sigmod/ManciniKCMA22}. As can be seen in Table~\ref{tab:musicbrainzQueries}, 
the ACQs form the majority of the benchmark, namely 351 out of 435.
Cyclic queries will be briefly discussed separately below. 
\ifArxivVersion
Further details on our experiments are provided in Appendix~\ref{app:AdditionalExp}, 
while further details on cyclic queries are given in Appendix~\ref{sec:cycles}.
\fi

Table~\ref{tab:comparisonNoCutoff} summarises our results for the 
ACQs in the benchmark. The Mean, Med. (Median), and Std. Dev. columns report statistical information for the running times of the benchmark queries. Queries timed out (i.e., which did not terminate within 20 minutes), are counted as having running time 20 minutes. The Max column reports the maximum running time of the queries that did not time out. The number of queries that did not terminate within the time limit is stated in the Timeouts column. Recall, that the Spark SQL experiments were performed on a significantly more powerful system and our experiments are not intended or suited for direct comparison of times between different baseline systems.

We see that the number of queries that execute within the time limit 
of 20 minutes 
is significantly higher when using \ourSystem{}: in DuckDB, the use of \ourSystem reduces the number of timeouts from 69 to 29 for full enumeration queries, and from 58 to no timeouts at all for 0MA aggregation queries. 
In Spark SQL, this reduction is from 
87 to 29 and from 91 to 3, respectively. Consequently, we also see an improvement of up to factor 2 in the mean running times\footnote{Please note that assuming 20 minutes running time here generally benefits the baseline systems as they produce significantly more timeouts and the actual time to execute those queries is often magnitudes beyond the 20 minute limit.}. In PostgreSQL, we see a reduction from 97 to 70 timeouts in the case of full enumeration queries, and from 91 timeouts to just 2
in the 0MA aggregation case. Furthermore, those additional queries that terminate within the time limit do so with a very clear margin as can be seen from the maximum times.

The low median, contrasting the much higher means, shows that at least half of the queries are reasonably easy to solve for the baseline systems. This is expected, as the number of relations is uniformly distributed in the queries, meaning that a fair amount of queries are small enough for typical query planning strategies to work well. This observation is discussed in further detail below.
The split into multiple SQL statements as well as the creation of various temporary tables as performed by \ourSystem naturally leads to some overhead. This is clearly visible in the the higher median execution time with \ourSystem. 
\nop{*******************
\footnote{One must however keep in mind the significant distorting effect of the smaller number of timeouts
with \ourSystem:  high (but still $\leq 20$ min) execution times 
of some hard queries show up in the statistics of \ourSystem but are hidden
by timeouts in the baseline case.}
*******************}
This comes as little surprise, since the structure-guided approach is 
most effective for hard cases. 

Beyond the general improvement, we observe a particularly large improvement for 0MA queries.
For DuckDB + \ourSystem, not only are all queries solved within the 20 minute time limit, but all are solved within only 16 \emph{seconds}. In case of Spark SQL  + \ourSystem and PostgreSQL + \ourSystem, even though the mean speed-up is smaller, we still observe an improvement by an order of magnitude as well as the elimination of almost all timeout cases (see below for further discussion of the remaining timeouts).
While we performed minimum aggregation in our experiments, any natural 0MA version of the queries (e.g., counting or enumerating distinct values of some attribute) would result in essentially the same running times in the \ourSystem cases. 
Note however, that the median times are still slightly lower without \ourSystem, again demonstrating that the structure-guided approach is particularly well suited for complementing traditional query processing strategies in difficult cases. Ultimately, we see that the use of \ourSystem makes these types of queries feasible, whereas our experiments show that none of the baseline systems tested 
here can be relied upon to produce answers for such queries within a reasonable
time limit.

While our results for all three systems follow the same trends, 
we see that PostgreSQL performs significantly worse. 
In particular, the 2 timeouts for 0MA aggregation queries 
with \ourSystem on top of PostgreSQL
are surprising and merit discussion. In these two cases, the join trees contain nodes with a large number of children. The resulting SQL statement generated by \ourSystem therefore expresses many semi-joins at once. While this is not an issue in principle, and generally works as expected, the query planner of PostgreSQL runs into the usual problem with large queries and fails to recognize that semi-joins are possible here. Instead, PostgreSQL chooses join operations, which leads to a blow-up of intermediate results, same as with the original query, and consequently PostgreSQL runs out of time. 
In fact, the problem that the query planner decides against semi-joins and uses joins instead 
also appears in 
the full enumeration case and occasionally also 
with Spark SQL. 
As a mitigation, 
one could adapt the rewriting to perform the semi-joins one after the other in such cases. However, we refrained from doing so since we wanted to 
provide a portable system that allows us to compare 
in a uniform way the general feasibility of the structure-guided methods over a wider variety of existing systems. Naturally, deeper integration of structure-guided methods into a DBMS would immediately 
eliminate such problems.

Especially in the context of Spark SQL, communication cost is another important factor where \ourSystem has notable impact.  A detailed report on the communication cost (and peak memory usage) in Spark SQL
\ifArxivVersion
is given in Appendix~\ref{sect:spark:appendix}.
\fi

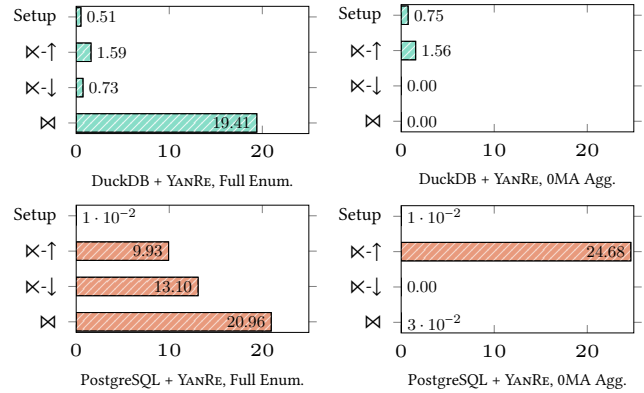
\begin{figure}[tb]

\begin{adjustwidth*}{}{-0.5em}

\begin{tikzpicture} [transform shape, scale=0.7]
\begin{axis}[
        xbar,  
        every tick label/.append style={scale=2},
        xmax=25,
        xmin=0,
        width=6cm,
        height=4cm,
        xlabel={DuckDB + \ourSystem{}, Full Enum.},
        symbolic y coords={$\bbowtie$,$\lltimes$-$\downarrow$, $\lltimes$-$\uparrow$,Setup},
        ytick=data,         
        nodes near coords  align={vertical},
        nodes near coords = \rotatebox{0}{{\pgfmathprintnumber[fixed zerofill, precision=2]
        {\pgfplotspointmeta}}},
        point meta=rawx,
    visualization depends on={meta < 5 \as \valueissmall},
    every node near coord/.append style={
            anchor={\ifdim\valueissmall  pt = 1 pt west\else east\fi},
            color={\ifdim\valueissmall  pt = 1 pt black\else black\fi}
    }
    ]
    \addplot +[black, fill=DuckDB,    postaction={pattern=north east lines, pattern color=DuckDB!20}]  coordinates {
                            (0.51,Setup)
                            (1.59,$\lltimes$-$\uparrow$)
                            (0.73,$\lltimes$-$\downarrow$)
                            (19.41,$\bbowtie$)
    };

\end{axis}
\end{tikzpicture}\phantom{a}\begin{tikzpicture} [transform shape, scale=0.7]
\begin{axis}[
        xbar,  
        every tick label/.append style={scale=2},
        xmax=25,
        xmin=0,
        width=6cm,
        height=4cm,
        xlabel={DuckDB + \ourSystem{}, 0MA Agg.},
        symbolic y coords={$\bbowtie$,$\lltimes$-$\downarrow$, $\lltimes$-$\uparrow$,Setup},
        ytick=data,         
        nodes near coords  align={vertical},
        nodes near coords = \rotatebox{0}{{\pgfmathprintnumber[fixed zerofill, precision=2]
        {\pgfplotspointmeta}}},
        point meta=rawx,
    visualization depends on={meta < 2  \as \valueissmall},
    every node near coord/.append style={
            anchor={\ifdim\valueissmall  pt = 1 pt west\else east\fi},
            color={\ifdim\valueissmall  pt = 1 pt black\else black\fi}
    }
    ]
    \addplot +[black, fill=DuckDB,    postaction={pattern=north east lines, pattern color=DuckDB!20}]  coordinates {
                          (0.75,Setup)
                          (1.56,$\lltimes$-$\uparrow$)
                          (0.0,$\lltimes$-$\downarrow$)
                          (0.0,$\bbowtie$)
    };

\end{axis}
\end{tikzpicture}

\begin{tikzpicture} [transform shape, scale=0.7]
\begin{axis}[
        xbar,
        every tick label/.append style={scale=2},
        xmax=25,
        xmin=0,
        width=6cm,
        height=4cm,
        xlabel={PostgreSQL + \ourSystem{}, Full Enum.},
        symbolic y coords={$\bbowtie$,$\lltimes$-$\downarrow$, $\lltimes$-$\uparrow$,Setup},
        ytick=data,         
        nodes near coords  align={vertical},
        nodes near coords = \rotatebox{0}{{\pgfmathprintnumber[fixed zerofill, precision=2]
        {\pgfplotspointmeta}}},
        point meta=rawx,
    visualization depends on={meta < 5 \as \valueissmall},
    every node near coord/.append style={
            anchor={\ifdim\valueissmall  pt = 1 pt west\else east\fi},
            color={\ifdim\valueissmall  pt = 1 pt black\else black\fi}
    }
    ]
    \addplot +[black, fill=Postgres,    postaction={pattern=north east lines, pattern color=Postgres!40}]  coordinates {
                            (0.01,Setup)
                            (9.93,$\lltimes$-$\uparrow$)
                            (13.10,$\lltimes$-$\downarrow$)
                            (20.96,$\bbowtie$)
    };

\end{axis}
\end{tikzpicture}\phantom{a}\begin{tikzpicture} [transform shape, scale=0.7]
\begin{axis}[
        xbar,  
        every tick label/.append style={scale=2},
        xmax=25,
        xmin=0,
        width=6cm,
        height=4cm,
        xlabel={PostgreSQL + \ourSystem{}, 0MA Agg.},
        symbolic y coords={$\bbowtie$,$\lltimes$-$\downarrow$, $\lltimes$-$\uparrow$,Setup},
        ytick=data,         
        nodes near coords  align={vertical},
        nodes near coords = \rotatebox{0}{{\pgfmathprintnumber[fixed zerofill, precision=2]
        {\pgfplotspointmeta}}},
        point meta=rawx,
    visualization depends on={meta < 5 \as \valueissmall},
    every node near coord/.append style={
            anchor={\ifdim\valueissmall  pt = 1 pt west\else east\fi},
            color={\ifdim\valueissmall  pt = 1 pt black\else black\fi}
    }
    ]
    \addplot +[black, fill=Postgres,    postaction={pattern=north east lines, pattern color=Postgres!40}]  coordinates {
                            (0.01,Setup)
                            (24.68,$\lltimes$-$\uparrow$)
                            (0.0,$\lltimes$-$\downarrow$)
                            (0.03,$\bbowtie$)
    };

\end{axis}
\end{tikzpicture}

\end{adjustwidth*}

\caption{A breakdown of the average time in seconds spent by DuckDB and PostgreSQL 
in each of the four stages of \ourSystem{}, 
corresponding to the 4 phases of \YA. }
\label{fig:breakdownPhases}
\end{figure}

In Figure~\ref{fig:histogram}, we provide histograms of how many queries could be executed within certain time brackets, with brackets of $t \leq 1$, $1 < t \leq 10$, $10< t \leq 100$,  $100 < t\leq 1000$ seconds (represented by their upper bounds in the figure). 
Additionally, we  also list the number of queries that timed out (TO).
In order to simplify the presentation, we thus ignore in total 5 runs that took between 1000 and 1200 seconds.
The left column of histograms represents the baseline systems. We see a trend of queries being either easy or very difficult for a system, with especially the large bracket of times between $100$ and $1000$ seconds being the least common (in particular, for DuckDB and PostgreSQL).
The histograms also give a better insight into the improvements achieved through the use of \ourSystem{}. With DuckDB, we see that many of the queries causing a timeout with the baseline system can be solved far below the timeout threshold with \ourSystem{}, even in the full enumeration case. At the same time, due to the overhead of  \ourSystem{}, the number of queries that are solved in under a second is significantly lower. For PostgreSQL and Spark SQL we see that the overhead and the aforementioned issues around planners avoiding semi-join operations cause a general trend towards slower evaluation in full enumeration queries, despite significant reduction in timeouts.

Figure~\ref{fig:breakdownPhases} provides a breakdown of the average time (in seconds) spent in each of the four stages of the \ourSystem rewriting (excluding timeouts). 
We have omitted Spark SQL in this figure, 
since we have applied there a slightly different approach of executing all stages as one query plan (for details,
\ifArxivVersion
see Appendix~\ref{sect:spark:appendix}.
\fi

The \emph{Setup} phase consists of the creation of various views that represent the initial relations for each node in the join tree (for details, see Appendix~\ref{sect:Implementation}). 
Surprisingly, this takes up a noticeable amount of time in some cases and we expect that these times can be significantly reduced by a 
full integration of structure-guided query evaluation into these systems.
In the case of full enumeration queries, we see for both, DuckDB and PostgreSQL,
that \ourSystem{} spends the most time in the join phase. It is interesting to note that PostgreSQL also spends a lot of time in the two semi-join phases, 
whereas for DuckDB, the time spent there is insignificant relative to the Join phase. As we discussed earlier, we have seen cases where the query planner of PostgeSQL eschews the use of semi-joins, which  explains parts of this marked difference in the time distribution. Additionally, the handling of internal tables and possible bottlenecks in their creation are another potential factor for this discrepancy. In the 0MA aggregation case, we see that both systems fare very similarly, with DuckDB again requiring more time for the Setup stage. The increase in Setup time over the full enumeration case here is due to the larger number of instances that could be solved without timeout for 0MA queries. Note that in the 0MA case, the "Join" phase consists only of the final aggregation in the root node, which 
explains the (almost) 0 time consumption.

\paragraph{\ourSystem Planning Time}
The time required by \ourSystem to create the rewriting is negligible even in our unoptimised proof-of-concept implementation. Even for the largest queries (30 relations), the computation of hypergraphs and join trees as well as the subsequent rewriting
requires only a few milliseconds. This is magnitudes faster than usual planning times by the host systems for complex queries and we therefore do not provide a more detailed analysis of our planning time here. Detailed records of the times spent by \ourSystem in the various phases of query execution 
are available in the aforementioned repository of data and code artifacts.

\nop{***************************
\subsection{Comparison of Host Systems}
\label{sec:hosts}

We have seen that for acyclic queries, structure-guided query evaluation can. This discrepancy mostly stems from the fact that the query planner of PostgreSQL  seems to strongly favour full joins over the use of semi-joins. As was mentioned, PostgreSQL is a more classical DBMS compared with DuckDB, which uses a columnar storage of data. Our experiments focus only on these two systems, though it might be interesting to explore in a next step if this strong preference of query plans with full joins is also seen in other DBMSs.  If other DBMSs behave similarly, it might be an argument for seeking a more closely integrated implementation of Yannakakis's algorithm, since simply rewriting query plans is not enough to get these systems to allow actual semi-joins to be used. 
On the other hand, for more modern database systems, such as DuckDB, the use of queries produced by \ourSystem{} is clearly already a very good idea when dealing with difficult join queries.
********************}

\paragraph{A Glimpse Beyond}
To get a feeling of how a structure-guided approach to query processing generalises beyond ACQs, 
\ifArxivVersion
we have carried out some very preliminary experiments with a few cyclic
queries from the benchmark of \cite{DBLP:conf/sigmod/ManciniKCMA22}, 
which we briefly discuss next.
Further details on cyclic queries are given in 
Appendix~\ref{sec:cycles}.

\fi
In Table~\ref{tab:ghdqueries}, we show some of these results: we have 
chosen 3 of the smallest cyclic queries from the benchmark (called 09ac, 11ag, and 
11al). As is indicated by their names, these queries involve the join of 
9 resp.\ 11 relations. 
For each of these queries, we have computed 8 different generalized hypertree decompositions (GHDs) of width 2, which is optimal in these cases. Actually, for 09ac, we were only able to find 7 distinct GHDs.
Turning the GHDs into join trees by carrying out the local joins at each node of the GHD and applying our \ourSystem system on DuckDB, we obtained the run times (sorted in ascending order) reported in Table~\ref{tab:ghdqueries}.  Without \ourSystem, the corresponding run times of  DuckDB are 
timeout (query 09ac), 22.22s (query 11ag), and  263.87s (query 11al), respectively. 
For all queries, we notice a striking discrepancy 
in execution times of DuckDB + \ourSystem depending on the chosen GHD: 
in the best case, DuckDB + \ourSystem may be way faster
than plain DuckDB, in the worst case, DuckDB + \ourSystem times out.

\begin{table}[t]
    \caption{Run times of cyclic queries with different GHDs}
    \label{tab:ghdqueries}
    \centering

    \begin{tabular}{l|cccccccc}
    \toprule

    {Query} & \multicolumn{8}{c}{Ordered Eval. Time by GHD (s)}  \\
    \midrule
     09ac    & 10.3 & 16.5 & 18.4 & t/o & t/o & t/o & t/o & --- \\
          11ag & 11.2 & 26.8 & t/o & t/o & t/o & t/o & t/o & t/o\\
          11al & 6.2 & 6.3 & 8.3 & 258 & t/o & t/o & t/o & t/o\\
    \bottomrule
    \end{tabular}
\end{table}

\nop{**************************
The generalisation of query acyclicity via various forms of decompositions and associated widths has been studied extensively for the last two decades -- with deep theoretical studies of 
the properties of more advanced decompositions and also with the development of powerful decomposition algorithms and their implementations, see e.g.,~\cite{DBLP:conf/pods/Lanzinger22,DBLP:journals/jacm/GottlobLPR21,DBLP:conf/pods/GottlobLS99,DBLP:journals/siamcomp/AtseriasGM13}.
Recent analyses of benchmarks and query logs have revealed that many 
queries in practice actually are acyclic or have low width
\cite{DBLP:journals/vldb/BonifatiMT20,DBLP:journals/jea/FischlGLP21}. 
**************************}

To summarise, our preliminary experiments with cyclic CQs show that 
there is clear potential for structure-guided
query answering beyond acyclic queries. But they also show that
this requires new methods for finding
the ``right'' decompositions. 
Indeed, the key observation 
is that a good choice of decomposition is absolutely crucial for the performance of query evaluation.
Previously mentioned related work by Ghionna et al.~\cite{DBLP:conf/icde/GhionnaGGS07}, 
and Scarcello et al.~\cite{DBLP:conf/pods/ScarcelloGL04} 
may provide a good starting point for this research direction.

\subsection{Deeper Insight into Improvements}
\label{sect:deeperInsights}

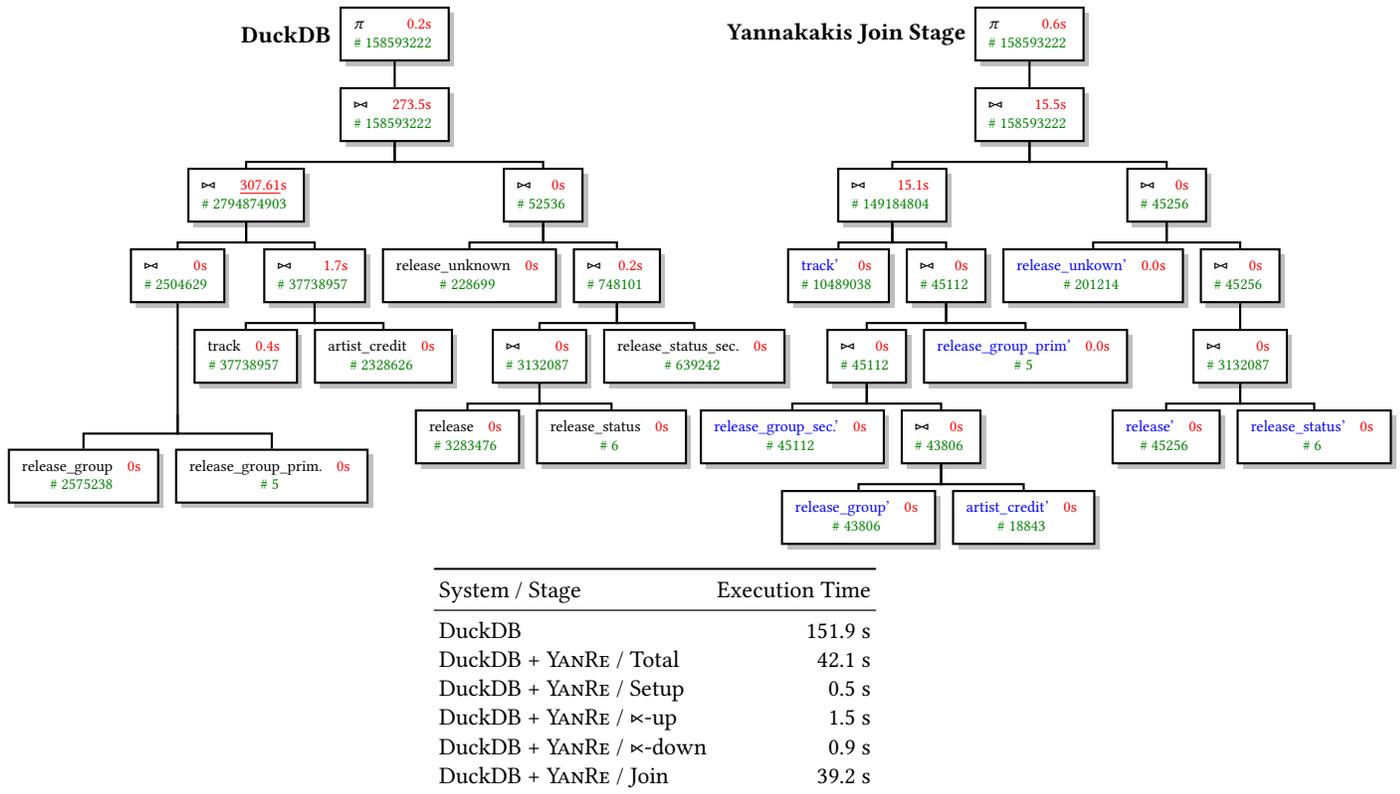
\begin{figure*}[t]
\centering

\begin{minipage}[t]{0.5\textwidth}
        \vspace{0pt}
        \begin{forest}
          my tree style
          [\QueryPlanNode{$\pi$}{0.2}{158593222},label=left:{\textbf{DuckDB}}
            [\QueryPlanNode{$\bowtie$}{273.5}{158593222}
                [\QueryPlanNode{$\bowtie$}{\underline{307.61}}{2794874903}
                    [\QueryPlanNode{$\bowtie$}{0}{2504629}
                        [[
                            [\QueryPlanNode{release\_group}{0}{2575238}]
                            [\QueryPlanNode{release\_group\_prim.}{0}{5}]
                        ]]
                    ]
                    [\QueryPlanNode{$\bowtie$}{1.7}{37738957}
                        [\QueryPlanNode{track}{0.4}{37738957}]
                        [\QueryPlanNode{artist\_credit}{0}{2328626}]
                    ]
                ]
                [\QueryPlanNode{$\bowtie$}{0}{52536}
                    [\QueryPlanNode{release\_unknown}{0}{228699}]
                    [\QueryPlanNode{$\bowtie$}{0.2}{748101}
                        [\QueryPlanNode{$\bowtie$}{0}{3132087}
                            [\QueryPlanNode{release}{0}{3283476}]
                            [\QueryPlanNode{release\_status}{0}{6}]
                        ]
                        [\QueryPlanNode{release\_status\_sec.}{0}{639242}]
                    ]
                ]
            ]
          ]
        \end{forest}
\end{minipage}~\phantom{a}~\begin{minipage}[t]{0.45\textwidth}
        \vspace{0pt}
        \begin{forest}
          my tree style
          [\QueryPlanNode{$\pi$}{0.6}{158593222},label=left:{\textbf{Yannakakis Join Stage}}
            [\QueryPlanNode{$\bowtie$}{15.5}{158593222}
                [\QueryPlanNode{$\bowtie$}{15.1}{149184804}
                    [\QueryPlanNode{\color{blue}track'}{0}{10489038}]
                    [\QueryPlanNode{$\bowtie$}{0}{45112}
                        [\QueryPlanNode{$\bowtie$}{0}{45112}
                            [\QueryPlanNode{\color{blue}release\_group\_sec.'}{0}{45112}]
                            [\QueryPlanNode{$\bowtie$}{0}{43806}
                                [\QueryPlanNode{\color{blue}release\_group'}{0}{43806}]
                                [\QueryPlanNode{\color{blue}artist\_credit'}{0}{18843}]
                            ]
                        ]
                        [\QueryPlanNode{\color{blue}release\_group\_prim'}{0.0}{5}]
                    ]
                ]
                [\QueryPlanNode{$\bowtie$}{0}{45256}
                    [\QueryPlanNode{\color{blue}release\_unkown'}{0.0}{201214}]
                    [\QueryPlanNode{$\bowtie$}{0}{45256}
                        [\QueryPlanNode{$\bowtie$}{0}{3132087}
                            [\QueryPlanNode{\color{blue}release'}{0}{45256}]
                            [\QueryPlanNode{\color{blue}release\_status'}{0}{6}]
                        ]
                    ]
                ]
            ]
          ]
        \end{forest}
\end{minipage}           

\bigskip 

\begin{tabular}{lr}
\toprule
System / Stage & Execution Time \\
\midrule 
DuckDB & 151.9 s \\
DuckDB + \ourSystem{} / Total & 42.1  s \\ 
DuckDB + \ourSystem{} / Setup & 0.5  s \\ 
DuckDB + \ourSystem{} / $\ltimes$-up & 1.5 s \\ 
DuckDB + \ourSystem{} / $\ltimes$-down & 0.9 s \\ 
DuckDB + \ourSystem{} / Join & 39.2 s \\ 
\bottomrule
\end{tabular}

\caption{Details of performance difference in query plans of query 08ad. Execution times of operations are in seconds, rounded to one decimal point.}
  \label{fig:q08ad}
\end{figure*}

We see that structure-guided query evaluation can significantly improve the performance of widely used DBMSs on difficult queries, even if all joins are along foreign key relationships. In this section, we further illustrate the reasons for these improvements in detail. 

We consider the evaluation of benchmark
query 08ad (for the full enumeration case), which is illustrated in Figure~\ref{fig:q08ad}. 
On the left-hand side, we show the query plan (projections at leaf nodes
are omitted in the figure) as produced by DuckDB on the input query. 
On the right-hand side, we show the query plan that was produced by DuckDB for the final Join stage query in the \ourSystem{} rewriting. That is, all relations at this point have been reduced by the two 
semi-join passes.
To emphasise this, we refer to the reduced version of each relation $R$ as $R'$ in the right tree and mark it in blue. The size of each relation is given in green after a \#, and the times in the nodes represent total CPU time (note that this differs strongly from wall clock time due to heavy parallelisation) spent on this operation. The query produces a large number of output tuples ($\approx$ 158 million). However, while our rewriting still has to materialise all of these tuples (at significant computational cost), the baseline query plan produces an even larger and more costly intermediate result with $\approx$ 2.8 billion tuples on the way to the final output. 
Actually, the huge discrepancy between the original vs.\ reduced relations is already seen at the leaf nodes of the two query plans: for instance, 
when we look at the 
relations artist\_credit, release, and release\_group, 
the reduction in size is by a factor of 123, 72, and 58, respectively.

The table at the bottom of the figure provides the wall clock times for evaluation of the baseline using only DuckDB, as well as DuckDB+\ourSystem{}. The baseline plan on the left required 151.9 seconds, while our approach took 42.1 seconds to execute. Notably, we see that the significant improvement in the join phase comes at a very cheap cost: the two semi-join phases that allowed us to avoid the blow-up required only a total of 2.4 seconds. Thus, while the query is still solvable in reasonable time in the baseline case, we see that even such cases can be significantly improved 
by a structure-guided approach. 

We note that this query has only 8 relations and the planning phase is therefore still manageable in the baseline case. Specifically, PostgreSQL manages to answer the query in 64 seconds, while only Spark SQL times out. Importantly, even if all joins follow foreign key relationships, there can still be an enormous blow-up of intermediate results if an 
evaluation strategy based solely on joins (without using semi-joins to 
remove dangling tuples first) is applied.  Advancements in cardinality estimation, which aim at the computation of good join plans, are therefore inherently insufficient on these types of challenging queries.

\paragraph{Indexes}
Indexes have traditionally been an important factor in fast join evaluation in 
DBMSs. However, when the time to evaluate a query is dominated by efforts related to large intermediate results, indexes are of little to no help as they cannot decrease the size of a join. This observation is also confirmed by our experiments with 
three different DBMSs, which apply significantly different indexing strategies and yet yield  comparable experimental results. 
In PostgreSQL, it is common to maintain a large number of explicitly specified and materialised indexes for all attributes that are deemed important. In our experiments for PostgreSQL we use all indexes that are set in the Musicbrainz dataset, which are effectively on all attributes over which joins are made in our queries. In contrast, Spark SQL supports no indexes at all and DuckDB does not allow persistent indexes (every new session requires a new creation of indexes), but internally maintains ad-hoc index structures for commonly accessed values and attributes. Our experiments therefore run without explicitly declared indexes on both systems\footnote{Creating all indexes in DuckDB takes over 30 minutes on our test system and it was infeasible to add this overhead to every tested query. Additional experiments showed that explicitly creating the same indexes in DuckDB as in PostgreSQL makes no significant difference to our measured times.}. 
Despite these differences, we see consistent improvements using \ourSystem{} over all systems. Furthermore, PostgreSQL performs worst in every measure despite the most elaborate support of indexes among the 
3 systems tested here.

\balance

\section{Conclusion and Future Work}
\label{sec:conclusion}

In this work, we have studied the effectiveness of 
Yannakakis-style query evaluation by common, widely used, 
relational DBMSs
on simply structured yet large queries. We observe that 
these kinds of queries can be highly challenging.
On the other hand, structure-guided query evaluation -- executed by the same DBMSs -- greatly improves on the number of such queries that are answerable in reasonable time (the majority of the remaining timeouts being  due to "unavoidable" materialisation of an infeasibly large number of output tuples). To the best of our knowledge, this is the first extensive study (based on 
over 300 benchmark queries from \cite{DBLP:conf/sigmod/ManciniKCMA22}) that confirms these long-standing theoretical ideas as also being useful in combination with standard database technology.

We have formally introduced a relevant class of queries
which are particularly well suited for structure-guided 
query processing -- with a potential speed-up by several orders of magnitude. 
However, our experiments show that also 
large join queries outside this class 
may significantly profit from such an approach. 
Our experiments 
were based on a novel rewriting technique, which enforces a Yannakakis-style query evaluation by state-of-the-art DBMSs without touching the internals of the DBMSs themselves. This opens the door for extending the experiments reported here  
also to closed source commercial DBMSs.

We conclude from our study that a systematic, deep integration of structure-guided query processing into existing database technology is a worthwhile goal for future research. This is an ambitious, highly non-trivial goal, which involves the reconciliation of two seemingly contradicting query processing paradigms. However, the prospect of providing a solution or, at least, an alleviation to two of the most 
pressing problems in query optimization and evaluation 
seems to justify the effort of such an endeavour, namely how to find a good join order for big join queries and how to avoid the explosion of intermediate results.

In addition to the challenging task of a full integration
of one query processing paradigm into the other, we 
envisage two main directions in which our work should be further extended:
first, Yannakakis-style query evaluation has to be extended from ACQs to queries or subqueries of low generalised hypertree-width (which, only in rare cases, is more than 2). The choice of an  {\em optimal\/} decomposition (from many possible decompositions that may not even be required to have 
minimum width) is a highly non-trivial problem, yet crucial as our 
very preliminary experiments with cyclic queries 
(see Table~\ref{tab:ghdqueries}) illustrate. Above all, this will 
require to re-think the computation of hypergraph decompositions and
to take statistics on the data as well as schema-related information (such as foreign keys and functional dependencies) into account.

Finally, we also want to study extensions of the class of 0MA queries and identify further classes of queries that can be evaluated without materialising the joins involved.
Actually, in Example~\ref{ex:tpch:one}
we have encountered a TPC-H query which falls into this category, i.e.: 
it is not 0MA according to Definition \ref{def:SCA} but it behaves like 
a 0MA query due to properties of the schema. We want to identify 
further conditions (on the schema and/or on the queries themselves) 
that allow for 
such a favourable, ``join-free'' evaluation strategy.

\nop{******************************
extending the ideas for 0MA queries to SUM/AVG/COUNT aggregates as outlined in Section~\ref{sect:theory} has the potential to speed up a variety of challenging data analysis tasks. 
******************************}

\ifArxivVersion
\begin{acks}
Georg Gottlob is a Royal Society Research Professor and acknowledges support by the Royal Society in this role through the  “RAISON DATA” project  (Reference No. RP\textbackslash{}R1\textbackslash{}201074). Matthias Lanzinger acknowledges support by the Royal Society  “RAISON DATA” project  (Reference No. RP\textbackslash{}R1\textbackslash{}201074).
The work of Cem Okulmus is supported by the Wallenberg AI, Autonomous Systems and Software
Program (WASP) funded by the Knut and Alice Wallenberg Foundation.
The work of Reinhard Pichler and Alexander Selzer has been 
supported by the Vienna Science and Technology Fund (WWTF) [10.47379/ICT2201, 10.47379/VRG18013, 10.47379/NXT22018]; and
the Christian Doppler Research Association (CDG) JRC LIVE.
\end{acks}
 
\fi

\bibliographystyle{ACM-Reference-Format}
\bibliography{main}


\begin{thebibliography}{49}


\ifx \showCODEN    \undefined \def \showCODEN     #1{\unskip}     \fi
\ifx \showDOI      \undefined \def \showDOI       #1{#1}\fi
\ifx \showISBNx    \undefined \def \showISBNx     #1{\unskip}     \fi
\ifx \showISBNxiii \undefined \def \showISBNxiii  #1{\unskip}     \fi
\ifx \showISSN     \undefined \def \showISSN      #1{\unskip}     \fi
\ifx \showLCCN     \undefined \def \showLCCN      #1{\unskip}     \fi
\ifx \shownote     \undefined \def \shownote      #1{#1}          \fi
\ifx \showarticletitle \undefined \def \showarticletitle #1{#1}   \fi
\ifx \showURL      \undefined \def \showURL       {\relax}        \fi
\providecommand\bibfield[2]{#2}
\providecommand\bibinfo[2]{#2}
\providecommand\natexlab[1]{#1}
\providecommand\showeprint[2][]{arXiv:#2}

\bibitem[mus(2022)]%
        {musicbrainz}
 \bibinfo{year}{2022}\natexlab{}.
\newblock \bibinfo{title}{{MusicBrainz} - {T}he {O}pen {M}usic {E}ncyclopedia}.
\newblock \bibinfo{howpublished}{\url{https://musicbrainz.org/}}.
\newblock
\newblock
\shownote{[Online; accessed 24-July-2022]}.


\bibitem[Aberger et~al\mbox{.}(2017)]%
        {DBLP:journals/tods/AbergerLTNOR17}
\bibfield{author}{\bibinfo{person}{Christopher~R. Aberger},
  \bibinfo{person}{Andrew Lamb}, \bibinfo{person}{Susan Tu},
  \bibinfo{person}{Andres N{\"{o}}tzli}, \bibinfo{person}{Kunle Olukotun},
  {and} \bibinfo{person}{Christopher R{\'{e}}}.}
  \bibinfo{year}{2017}\natexlab{}.
\newblock \showarticletitle{EmptyHeaded: {A} Relational Engine for Graph
  Processing}.
\newblock \bibinfo{journal}{\emph{{ACM} Trans. Database Syst.}}
  \bibinfo{volume}{42}, \bibinfo{number}{4} (\bibinfo{year}{2017}),
  \bibinfo{pages}{20:1--20:44}.
\newblock
\urldef\tempurl%
\url{https://doi.org/10.1145/3129246}
\showDOI{\tempurl}


\bibitem[Afrati et~al\mbox{.}(2017)]%
        {DBLP:conf/icdt/AfratiJRSU17}
\bibfield{author}{\bibinfo{person}{Foto~N. Afrati}, \bibinfo{person}{Manas~R.
  Joglekar}, \bibinfo{person}{Christopher R{\'{e}}}, \bibinfo{person}{Semih
  Salihoglu}, {and} \bibinfo{person}{Jeffrey~D. Ullman}.}
  \bibinfo{year}{2017}\natexlab{}.
\newblock \showarticletitle{{GYM:} {A} Multiround Distributed Join Algorithm}.
  In \bibinfo{booktitle}{\emph{20th International Conference on Database
  Theory, {ICDT} 2017, March 21-24, 2017, Venice, Italy}}
  \emph{(\bibinfo{series}{LIPIcs}, Vol.~\bibinfo{volume}{68})},
  \bibfield{editor}{\bibinfo{person}{Michael Benedikt} {and}
  \bibinfo{person}{Giorgio Orsi}} (Eds.). \bibinfo{publisher}{Schloss Dagstuhl
  - Leibniz-Zentrum f{\"{u}}r Informatik}, \bibinfo{pages}{4:1--4:18}.
\newblock
\urldef\tempurl%
\url{https://doi.org/10.4230/LIPIcs.ICDT.2017.4}
\showDOI{\tempurl}


\bibitem[Atserias et~al\mbox{.}(2013)]%
        {DBLP:journals/siamcomp/AtseriasGM13}
\bibfield{author}{\bibinfo{person}{Albert Atserias}, \bibinfo{person}{Martin
  Grohe}, {and} \bibinfo{person}{D{\'{a}}niel Marx}.}
  \bibinfo{year}{2013}\natexlab{}.
\newblock \showarticletitle{Size Bounds and Query Plans for Relational Joins}.
\newblock \bibinfo{journal}{\emph{{SIAM} J. Comput.}} \bibinfo{volume}{42},
  \bibinfo{number}{4} (\bibinfo{year}{2013}), \bibinfo{pages}{1737--1767}.
\newblock
\urldef\tempurl%
\url{https://doi.org/10.1137/110859440}
\showDOI{\tempurl}


\bibitem[Bagan et~al\mbox{.}(2007)]%
        {DBLP:conf/csl/BaganDG07}
\bibfield{author}{\bibinfo{person}{Guillaume Bagan}, \bibinfo{person}{Arnaud
  Durand}, {and} \bibinfo{person}{Etienne Grandjean}.}
  \bibinfo{year}{2007}\natexlab{}.
\newblock \showarticletitle{On Acyclic Conjunctive Queries and Constant Delay
  Enumeration}. In \bibinfo{booktitle}{\emph{Computer Science Logic, 21st
  International Workshop, {CSL} 2007, 16th Annual Conference of the EACSL,
  Lausanne, Switzerland, September 11-15, 2007, Proceedings}}
  \emph{(\bibinfo{series}{Lecture Notes in Computer Science},
  Vol.~\bibinfo{volume}{4646})}, \bibfield{editor}{\bibinfo{person}{Jacques
  Duparc} {and} \bibinfo{person}{Thomas~A. Henzinger}} (Eds.).
  \bibinfo{publisher}{Springer}, \bibinfo{pages}{208--222}.
\newblock
\urldef\tempurl%
\url{https://doi.org/10.1007/978-3-540-74915-8\_18}
\showDOI{\tempurl}


\bibitem[Bonifati et~al\mbox{.}(2020)]%
        {DBLP:journals/vldb/BonifatiMT20}
\bibfield{author}{\bibinfo{person}{Angela Bonifati}, \bibinfo{person}{Wim
  Martens}, {and} \bibinfo{person}{Thomas Timm}.}
  \bibinfo{year}{2020}\natexlab{}.
\newblock \showarticletitle{An analytical study of large {SPARQL} query logs}.
\newblock \bibinfo{journal}{\emph{{VLDB} J.}} \bibinfo{volume}{29},
  \bibinfo{number}{2-3} (\bibinfo{year}{2020}), \bibinfo{pages}{655--679}.
\newblock
\urldef\tempurl%
\url{https://doi.org/10.1007/s00778-019-00558-9}
\showDOI{\tempurl}


\bibitem[Brault{-}Baron(2016)]%
        {DBLP:journals/csur/Brault-Baron16}
\bibfield{author}{\bibinfo{person}{Johann Brault{-}Baron}.}
  \bibinfo{year}{2016}\natexlab{}.
\newblock \showarticletitle{Hypergraph Acyclicity Revisited}.
\newblock \bibinfo{journal}{\emph{{ACM} Comput. Surv.}} \bibinfo{volume}{49},
  \bibinfo{number}{3} (\bibinfo{year}{2016}), \bibinfo{pages}{54:1--54:26}.
\newblock
\urldef\tempurl%
\url{https://doi.org/10.1145/2983573}
\showDOI{\tempurl}


\bibitem[Carmeli and Kr{\"{o}}ll(2020)]%
        {DBLP:journals/mst/CarmeliK20}
\bibfield{author}{\bibinfo{person}{Nofar Carmeli} {and} \bibinfo{person}{Markus
  Kr{\"{o}}ll}.} \bibinfo{year}{2020}\natexlab{}.
\newblock \showarticletitle{Enumeration Complexity of Conjunctive Queries with
  Functional Dependencies}.
\newblock \bibinfo{journal}{\emph{Theory Comput. Syst.}} \bibinfo{volume}{64},
  \bibinfo{number}{5} (\bibinfo{year}{2020}), \bibinfo{pages}{828--860}.
\newblock
\urldef\tempurl%
\url{https://doi.org/10.1007/s00224-019-09937-9}
\showDOI{\tempurl}


\bibitem[Carmeli and Kr{\"{o}}ll(2021)]%
        {DBLP:journals/tods/CarmeliK21}
\bibfield{author}{\bibinfo{person}{Nofar Carmeli} {and} \bibinfo{person}{Markus
  Kr{\"{o}}ll}.} \bibinfo{year}{2021}\natexlab{}.
\newblock \showarticletitle{On the Enumeration Complexity of Unions of
  Conjunctive Queries}.
\newblock \bibinfo{journal}{\emph{{ACM} Trans. Database Syst.}}
  \bibinfo{volume}{46}, \bibinfo{number}{2} (\bibinfo{year}{2021}),
  \bibinfo{pages}{5:1--5:41}.
\newblock
\urldef\tempurl%
\url{https://doi.org/10.1145/3450263}
\showDOI{\tempurl}


\bibitem[Carmeli et~al\mbox{.}(2021)]%
        {DBLP:conf/pods/CarmeliTGKR21}
\bibfield{author}{\bibinfo{person}{Nofar Carmeli}, \bibinfo{person}{Nikolaos
  Tziavelis}, \bibinfo{person}{Wolfgang Gatterbauer}, \bibinfo{person}{Benny
  Kimelfeld}, {and} \bibinfo{person}{Mirek Riedewald}.}
  \bibinfo{year}{2021}\natexlab{}.
\newblock \showarticletitle{Tractable Orders for Direct Access to Ranked
  Answers of Conjunctive Queries}. In \bibinfo{booktitle}{\emph{PODS'21:
  Proceedings of the 40th {ACM} {SIGMOD-SIGACT-SIGAI} Symposium on Principles
  of Database Systems, Virtual Event, China, June 20-25, 2021}},
  \bibfield{editor}{\bibinfo{person}{Leonid Libkin}, \bibinfo{person}{Reinhard
  Pichler}, {and} \bibinfo{person}{Paolo Guagliardo}} (Eds.).
  \bibinfo{publisher}{{ACM}}, \bibinfo{pages}{325--341}.
\newblock
\urldef\tempurl%
\url{https://doi.org/10.1145/3452021.3458331}
\showDOI{\tempurl}


\bibitem[Carmeli et~al\mbox{.}(2022)]%
        {DBLP:journals/tods/CarmeliZBCKS22}
\bibfield{author}{\bibinfo{person}{Nofar Carmeli}, \bibinfo{person}{Shai
  Zeevi}, \bibinfo{person}{Christoph Berkholz}, \bibinfo{person}{Alessio
  Conte}, \bibinfo{person}{Benny Kimelfeld}, {and} \bibinfo{person}{Nicole
  Schweikardt}.} \bibinfo{year}{2022}\natexlab{}.
\newblock \showarticletitle{Answering (Unions of) Conjunctive Queries using
  Random Access and Random-Order Enumeration}.
\newblock \bibinfo{journal}{\emph{{ACM} Trans. Database Syst.}}
  \bibinfo{volume}{47}, \bibinfo{number}{3} (\bibinfo{year}{2022}),
  \bibinfo{pages}{9:1--9:49}.
\newblock
\urldef\tempurl%
\url{https://doi.org/10.1145/3531055}
\showDOI{\tempurl}


\bibitem[Dieu et~al\mbox{.}(2009)]%
        {DBLP:journals/pvldb/DieuDFLS09}
\bibfield{author}{\bibinfo{person}{Nicolas Dieu}, \bibinfo{person}{Adrian
  Dragusanu}, \bibinfo{person}{Fran{\c{c}}oise Fabret},
  \bibinfo{person}{Fran{\c{c}}ois Llirbat}, {and} \bibinfo{person}{Eric
  Simon}.} \bibinfo{year}{2009}\natexlab{}.
\newblock \showarticletitle{1, 000 Tables Inside the From}.
\newblock \bibinfo{journal}{\emph{Proc. {VLDB} Endow.}} \bibinfo{volume}{2},
  \bibinfo{number}{2} (\bibinfo{year}{2009}), \bibinfo{pages}{1450--1461}.
\newblock
\urldef\tempurl%
\url{https://doi.org/10.14778/1687553.1687572}
\showDOI{\tempurl}


\bibitem[Fagin(1983)]%
        {DBLP:journals/jacm/Fagin83}
\bibfield{author}{\bibinfo{person}{Ronald Fagin}.}
  \bibinfo{year}{1983}\natexlab{}.
\newblock \showarticletitle{Degrees of Acyclicity for Hypergraphs and
  Relational Database Schemes}.
\newblock \bibinfo{journal}{\emph{J. {ACM}}} \bibinfo{volume}{30},
  \bibinfo{number}{3} (\bibinfo{year}{1983}), \bibinfo{pages}{514--550}.
\newblock
\urldef\tempurl%
\url{https://doi.org/10.1145/2402.322390}
\showDOI{\tempurl}


\bibitem[Fischl et~al\mbox{.}(2021)]%
        {DBLP:journals/jea/FischlGLP21}
\bibfield{author}{\bibinfo{person}{Wolfgang Fischl}, \bibinfo{person}{Georg
  Gottlob}, \bibinfo{person}{Davide~Mario Longo}, {and}
  \bibinfo{person}{Reinhard Pichler}.} \bibinfo{year}{2021}\natexlab{}.
\newblock \showarticletitle{HyperBench: {A} Benchmark and Tool for Hypergraphs
  and Empirical Findings}.
\newblock \bibinfo{journal}{\emph{{ACM} J. Exp. Algorithmics}}
  \bibinfo{volume}{26} (\bibinfo{year}{2021}), \bibinfo{pages}{1.6:1--1.6:40}.
\newblock
\urldef\tempurl%
\url{https://doi.org/10.1145/3440015}
\showDOI{\tempurl}


\bibitem[Flum and Grohe(2004)]%
        {DBLP:journals/siamcomp/FlumG04}
\bibfield{author}{\bibinfo{person}{J{\"{o}}rg Flum} {and}
  \bibinfo{person}{Martin Grohe}.} \bibinfo{year}{2004}\natexlab{}.
\newblock \showarticletitle{The Parameterized Complexity of Counting Problems}.
\newblock \bibinfo{journal}{\emph{{SIAM} J. Comput.}} \bibinfo{volume}{33},
  \bibinfo{number}{4} (\bibinfo{year}{2004}), \bibinfo{pages}{892--922}.
\newblock
\urldef\tempurl%
\url{https://doi.org/10.1137/S0097539703427203}
\showDOI{\tempurl}


\bibitem[Garcia{-}Molina et~al\mbox{.}(2002)]%
        {DBLP:books/daglib/0010423}
\bibfield{author}{\bibinfo{person}{Hector Garcia{-}Molina},
  \bibinfo{person}{Jeffrey~D. Ullman}, {and} \bibinfo{person}{Jennifer Widom}.}
  \bibinfo{year}{2002}\natexlab{}.
\newblock \bibinfo{booktitle}{\emph{Database Systems: The Complete Book}}.
\newblock \bibinfo{publisher}{Pearson Education}.
\newblock
\showISBNx{978-0-13-098043-4}


\bibitem[Geck et~al\mbox{.}(2022)]%
        {DBLP:conf/icdt/GeckKSS22}
\bibfield{author}{\bibinfo{person}{Gaetano Geck}, \bibinfo{person}{Jens
  Keppeler}, \bibinfo{person}{Thomas Schwentick}, {and}
  \bibinfo{person}{Christopher Spinrath}.} \bibinfo{year}{2022}\natexlab{}.
\newblock \showarticletitle{Rewriting with Acyclic Queries: Mind Your Head}. In
  \bibinfo{booktitle}{\emph{25th International Conference on Database Theory,
  {ICDT} 2022, March 29 to April 1, 2022, Edinburgh, {UK} (Virtual
  Conference)}} \emph{(\bibinfo{series}{LIPIcs}, Vol.~\bibinfo{volume}{220})},
  \bibfield{editor}{\bibinfo{person}{Dan Olteanu} {and} \bibinfo{person}{Nils
  Vortmeier}} (Eds.). \bibinfo{publisher}{Schloss Dagstuhl - Leibniz-Zentrum
  f{\"{u}}r Informatik}, \bibinfo{pages}{8:1--8:20}.
\newblock
\urldef\tempurl%
\url{https://doi.org/10.4230/LIPIcs.ICDT.2022.8}
\showDOI{\tempurl}


\bibitem[Ghionna et~al\mbox{.}(2007)]%
        {DBLP:conf/icde/GhionnaGGS07}
\bibfield{author}{\bibinfo{person}{Lucantonio Ghionna}, \bibinfo{person}{Luigi
  Granata}, \bibinfo{person}{Gianluigi Greco}, {and} \bibinfo{person}{Francesco
  Scarcello}.} \bibinfo{year}{2007}\natexlab{}.
\newblock \showarticletitle{Hypertree Decompositions for Query Optimization}.
  In \bibinfo{booktitle}{\emph{Proc. {ICDE} 2007}}. \bibinfo{publisher}{{IEEE}
  Computer Society}, \bibinfo{pages}{36--45}.
\newblock
\urldef\tempurl%
\url{https://doi.org/10.1109/ICDE.2007.367849}
\showDOI{\tempurl}


\bibitem[Gottlob et~al\mbox{.}(2023)]%
        {DBLP:journals/corr/abs-2303-02723}
\bibfield{author}{\bibinfo{person}{Georg Gottlob}, \bibinfo{person}{Matthias
  Lanzinger}, \bibinfo{person}{Davide~Mario Longo}, \bibinfo{person}{Cem
  Okulmus}, \bibinfo{person}{Reinhard Pichler}, {and}
  \bibinfo{person}{Alexander Selzer}.} \bibinfo{year}{2023}\natexlab{}.
\newblock \showarticletitle{Structure-Guided Query Evaluation: Towards Bridging
  the Gap from Theory to Practice}.
\newblock \bibinfo{journal}{\emph{CoRR}}  \bibinfo{volume}{abs/2303.02723}
  (\bibinfo{year}{2023}).
\newblock
\urldef\tempurl%
\url{https://doi.org/10.48550/arXiv.2303.02723}
\showDOI{\tempurl}
\showeprint[arXiv]{2303.02723}


\bibitem[Gottlob et~al\mbox{.}(2001)]%
        {DBLP:journals/jacm/GottlobLS01}
\bibfield{author}{\bibinfo{person}{Georg Gottlob}, \bibinfo{person}{Nicola
  Leone}, {and} \bibinfo{person}{Francesco Scarcello}.}
  \bibinfo{year}{2001}\natexlab{}.
\newblock \showarticletitle{The complexity of acyclic conjunctive queries}.
\newblock \bibinfo{journal}{\emph{J. {ACM}}} \bibinfo{volume}{48},
  \bibinfo{number}{3} (\bibinfo{year}{2001}), \bibinfo{pages}{431--498}.
\newblock
\urldef\tempurl%
\url{https://doi.org/10.1145/382780.382783}
\showDOI{\tempurl}


\bibitem[Gottlob et~al\mbox{.}(2002)]%
        {DBLP:journals/jcss/GottlobLS02}
\bibfield{author}{\bibinfo{person}{Georg Gottlob}, \bibinfo{person}{Nicola
  Leone}, {and} \bibinfo{person}{Francesco Scarcello}.}
  \bibinfo{year}{2002}\natexlab{}.
\newblock \showarticletitle{Hypertree Decompositions and Tractable Queries}.
\newblock \bibinfo{journal}{\emph{J. Comput. Syst. Sci.}} \bibinfo{volume}{64},
  \bibinfo{number}{3} (\bibinfo{year}{2002}), \bibinfo{pages}{579--627}.
\newblock
\urldef\tempurl%
\url{https://doi.org/10.1006/jcss.2001.1809}
\showDOI{\tempurl}


\bibitem[Graefe and McKenna(1993)]%
        {DBLP:conf/icde/GraefeM93}
\bibfield{author}{\bibinfo{person}{Goetz Graefe} {and}
  \bibinfo{person}{William~J. McKenna}.} \bibinfo{year}{1993}\natexlab{}.
\newblock \showarticletitle{The Volcano Optimizer Generator: Extensibility and
  Efficient Search}. In \bibinfo{booktitle}{\emph{Proceedings of the Ninth
  International Conference on Data Engineering, April 19-23, 1993, Vienna,
  Austria}}. \bibinfo{publisher}{{IEEE} Computer Society},
  \bibinfo{pages}{209--218}.
\newblock
\urldef\tempurl%
\url{https://doi.org/10.1109/ICDE.1993.344061}
\showDOI{\tempurl}


\bibitem[Graham(1979)]%
        {report/toronto/Gra79}
\bibfield{author}{\bibinfo{person}{Marc~H. Graham}.}
  \bibinfo{year}{1979}\natexlab{}.
\newblock \bibinfo{booktitle}{\emph{{On The Universal Relation}}}.
\newblock \bibinfo{type}{{T}echnical {R}eport}.
  \bibinfo{institution}{University of Toronto}.
\newblock


\bibitem[Grohe(2007)]%
        {DBLP:journals/jacm/Grohe07}
\bibfield{author}{\bibinfo{person}{Martin Grohe}.}
  \bibinfo{year}{2007}\natexlab{}.
\newblock \showarticletitle{The complexity of homomorphism and constraint
  satisfaction problems seen from the other side}.
\newblock \bibinfo{journal}{\emph{J. {ACM}}} \bibinfo{volume}{54},
  \bibinfo{number}{1} (\bibinfo{year}{2007}), \bibinfo{pages}{1:1--1:24}.
\newblock
\urldef\tempurl%
\url{https://doi.org/10.1145/1206035.1206036}
\showDOI{\tempurl}


\bibitem[Grohe and Marx(2014)]%
        {2014grohemarx}
\bibfield{author}{\bibinfo{person}{Martin Grohe} {and}
  \bibinfo{person}{D{\'{a}}niel Marx}.} \bibinfo{year}{2014}\natexlab{}.
\newblock \showarticletitle{Constraint Solving via Fractional Edge Covers}.
\newblock \bibinfo{journal}{\emph{{ACM} Trans. Algorithms}}
  \bibinfo{volume}{11}, \bibinfo{number}{1} (\bibinfo{year}{2014}),
  \bibinfo{pages}{4:1--4:20}.
\newblock


\bibitem[Hu and Wang(2023)]%
        {DBLP:journals/corr/abs-2302-13140}
\bibfield{author}{\bibinfo{person}{Xiao Hu} {and} \bibinfo{person}{Qichen
  Wang}.} \bibinfo{year}{2023}\natexlab{}.
\newblock \showarticletitle{Computing the Difference of Conjunctive Queries
  Efficiently}.
\newblock \bibinfo{journal}{\emph{CoRR}}  \bibinfo{volume}{abs/2302.13140}
  (\bibinfo{year}{2023}).
\newblock
\urldef\tempurl%
\url{https://doi.org/10.48550/arXiv.2302.13140}
\showDOI{\tempurl}
\showeprint[arXiv]{2302.13140}


\bibitem[Idris et~al\mbox{.}(2017)]%
        {DBLP:conf/sigmod/IdrisUV17}
\bibfield{author}{\bibinfo{person}{Muhammad Idris},
  \bibinfo{person}{Mart{\'{\i}}n Ugarte}, {and} \bibinfo{person}{Stijn
  Vansummeren}.} \bibinfo{year}{2017}\natexlab{}.
\newblock \showarticletitle{The Dynamic Yannakakis Algorithm: Compact and
  Efficient Query Processing Under Updates}. In
  \bibinfo{booktitle}{\emph{Proceedings of the 2017 {ACM} International
  Conference on Management of Data, {SIGMOD} Conference 2017, Chicago, IL, USA,
  May 14-19, 2017}}, \bibfield{editor}{\bibinfo{person}{Semih Salihoglu},
  \bibinfo{person}{Wenchao Zhou}, \bibinfo{person}{Rada Chirkova},
  \bibinfo{person}{Jun Yang}, {and} \bibinfo{person}{Dan Suciu}} (Eds.).
  \bibinfo{publisher}{{ACM}}, \bibinfo{pages}{1259--1274}.
\newblock
\urldef\tempurl%
\url{https://doi.org/10.1145/3035918.3064027}
\showDOI{\tempurl}


\bibitem[Idris et~al\mbox{.}(2020)]%
        {DBLP:journals/vldb/IdrisUVVL20}
\bibfield{author}{\bibinfo{person}{Muhammad Idris},
  \bibinfo{person}{Mart{\'{\i}}n Ugarte}, \bibinfo{person}{Stijn Vansummeren},
  \bibinfo{person}{Hannes Voigt}, {and} \bibinfo{person}{Wolfgang Lehner}.}
  \bibinfo{year}{2020}\natexlab{}.
\newblock \showarticletitle{General dynamic Yannakakis: conjunctive queries
  with theta joins under updates}.
\newblock \bibinfo{journal}{\emph{{VLDB} J.}} \bibinfo{volume}{29},
  \bibinfo{number}{2-3} (\bibinfo{year}{2020}), \bibinfo{pages}{619--653}.
\newblock
\urldef\tempurl%
\url{https://doi.org/10.1007/s00778-019-00590-9}
\showDOI{\tempurl}


\bibitem[Leis et~al\mbox{.}(2018)]%
        {DBLP:journals/vldb/LeisRGMBKN18}
\bibfield{author}{\bibinfo{person}{Viktor Leis}, \bibinfo{person}{Bernhard
  Radke}, \bibinfo{person}{Andrey Gubichev}, \bibinfo{person}{Atanas Mirchev},
  \bibinfo{person}{Peter~A. Boncz}, \bibinfo{person}{Alfons Kemper}, {and}
  \bibinfo{person}{Thomas Neumann}.} \bibinfo{year}{2018}\natexlab{}.
\newblock \showarticletitle{Query optimization through the looking glass, and
  what we found running the Join Order Benchmark}.
\newblock \bibinfo{journal}{\emph{{VLDB} J.}} \bibinfo{volume}{27},
  \bibinfo{number}{5} (\bibinfo{year}{2018}), \bibinfo{pages}{643--668}.
\newblock
\urldef\tempurl%
\url{https://doi.org/10.1007/s00778-017-0480-7}
\showDOI{\tempurl}


\bibitem[Lutz and Przybylko(2022)]%
        {DBLP:conf/pods/LutzP22}
\bibfield{author}{\bibinfo{person}{Carsten Lutz} {and} \bibinfo{person}{Marcin
  Przybylko}.} \bibinfo{year}{2022}\natexlab{}.
\newblock \showarticletitle{Efficiently Enumerating Answers to
  Ontology-Mediated Queries}. In \bibinfo{booktitle}{\emph{{PODS} '22:
  International Conference on Management of Data, Philadelphia, PA, USA, June
  12 - 17, 2022}}, \bibfield{editor}{\bibinfo{person}{Leonid Libkin} {and}
  \bibinfo{person}{Pablo Barcel{\'{o}}}} (Eds.). \bibinfo{publisher}{{ACM}},
  \bibinfo{pages}{277--289}.
\newblock
\urldef\tempurl%
\url{https://doi.org/10.1145/3517804.3524166}
\showDOI{\tempurl}


\bibitem[Mageirakos et~al\mbox{.}(2022)]%
        {DBLP:conf/icde/MageirakosMKCA22}
\bibfield{author}{\bibinfo{person}{Vasilis Mageirakos},
  \bibinfo{person}{Riccardo Mancini}, \bibinfo{person}{Srinivas Karthik},
  \bibinfo{person}{Bikash Chandra}, {and} \bibinfo{person}{Anastasia
  Ailamaki}.} \bibinfo{year}{2022}\natexlab{}.
\newblock \showarticletitle{Efficient GPU-accelerated Join Optimization for
  Complex Queries}. In \bibinfo{booktitle}{\emph{38th {IEEE} International
  Conference on Data Engineering, {ICDE} 2022, Kuala Lumpur, Malaysia, May
  9-12, 2022}}. \bibinfo{publisher}{{IEEE}}, \bibinfo{pages}{3190--3193}.
\newblock
\urldef\tempurl%
\url{https://doi.org/10.1109/ICDE53745.2022.00295}
\showDOI{\tempurl}


\bibitem[Mancini et~al\mbox{.}(2022)]%
        {DBLP:conf/sigmod/ManciniKCMA22}
\bibfield{author}{\bibinfo{person}{Riccardo Mancini}, \bibinfo{person}{Srinivas
  Karthik}, \bibinfo{person}{Bikash Chandra}, \bibinfo{person}{Vasilis
  Mageirakos}, {and} \bibinfo{person}{Anastasia Ailamaki}.}
  \bibinfo{year}{2022}\natexlab{}.
\newblock \showarticletitle{Efficient {Massively} {Parallel} {Join}
  {Optimization} for {Large} {Queries}}. In
  \bibinfo{booktitle}{\emph{Proceedings of the 2022 {ACM} {SIGMOD}
  International Conference on Management of Data, {SIGMOD} Conference 2022}}.
  \bibinfo{publisher}{{ACM}}, \bibinfo{pages}{122--135}.
\newblock
\urldef\tempurl%
\url{https://doi.org/10.1145/3514221.3517871}
\showDOI{\tempurl}


\bibitem[Neumann and Radke(2018)]%
        {DBLP:conf/sigmod/NeumannR18}
\bibfield{author}{\bibinfo{person}{Thomas Neumann} {and}
  \bibinfo{person}{Bernhard Radke}.} \bibinfo{year}{2018}\natexlab{}.
\newblock \showarticletitle{Adaptive Optimization of Very Large Join Queries}.
  In \bibinfo{booktitle}{\emph{Proceedings of the 2018 International Conference
  on Management of Data, {SIGMOD} Conference 2018, Houston, TX, USA, June
  10-15, 2018}}, \bibfield{editor}{\bibinfo{person}{Gautam Das},
  \bibinfo{person}{Christopher~M. Jermaine}, {and} \bibinfo{person}{Philip~A.
  Bernstein}} (Eds.). \bibinfo{publisher}{{ACM}}, \bibinfo{pages}{677--692}.
\newblock
\urldef\tempurl%
\url{https://doi.org/10.1145/3183713.3183733}
\showDOI{\tempurl}


\bibitem[Ngo et~al\mbox{.}(2018)]%
        {DBLP:journals/jacm/NgoPRR18}
\bibfield{author}{\bibinfo{person}{Hung~Q. Ngo}, \bibinfo{person}{Ely Porat},
  \bibinfo{person}{Christopher R{\'{e}}}, {and} \bibinfo{person}{Atri Rudra}.}
  \bibinfo{year}{2018}\natexlab{}.
\newblock \showarticletitle{Worst-case Optimal Join Algorithms}.
\newblock \bibinfo{journal}{\emph{J. {ACM}}} \bibinfo{volume}{65},
  \bibinfo{number}{3} (\bibinfo{year}{2018}), \bibinfo{pages}{16:1--16:40}.
\newblock
\urldef\tempurl%
\url{https://doi.org/10.1145/3180143}
\showDOI{\tempurl}


\bibitem[Ngo et~al\mbox{.}(2013)]%
        {DBLP:journals/sigmod/NgoRR13}
\bibfield{author}{\bibinfo{person}{Hung~Q. Ngo}, \bibinfo{person}{Christopher
  R{\'{e}}}, {and} \bibinfo{person}{Atri Rudra}.}
  \bibinfo{year}{2013}\natexlab{}.
\newblock \showarticletitle{Skew strikes back: new developments in the theory
  of join algorithms}.
\newblock \bibinfo{journal}{\emph{{SIGMOD} Rec.}} \bibinfo{volume}{42},
  \bibinfo{number}{4} (\bibinfo{year}{2013}), \bibinfo{pages}{5--16}.
\newblock
\urldef\tempurl%
\url{https://doi.org/10.1145/2590989.2590991}
\showDOI{\tempurl}


\bibitem[Perelman and R{\'{e}}(2015)]%
        {DBLP:conf/sigmod/PerelmanR15}
\bibfield{author}{\bibinfo{person}{Adam Perelman} {and}
  \bibinfo{person}{Christopher R{\'{e}}}.} \bibinfo{year}{2015}\natexlab{}.
\newblock \showarticletitle{DunceCap: Compiling Worst-Case Optimal Query
  Plans}. In \bibinfo{booktitle}{\emph{Proceedings of the 2015 {ACM} {SIGMOD}
  International Conference on Management of Data, Melbourne, Victoria,
  Australia, May 31 - June 4, 2015}},
  \bibfield{editor}{\bibinfo{person}{Timos~K. Sellis},
  \bibinfo{person}{Susan~B. Davidson}, {and} \bibinfo{person}{Zachary~G. Ives}}
  (Eds.). \bibinfo{publisher}{{ACM}}, \bibinfo{pages}{2075--2076}.
\newblock
\urldef\tempurl%
\url{https://doi.org/10.1145/2723372.2764945}
\showDOI{\tempurl}


\bibitem[Pichler and Skritek(2013)]%
        {DBLP:journals/jcss/PichlerS13}
\bibfield{author}{\bibinfo{person}{Reinhard Pichler} {and}
  \bibinfo{person}{Sebastian Skritek}.} \bibinfo{year}{2013}\natexlab{}.
\newblock \showarticletitle{Tractable counting of the answers to conjunctive
  queries}.
\newblock \bibinfo{journal}{\emph{J. Comput. Syst. Sci.}} \bibinfo{volume}{79},
  \bibinfo{number}{6} (\bibinfo{year}{2013}), \bibinfo{pages}{984--1001}.
\newblock
\urldef\tempurl%
\url{https://doi.org/10.1016/j.jcss.2013.01.012}
\showDOI{\tempurl}


\bibitem[Raasveldt and M{\"{u}}hleisen(2019)]%
        {DBLP:conf/sigmod/RaasveldtM19}
\bibfield{author}{\bibinfo{person}{Mark Raasveldt} {and}
  \bibinfo{person}{Hannes M{\"{u}}hleisen}.} \bibinfo{year}{2019}\natexlab{}.
\newblock \showarticletitle{{DuckDB}: an {E}mbeddable {A}nalytical {D}atabase}.
  In \bibinfo{booktitle}{\emph{Proceedings of the 2019 International Conference
  on Management of Data, {SIGMOD} Conference 2019}}.
  \bibinfo{publisher}{{ACM}}, \bibinfo{pages}{1981--1984}.
\newblock


\bibitem[Scarcello et~al\mbox{.}(2004)]%
        {DBLP:conf/pods/ScarcelloGL04}
\bibfield{author}{\bibinfo{person}{Francesco Scarcello},
  \bibinfo{person}{Gianluigi Greco}, {and} \bibinfo{person}{Nicola Leone}.}
  \bibinfo{year}{2004}\natexlab{}.
\newblock \showarticletitle{Weighted Hypertree Decompositions and Optimal Query
  Plans}. In \bibinfo{booktitle}{\emph{Proceedings of the Twenty-third {ACM}
  {SIGACT-SIGMOD-SIGART} Symposium on Principles of Database Systems, June
  14-16, 2004, Paris, France}}, \bibfield{editor}{\bibinfo{person}{Catriel
  Beeri} {and} \bibinfo{person}{Alin Deutsch}} (Eds.).
  \bibinfo{publisher}{{ACM}}, \bibinfo{pages}{210--221}.
\newblock
\urldef\tempurl%
\url{https://doi.org/10.1145/1055558.1055587}
\showDOI{\tempurl}


\bibitem[Seshadri et~al\mbox{.}(1996)]%
        {DBLP:conf/icde/SeshadriPL96}
\bibfield{author}{\bibinfo{person}{Praveen Seshadri}, \bibinfo{person}{Hamid
  Pirahesh}, {and} \bibinfo{person}{T.~Y.~Cliff Leung}.}
  \bibinfo{year}{1996}\natexlab{}.
\newblock \showarticletitle{Complex Query Decorrelation}. In
  \bibinfo{booktitle}{\emph{Proc. {ICDE}'96}}. \bibinfo{publisher}{{IEEE}
  Computer Society}, \bibinfo{pages}{450--458}.
\newblock
\urldef\tempurl%
\url{https://doi.org/10.1109/ICDE.1996.492194}
\showDOI{\tempurl}


\bibitem[Stonebraker and Kemnitz(1991)]%
        {DBLP:journals/cacm/StonebrakerK91}
\bibfield{author}{\bibinfo{person}{Michael Stonebraker} {and}
  \bibinfo{person}{Greg Kemnitz}.} \bibinfo{year}{1991}\natexlab{}.
\newblock \showarticletitle{The Postgres Next Generation Database Management
  System}.
\newblock \bibinfo{journal}{\emph{Commun. {ACM}}} \bibinfo{volume}{34},
  \bibinfo{number}{10} (\bibinfo{year}{1991}), \bibinfo{pages}{78--92}.
\newblock
\urldef\tempurl%
\url{https://doi.org/10.1145/125223.125262}
\showDOI{\tempurl}


\bibitem[Tu and R{\'{e}}(2015)]%
        {DBLP:conf/sigmod/TuR15}
\bibfield{author}{\bibinfo{person}{Susan Tu} {and} \bibinfo{person}{Christopher
  R{\'{e}}}.} \bibinfo{year}{2015}\natexlab{}.
\newblock \showarticletitle{DunceCap: Query Plans Using Generalized Hypertree
  Decompositions}. In \bibinfo{booktitle}{\emph{Proceedings of the 2015 {ACM}
  {SIGMOD} International Conference on Management of Data, Melbourne, Victoria,
  Australia, May 31 - June 4, 2015}},
  \bibfield{editor}{\bibinfo{person}{Timos~K. Sellis},
  \bibinfo{person}{Susan~B. Davidson}, {and} \bibinfo{person}{Zachary~G. Ives}}
  (Eds.). \bibinfo{publisher}{{ACM}}, \bibinfo{pages}{2077--2078}.
\newblock
\urldef\tempurl%
\url{https://doi.org/10.1145/2723372.2764946}
\showDOI{\tempurl}


\bibitem[Wang et~al\mbox{.}(2023)]%
        {DBLP:journals/corr/abs-2301-04003}
\bibfield{author}{\bibinfo{person}{Qichen Wang}, \bibinfo{person}{Xiao Hu},
  \bibinfo{person}{Binyang Dai}, {and} \bibinfo{person}{Ke Yi}.}
  \bibinfo{year}{2023}\natexlab{}.
\newblock \showarticletitle{Change Propagation Without Joins}.
\newblock \bibinfo{journal}{\emph{CoRR}}  \bibinfo{volume}{abs/2301.04003}
  (\bibinfo{year}{2023}).
\newblock
\urldef\tempurl%
\url{https://doi.org/10.48550/arXiv.2301.04003}
\showDOI{\tempurl}
\showeprint[arXiv]{2301.04003}


\bibitem[Wang and Yi(2020)]%
        {DBLP:conf/sigmod/WangY20}
\bibfield{author}{\bibinfo{person}{Qichen Wang} {and} \bibinfo{person}{Ke Yi}.}
  \bibinfo{year}{2020}\natexlab{}.
\newblock \showarticletitle{Maintaining Acyclic Foreign-Key Joins under
  Updates}. In \bibinfo{booktitle}{\emph{Proceedings of the 2020 International
  Conference on Management of Data, {SIGMOD} Conference 2020, online conference
  [Portland, OR, USA], June 14-19, 2020}},
  \bibfield{editor}{\bibinfo{person}{David Maier}, \bibinfo{person}{Rachel
  Pottinger}, \bibinfo{person}{AnHai Doan}, \bibinfo{person}{Wang{-}Chiew Tan},
  \bibinfo{person}{Abdussalam Alawini}, {and} \bibinfo{person}{Hung~Q. Ngo}}
  (Eds.). \bibinfo{publisher}{{ACM}}, \bibinfo{pages}{1225--1239}.
\newblock
\urldef\tempurl%
\url{https://doi.org/10.1145/3318464.3380586}
\showDOI{\tempurl}


\bibitem[Wang and Yi(2022)]%
        {DBLP:conf/sigmod/0001022}
\bibfield{author}{\bibinfo{person}{Qichen Wang} {and} \bibinfo{person}{Ke Yi}.}
  \bibinfo{year}{2022}\natexlab{}.
\newblock \showarticletitle{Conjunctive Queries with Comparisons}. In
  \bibinfo{booktitle}{\emph{{SIGMOD} '22: International Conference on
  Management of Data, Philadelphia, PA, USA, June 12 - 17, 2022}},
  \bibfield{editor}{\bibinfo{person}{Zachary~G. Ives}, \bibinfo{person}{Angela
  Bonifati}, {and} \bibinfo{person}{Amr~El Abbadi}} (Eds.).
  \bibinfo{publisher}{{ACM}}, \bibinfo{pages}{108--121}.
\newblock
\urldef\tempurl%
\url{https://doi.org/10.1145/3514221.3517830}
\showDOI{\tempurl}


\bibitem[Wang and Yi(2021)]%
        {DBLP:conf/sigmod/Wang021}
\bibfield{author}{\bibinfo{person}{Yilei Wang} {and} \bibinfo{person}{Ke Yi}.}
  \bibinfo{year}{2021}\natexlab{}.
\newblock \showarticletitle{Secure Yannakakis: Join-Aggregate Queries over
  Private Data}. In \bibinfo{booktitle}{\emph{{SIGMOD} '21: International
  Conference on Management of Data, Virtual Event, China, June 20-25, 2021}},
  \bibfield{editor}{\bibinfo{person}{Guoliang Li}, \bibinfo{person}{Zhanhuai
  Li}, \bibinfo{person}{Stratos Idreos}, {and} \bibinfo{person}{Divesh
  Srivastava}} (Eds.). \bibinfo{publisher}{{ACM}}, \bibinfo{pages}{1969--1981}.
\newblock
\urldef\tempurl%
\url{https://doi.org/10.1145/3448016.3452808}
\showDOI{\tempurl}


\bibitem[Yannakakis(1981)]%
        {DBLP:conf/vldb/Yannakakis81}
\bibfield{author}{\bibinfo{person}{Mihalis Yannakakis}.}
  \bibinfo{year}{1981}\natexlab{}.
\newblock \showarticletitle{Algorithms for Acyclic Database Schemes}. In
  \bibinfo{booktitle}{\emph{Proceedings of the 7th International Conference on
  Very Large Databases, {VLDB} 1981, Cannes}}. \bibinfo{publisher}{VLDB},
  \bibinfo{pages}{82--94}.
\newblock


\bibitem[Yu and Özsoyoğlu(1979)]%
        {DBLP:conf/compsac/YuO79}
\bibfield{author}{\bibinfo{person}{C.~T. Yu} {and} \bibinfo{person}{M.~Z.
  Özsoyoğlu}.} \bibinfo{year}{1979}\natexlab{}.
\newblock \showarticletitle{An algorithm for tree-query membership of a
  distributed query}. In \bibinfo{booktitle}{\emph{The {IEEE} Computer
  Society's Third International Computer Software and Applications Conference,
  {COMPSAC} 1979}}. \bibinfo{pages}{306--312}.
\newblock


\bibitem[Zaharia et~al\mbox{.}(2016)]%
        {DBLP:journals/cacm/ZahariaXWDADMRV16}
\bibfield{author}{\bibinfo{person}{Matei Zaharia}, \bibinfo{person}{Reynold~S.
  Xin}, \bibinfo{person}{Patrick Wendell}, \bibinfo{person}{Tathagata Das},
  \bibinfo{person}{Michael Armbrust}, \bibinfo{person}{Ankur Dave},
  \bibinfo{person}{Xiangrui Meng}, \bibinfo{person}{Josh Rosen},
  \bibinfo{person}{Shivaram Venkataraman}, \bibinfo{person}{Michael~J.
  Franklin}, \bibinfo{person}{Ali Ghodsi}, \bibinfo{person}{Joseph Gonzalez},
  \bibinfo{person}{Scott Shenker}, {and} \bibinfo{person}{Ion Stoica}.}
  \bibinfo{year}{2016}\natexlab{}.
\newblock \showarticletitle{Apache Spark: a unified engine for big data
  processing}.
\newblock \bibinfo{journal}{\emph{Commun. {ACM}}} \bibinfo{volume}{59},
  \bibinfo{number}{11} (\bibinfo{year}{2016}), \bibinfo{pages}{56--65}.
\newblock
\urldef\tempurl%
\url{https://doi.org/10.1145/2934664}
\showDOI{\tempurl}


\end{thebibliography}

\ifArxivVersion
\clearpage
\appendix

\section{Implementation Details of \ourSystem}
\label{sect:Implementation}

\ifArxivVersion
In this section, we provide some implementation details of \ourSystem.
\else
We round off this paper by providing some implementation details of \ourSystem.
\fi
As was mentioned in Section~\ref{sec:expeval}, the
rewriting-based approach of our \ourSystem system proceeds in several 
steps: 

\begin{itemize}
    \item extraction of the CQ from the SQL query
    \item transformation into a hypergraph
    \item join tree computation
    \item SQL statement generation
\end{itemize}

The queries in the benchmark of 
\cite{DBLP:conf/sigmod/ManciniKCMA22}
are all straightforward SELECT-PROJECT-JOIN queries
(in particular, no GROUP BY and HAVING clauses, no subqueries).
We process these queries via a simplified version of 
the SQL-to-CQ translation from~\cite{DBLP:journals/jea/FischlGLP21}, 
which also provides the further translation of the CQ into a hypergraph.
Recall that the hypergraph $H= (V,E)$ of 
a CQ $Q$ is obtained 
by identifying the vertices in $V$ with the variables in $Q$ and
defining as edges in $E$ those sets of vertices where the corresponding
variables occur jointly in an atom of $Q$.
The join tree computation and the generation of SQL statements 
are discussed below in~more~detail.

\subsection{Join Tree Computation}
\label{sect:JoinTree}

The  GYO algorithm
\cite{report/toronto/Gra79,DBLP:conf/compsac/YuO79} 
for deciding whether a hypergraph (and thus the corresponding query) is acyclic works by non-deter\-mi\-nistic
application of the following steps: 
i)
deleting a vertex with degree 1 (i.e., a vertex occurring in a single edge),
ii) deleting an empty edge, or iii)
deleting an edge that is a subset of another edge.
In Algorithm~\ref{alg}, we choose a particular order 
in which the elimination steps of the GYO-algorithm are executed. 
Technically, deletion of degree 1 vertices from an edge $e$ of $H$ may produce a new edge that is not part of the join tree. We thus use $\lbl(e)$ in Algorithm~\ref{alg} to always refer to the name of the original edge before
vertex removals.
The algorithm produces join trees  with a particular property expressed in the following theorem: 

\begin{algorithm}[t]
  \SetKwInOut{Input}{input}\SetKwInOut{Output}{output}
  \SetKw{Reject}{Reject}
  \SetKw{Continue}{Continue}

  \Input{A connected $\alpha$-acyclic hypergraph $H$}
  \Output{A join tree of $H$}

   $J \leftarrow $ empty  tree\;
  \While{$H$ contains more than 1 edge}{
    Delete all degree 1 vertices from $H$\;
    \For{$e \in E(H)$ s.t. there is no $f \in E(H)$ with $e \subset f$}{ \label{condmax}
      $C_e \leftarrow \{ c \in E(H) \mid c \subseteq e \}$\;
      \For{$c \in C_e$}{
          Set $\lbl(c)$ as child of $\lbl(e)$ in $J$\;
          Remove $c$ from $H$\; \label{line:deledge}
        }
    }
  }
  \Return $J$\;
  \caption{The Flat-GYO algorithm}
  \label{alg}
\end{algorithm}

\begin{theorem}
Let $H = (V(H),E(H))$ be an acyclic hypergraph and let 
$T$ denote the join tree resulting from applying Algorithm~\ref{alg} 
to $H$. Then $T$ has minimal depth among all join trees of $H$.
\end{theorem}

\begin{proof}%
The proof proceeds in three steps:
(1) First, 
we observe that there is still some non-determinism left in 
Algorithm~\ref{alg}, that depends on the order in 
which the edges in the for-loop on line 4 are processed. It may happen (i) that $e = e'$ holds for two edges with 
$\lbl(e) \neq \lbl(e')$ and that (ii) 
for two distinct maximal edges $e, e'$, an edge $c \in E(H)$
satisfies both $c \subseteq e$ and $c \subseteq e'$ on line 5. 
Nevertheless, the number of iterations of the while-loop is
independent of the order in which the maximal edges are processed 
in the for-loop.
This property follows from the easily verifiable fact that the 
set of edges $\{ e_{i_1}, \dots, e_{i_m} \}$ resulting 
from an iteration of the while-loop is independent
of this non-determinism, even though (due to (i)) there 
may be an alternative set of edges 
with different labels and (due to (ii)) also an alternative collection of parent/child relationships may be  possible.

(2) Second, if a run of Algorithm~\ref{alg} has $k$ iterations of the while-loop, then the join tree constructed by this run has at most depth $k$ (max.~distance from root to leaf). This is due to the fact that, on line 7, existing partially constructed trees may be appended below a new root node but no further nesting may happen here. Hence, the depth of the partially constructed trees grows by at most 1.

(3) Finally, if there exists a join tree $T$ of depth $k$, then 
there exists a run of Algorithm~\ref{alg} with at most $k$ iterations of the while loop. This property is proved by a simple induction argument: there exists an order in which the maximal edges are processed in the for-loop, so that 
all leaf nodes of $T$ get removed on line 8 -- thus decreasing the depth of $T$ by at least 1.

The theorem can then be proved as follows: suppose that, for a given hypergraph $H$, the 
minimum depth of any join tree of $H$ is $k$. Then there exists a join tree $T$ of depth $k$. Hence, 
by (3), Algorithm~\ref{alg} has a run with at most $k$ iterations of the while-loop and, therefore, by (1), any run of Algorithm~\ref{alg} has a run with at most $k$ iterations of the while-loop.
Thus, by (2), 
any run of Algorithm~\ref{alg} 
produces a 
join tree of depth at most $k$. 
\end{proof}

\subsection{Query Plan Generation and Execution}
\label{sect:QueryPlan}

In a final step, we create a sequence of SQL statements that express the execution of \YA over the join tree and reintroduce final projection and aggregation if applicable.
The overall evaluation of the query is thus split into four stages, 
which we briefly describe below. We will illustrate these steps
by means of the SQL query given in the following example.

\begin{example}
\label{ex:unischema:two}
Recall the university schema of 
Example~\ref{ex:sca}
with relations 
$\mathsf{exams}(\mathrm{cid}$, 
$\mathrm{student}$, $\mathrm{grade})$
and
$\mathsf{courses}(\mathrm{cid}, \mathrm{faculty})$. 
We now 
add the two relations
$\mathsf{tutors}(\mathrm{student},\mathrm{cid},\mathrm{num\_semesters})$
and
$\mathsf{enrolled}(\mathrm{student},\mathrm{program})$.
The following query retrieves, for each fixed pair of program and course, 
the lowest grade obtained in exams of the CS faculty by any student 
enrolled in that program and who has been tutored for more than 1 semester in that course.
  {\small
  \begin{lstlisting}[language=SQL]
SELECT enrolled.program, exams.cid, 
       MIN(exams.grade)
FROM exams, courses, enrolled, tutors
WHERE exams.cid = courses.cid
  AND exams.student = enrolled.student
  AND exams.cid = tutors.cid
  AND courses.faculty = 'ComputerScience'
  AND exams.student = tutors.student
  AND tutors.num_semesters > 1
GROUP BY enrolled.program, exams.cid;
\end{lstlisting}
}

\noindent
The query is acyclic but not 0MA (it is not guarded).
\nop{******************************
This is due to the fact that  attributes from two relations $\mathsf{exams}$ and $\mathsf{enrolled}$
appear in the SELECT clause.
******************************}
Its hypergraph and a 
possible join tree are depicted in Figure~\ref{fig:joinTreeExampleUni}, 
where, 
for the sake of readability,
the names of 
vertices are abbreviated to the first character.
\nop{**************************************
We thus have to carry out also the top-down traversal and the second bottom-up traversal. 
But they can be significantly restricted -- namely to the root node and its child node, which together 
cover all attributes in the SELECT clause. That is, 
in the top-down traversal, it suffices to semi-join the $\mathsf{enrolled}$ relation into $\mathsf{exams}$; 
and also the joins of the second bottom-up traversal can be restricted to these two relations.
The relations at the two leaf nodes of the join tree can be completely ignored in these two traversals.
**************************************}
\hfill $\diamond$
\end{example}

\begin{figure}[t]

\begin{tabular}{cc}
\begin{tikzpicture}
	\node	(p)	at	(1,2)	{p};
	\node	(c)	at	(0,1)	{c};
	\node	(s)	at	(1,1)	{s};
	\node	(g)	at	(2,1)	{g};
	\node	(f)	at	(0,0)	{f};
	\node	(n)	at	(1,0)	{n};
	\draw[rotate around={90:(1,1.5)}] (1,1.5) ellipse [x radius=0.8, y radius=0.3]; %
	\node[scale=.9]   (enr)   at  (2,2.3)  {enrolled(s,p)};
	\draw (s) ellipse [x radius=1.5, y radius=0.4]; %
	\node[scale=.9]   (ex)   at  (3.2,1.3)  {exams(c,s,g)};
	\draw[rotate around={90:(0,.5)}] (0,.5) ellipse [x radius=0.85, y radius=0.3]; %
	\node[scale=.9]   (crs)   at  (-1,-.3)  {courses(c,f)};
	\draw[rounded corners] ($(c.north west) + (0,.05)$)
	    -- ($(s.north east) + (.05,.05)$)
	    -- ($(n.south east) + (.05,0)$)
	    -- ($(n) + (0,-.37)$)
	    -- ($(n.south west) + (-.03,0)$)
	    -- ($(s.south west) + (-.05,-.05)$)
	    -- ($(c.south west) + (0,-.03)$)
	    -- ($(c) + (-.37,0)$)
		-- cycle; %
	\node[scale=.9]   (tut)   at  (2,-.3)  {tutors(c,s,n)};
\end{tikzpicture} & \hspace{-3em}
\tikzset{
  my node style/.style={
    font=\small,
    top color=white,
    bottom color=white,
    rectangle,
    minimum size=10mm,
    draw=black,
    thick,
    drop shadow,
    align=center,
  }
}
\forestset{
  my tree style/.style={
    for tree={
      parent anchor=south,
      child anchor=north,
      l sep-=0pt,
      my node style,
      edge={draw=black, thick},
      edge path={
        \noexpand\path [draw, \forestoption{edge}] (!u.parent anchor) -- +(0,-7.5pt) -| (.child anchor)\forestoption{edge label};
      },
      if n children=3{
        for children={
          if n=2{calign with current}{}
        }
      }{},
      delay={if content={}{shape=coordinate}{}}
    }
  }
}

\centering
\begin{forest}
          my tree style
          [\QueryPlanNodePlain{\color{black}enrolled}
                [\QueryPlanNodePlain{\color{black}exams}
                     [\QueryPlanNodePlain{courses}]
                            [\QueryPlanNodePlain{tutors}]
          ]]
\end{forest}
\end{tabular}

\caption{Hypergraph and join tree for Example \ref{ex:unischema:two}}        
\label{fig:joinTreeExampleUni}
\end{figure}
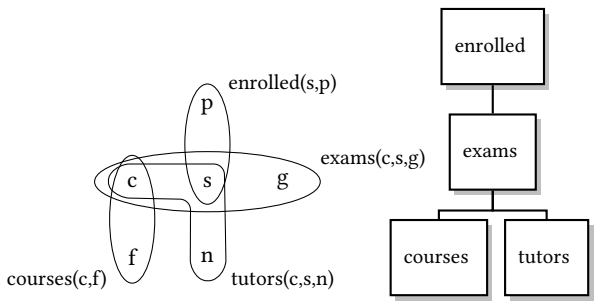

\paragraph{The Setup Stage}
We first rename the attributes in such a way that all equi-joins are replaced by natural joins throughout the rest of the process. Then, from the join tree perspective, we create one view per node, representing the relation in the join tree before the execution of \YA.
Early projection to the attributes which are actually used in the query 
(either as a join attribute or as part of the final result) as well as applicable selections are also incorporated directly into these views. 
For instance, for the query and join tree from Example~\ref{ex:unischema:two}, 
the leaf node for relation \textsf{courses} induces 
the following view \texttt{courses\_setup}:
{\small
\begin{lstlisting}[language=SQL]
  CREATE VIEW courses_setup AS SELECT cid
  FROM courses WHERE faculty='ComputerScience';
\end{lstlisting}
}

\nop{*************************
Overall, this step is only a matter of convenience and the views could be instead rolled into the following steps. However, the clean separation of stages here makes the later extension to cyclic queries and GHDs clearer and simpler.
*************************}

\paragraph{The Semi-Join Stages}
The views from the setup stage are used to generate
SQL statements for the semi-joins of the first bottom-up traversal and,
if the query does not satisfy the 0MA-property, 
also for the top-down traversal of the join tree. 
\nop{**************************
We refer to the respective generated statements as the $\ltimes$-up and $\ltimes$-down stage, respectively. 
**************************}
The result of each semi-join is stored in 
an auxiliary temporary table.
Semi-joins are expressed in the standard manner via the 
IN operator of SQL.

To illustrate the semi-join stages,  we continue our example from above. Assuming that all views from the  setup stage are named with the \texttt{\_setup} suffix, the first semi-joins of the bottom-up traversal are 
realised in SQL as follows (for clarity, the previously mentioned renaming of attributes is not performed here):
{\small
\begin{lstlisting}[language=SQL]
  CREATE TEMP TABLE exams_sjup AS
  SELECT * FROM exams_setup WHERE
  cid IN (SELECT cid from courses_setup) AND
  cid, student IN (SELECT cid, student 
                   FROM tutors_setup);
\end{lstlisting}
}
We thus create a new intermediate relation for the \textsf{exams} node. Importantly, the analogous statement expressing the semi-join from the \texttt{exams} node into the \texttt{enrolled} node will now make use of \texttt{exams\_sjup} rather than the setup view for the \texttt{exams} node.

\nop{*************
From the join tree in Figure \ref{fig:join_tree_example_reduced}, we gain the stages and layers shown in Figure \ref{fig:query_execution_yannakakis_example}. Note that this specific join tree was chosen for its simplicity and its execution cannot be parallelised. 

\begin{figure}
    \centering
    \begin{tabular}{c|c|c}
            \textbf{stage} & \textbf{layer} & \textbf{joins} \\
            \hline
            \multirow{ 2}{*}{1} & 1 & $\{ss_1 := (\text{store\_sales} \ltimes \text{store}) \ltimes \text{d1}\}$ \\
              & 2 & $\{sr_1 :=  (\text{store\_returns} \ltimes ss_1) \ltimes \text{d2}\}$ \\
            \hline
            \multirow{ 2}{*}{2} & 1 & $\{ss_2 := ss_1 \ltimes sr_2\}$ \\
              & 2 & $\{d2_2 := \text{d2} \ltimes sr_1\}$ \\
              & 3 & $\{\text{store}_2 := \text{store} \ltimes ss_2\}$ \\
              & 4 & $\{d1_2 := \text{d1} \ltimes ss_2\}$ \\
            \hline
            \multirow{ 2}{*}{3} & 1 & $\{sr_1 \bowtie ss_2 \bowtie d2_2 \bowtie store_2 \bowtie d1_2\}$
    \end{tabular}
    \caption{}
    \label{fig:query_execution_yannakakis_example}
\end{figure}
*************} %

\paragraph{The Join Stage}
Finally, the temporary tables representing the relations after the semi-join stages are combined by natural joins. The straightforward way to do this is either via step-wise joins along the join tree in a bottom-up manner or, alternatively, all relations can be joined in one large statement. The latter option seems to introduce less overhead, but for large original queries, it reintroduces the problem of planning queries with many joins. 
\nop{**************************
Recall that, at this point, there are no more dangling tuples; so join order is less critical. But for a very large number of joins,
systems may still plan them badly. 
**************************}
We therefore take a middle ground and group (via a straightforward greedy procedure) the join tree into subtrees of at most 12 nodes each and materialise the final joins with one join query per subtree, plus a final query joining the subtrees.
Of course, for 0MA queries, no computation of joins is necessary. In this case, the join phase simply refers to the final aggregation over the root node.

\smallskip
Finally, note that these stages are also amenable to parallelisation: as we follow a tree structure, 
we know that the semi-joins and joins for nodes in different subtrees can be computed independently of each other. 
	This thread is not further followed in this paper as the host systems considered here already parallelise query execution to an extent where further parallelisation ``from the outside'' does not seem particularly helpful. However, the additional potential of parallelisation  clearly deserves 
further study in case of full integration of 
Yannakakis-style query execution into these DBMSs.

\section{Further Details on the Experimental Evaluation}
\label{app:AdditionalExp}

As mentioned in Section~\ref{sec:expeval}, 
all data produced by our experiments 
\ifArxivVersion
as well as all information needed for 
reproducing the experiments are available on Figshare:~\figARXIV. 
\else
as well as all information needed for 
reproducing the experiments are available on Figshare:~\figEDBT.
\fi
The most important insights gained with these experiments were summarised 
in Section~\ref{sec:expeval}. 
In this section, we provide some additional details. More specifically, we have a deeper look into 
two important effects of structure-guided query evaluation compared with the traditional approach, namely avoiding the blow-up of intermediate results and reducing the memory and communication cost. For the former, in Section \ref{sect:blowup}, we inspect further 
query plans and the costs of operations inside these plans for DuckDB. For the latter,
in Section~\ref{sect:spark:appendix},
we provide a detailed analysis of these cost components of Spark SQL. 
In both cases, we compare the behaviour of the plain system on the one hand, with \ourSystem on top of the corresponding system on the other hand.

\subsection{Blow-Up of Intermediate Results}
\label{sect:blowup}

\begin{figure*}[h]
  \centering

\begin{minipage}{0.33\textwidth}
        \vspace{0pt}

        \begin{forest}
          my tree style
          [\QueryPlanNode{min}{2.7}{1}, label=left:{\textbf{DuckDB}}
            [ \QueryPlanNode{$\bowtie$}{9}{1\,619\,562\,945}
                [\QueryPlanNode{$\bowtie$}{\underline{144.5}}{1\,699\,148\,884}
                    [\QueryPlanNode{$\bowtie$}{0.9}{27\,213\,862}
                        [\QueryPlanNode{ recording}{0.2}{27\,213\,862}]
                        [\QueryPlanNode{ artist\_credit}{0}{2\,328\,629}]
                    ]
                    [\QueryPlanNode{ release\_group}{0}{2\,375\,238}]
                ]
                [\QueryPlanNode{ release\_group\_prim.}{0}{5}]
            ]
          ]
        \end{forest}

    \end{minipage}
    \begin{minipage}{0.25\textwidth}
        \vspace{0pt}

        \begin{forest}
          my tree style
          [\QueryPlanNode{ \color{blue} recording'}{1.71}{27\,213\,862}, label=left:{\textbf{Yannakakis A}}
                [\QueryPlanNode{$\ltimes$}{0.6}{27\,213\,862}
                    [\QueryPlanNode{recording}{0.2}{27\,213\,862}]
                    [\QueryPlanNode{artist\_credit}{0}{2\,328\,626}]
                ]
          ]
        \end{forest}

\end{minipage}
\begin{minipage}{0.3\textwidth}
    \vspace{0pt}

        \begin{forest}
          my tree style
          [\QueryPlanNode{min}{0}{1}, label=left:{\textbf{Yannakakis B}}
            [\QueryPlanNode{$\ltimes$}{0}{2382938}
                [\QueryPlanNode{$\ltimes$}{1.2}{2450362}
                    [\QueryPlanNode{ release\_group}{0}{2575238}]
                    [\QueryPlanNode{ \color{blue} recording'}{0}{27213862}]
                ]
                [\QueryPlanNode{release\_group\_prim.}{0}{5}]
            ]
          ]
        \end{forest}

\end{minipage}

\begin{minipage}{0.3\textwidth}
\begin{tabular}{lr}
\toprule
System / Phase & Execution Time \\
\midrule 
DuckDB & 40.4 s \\
DuckDB + \ourSystem{} / Total & 1.0 s \\ 
DuckDB + \ourSystem{} / Setup & 0  s \\ 
DuckDB + \ourSystem{} / $\ltimes$-up &  1.0 s \\ 
DuckDB + \ourSystem{} / $\ltimes$-down & -- \\
DuckDB + \ourSystem{} / Join & 0 s \\ 
\bottomrule
\end{tabular}
\end{minipage}

\begin{minipage}{0.3\textwidth}
\centering
\begin{tikzpicture}
	\node	(a)	at	(0,1)	{a};
	\node	(b)	at	(1,1)	{b};
	\node	(c)	at	(2,1)	{c};
	\node	(d)	at	(0,0)	{d};
	\draw (b) ellipse [x radius=1.5, y radius=.3]; %
	\node[scale=.9]   (rg)   at  (3.6,1.3)  {release\_group(a,b,c)};
	\draw[rotate around={90:(0,.5)}] (0,.5) ellipse [x radius=0.85, y radius=.3]; %
	\node[scale=.9]   (rec)   at  (1,-.3)  {recording(a,d)};
	\node[scale=.8,align=left] (txt) at (4,0) {
	    Additional unary edges:\\
	    release\_group\_primary\_type(b)\\
	    artist\_credit(a)
	};
\end{tikzpicture}
\end{minipage}

  \caption{Details of performance difference in query plans of query 04aa. Execution times of operations are in seconds, rounded to one decimal point.}
  \label{fig:q04aa}
\end{figure*}
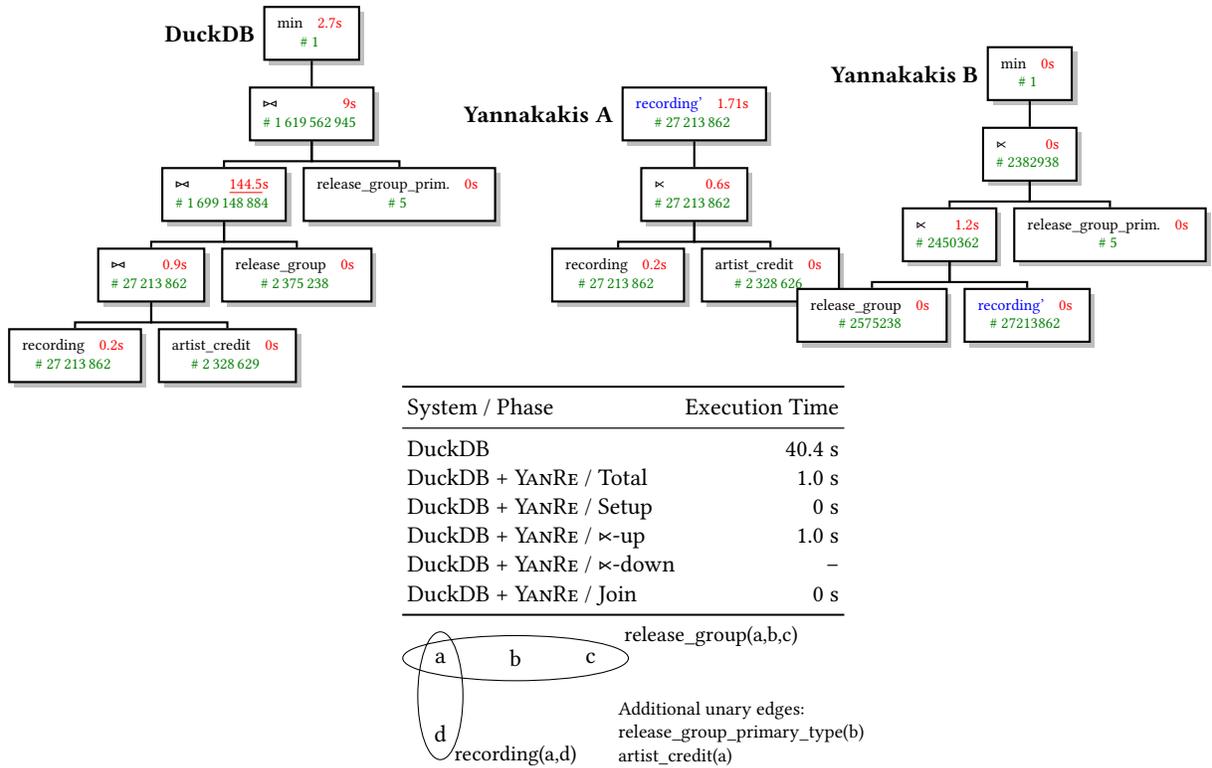

In  Section \ref{sect:deeperInsights},
we have already seen that Yannakakis-style query evaluation 
can be very effective in avoiding the explosion of intermediate results. 
More concretely, we compared in Figure~\ref{fig:q08ad} the query plan of 
plain DuckDB with the query plan of the join phase of \ourSystem{}
on top of DuckDB for the query \texttt{08ad} from the benchmark 
of~\cite{DBLP:conf/sigmod/ManciniKCMA22}. 
We thus 
inspected the full enumeration version of that query. 
We now also compare the query plan for the aggregation version 
of a query from~\cite{DBLP:conf/sigmod/ManciniKCMA22}. 
In Figure~\ref{fig:q04aa}, 
we thus look at query plans 
 together with execution times of each operation
for query \texttt{04aa}. 
  Note that the execution times inside the nodes represent total CPU time. Hence, for operations that run on multiple CPUs,
  they can be higher than the total wall clock time. 
To make the actual wall clock time clear in each case, we also provide a table beneath 
the query plan and join tree. This table shows the execution times for the baseline case as well as for the case when using \ourSystem{}. 
And we also provide a breakdown of the execution time for each stage of \ourSystem. 
Recall that for a 0MA query, the top-down traversal of \YA may be omitted. 
In principle, also the second bottom-up traversal with the joins may be omitted in this 
case. We nevertheless record the time for the ``Join'' phase in the table of Figure~\ref{fig:q04aa} (with 0 seconds, though) since, in \ourSystem,  
this phase takes care of the 
actual aggregate computation.

We can see the query plans for the query \texttt{04aa} in Figure~\ref{fig:q04aa}. On the left, we show the entire query plan for the original 0MA aggregation variant of \texttt{04aa} as produced by DuckDB. On the right, we show the query plans for the two  statements from the \ourSystem{} rewriting, split into two parts for better readability. The output of the query plan ``Yannakakis A'' is used in the query plan ``Yannakakis B'', as visually marked by a blue node named \texttt{recording'} (i.e., the right-most leaf node of the tree ``Yannakakis B''). 
We can see in the original query plan how the entire plan is dominated by a large intermediate result, requiring 144.5 seconds of CPU time for computing the respective join. On the other hand, the query plan for  \ourSystem{} naturally avoids this type of blow-up as it consists only of semi-joins. We want to emphasise here, that this increase in performance is not due to smarter heuristics or join orderings and cannot be mitigated by indexing or similar techniques. This type of blow-up is generally unavoidable if one relies on the splitting of a join query into a sequence of two-way joins -- without prior elimination 
of dangling tuples. Furthermore, while we illustrate only the 0MA case 
for query \texttt{04aa} here, we note that 
a similar effect can also be observed for 
the full enumeration queries. %

\tikzset{
  my node style2/.style={
    font=\small,
    top color=white,
    bottom color=white,
    circle,
    rectangle,
    draw=black,
    thick,
    align=center,
  }
}
\forestset{
  my tree style2/.style={
    for tree={
      parent anchor=south,
      child anchor=north,
      l sep=1pt,
      my node style2,
      edge={draw=black, thick},
      edge path={
        \noexpand\path [draw, \forestoption{edge}] (!u.parent anchor) -- +(0,-7.5pt) -| (.child anchor)\forestoption{edge label};
      },
      if n children=3{
        for children={
          if n=2{calign with current}{}
        }
      }{},
      delay={if content={}{shape=coordinate}{}}
    }
  }
}

\subsection{Memory and Communication Cost with Spark SQL}
\label{sect:spark:appendix}

Spark SQL, as it was designed to be primarily a distributed computation system, is fundamentally different from PostgreSQL and DuckDB. We cover some of the details in this section. We thus first explain some basic principles of query evaluation in Spark SQL and then present metrics covering memory consumption and communication costs, which are especially relevant in the distributed setting. 

\paragraph{Basic principles of query evaluation in Spark SQL}
As opposed to traditional database systems, Spark by itself is only considered a distributed computation framework, operating on an abstraction known as \textit{resilient distributed datasets (RDDs)}. Spark SQL extends this by introducing the \textit{DataSet} and \textit{DataFrame} APIs, the \textit{Catalyst} query optimizer, and SQL query execution, among other features. Therefore, Spark SQL is tailored towards in-memory distributed computation of large amounts of data and does not on its own feature a storage layer for long-term data. In our experiments, the data is therefore initially imported from a PostgreSQL database over the JDBC interface.
Due to its in-memory-first architecture (data is only spilled to disk when the memory is exhausted), Spark has no concept of \textit{tables}, only of \textit{temporary views}. Furthermore, due to Spark's role as (primarily) a batch processing system rather than a data management system, there is no native support of indexes, and only a basic query optimizer compared to a traditional DBMS. Spark performs rule-based optimizations, but only a very limited form of cost-based optimization in joins, e.g., by converting sort-merge joins to hash joins or broadcast joins and by coalescing post-shuffle partitions.

In our integration of \ourSystem{} into PostgreSQL and DuckDB, the execution was split into smaller steps and intermediate results saved in-between as temporary tables in order to prevent the query optimizer from re-ordering the execution. The performance drawbacks of this were not as significant as these systems run only on a single machine with a limited degree of parallelism. Spark SQL, however, was designed for highly parallelisable workloads and the same approach would have been much less effective. By using the fact that Spark SQL's query optimizer does not re-order joins, it is possible to pass all steps of \YA to Spark as 
a single query plan.

In order to illustrate how \ourSystem{} runs on Spark, we present the execution of a specific aggregation query on Spark SQL. Query \textit{07al} (Figure. \ref{fig:query_07al}) comprises an aggregation over 7 tables and 6 joins and turned out to be challenging for Spark SQL on its own.
\ourSystem{} rewrites the query into several \textit{CREATE TEMP VIEW} statements, implementing the bottom-up semi-joins of \YA stage 1.
When executing a query, Spark SQL starts by constructing a \textit{parsed logical plan}, which is next transformed into an \textit{analysed logical plan}, where attributes and relations are mapped to known objects. After running through the optimizer, the \textit{optimized logical plan} is translated into a \textit{physical plan}, which describes the lower-level details, such as sort and exchange steps.
In Figure \ref{fig:query07al_execution_baseline}, the optimized logical plan (with some details removed) resulting from query 07al is shown. Spark SQL executes the query as a left-deep binary tree realising a sequence of 
inner joins, leading to suboptimal performance in this case.
\ourSystem{} on top of Spark SQL, on the other hand, 
produces a rewriting which leads to the query plan seen in
Figure~\ref{fig:query07al_execution_yannakakis}. We make two crucial observations here: when fed the rewritten query, Spark SQL explicitly chooses a semi-join operator for the bottom-up traversal and the tree shape of the query
execution plan is not restricted to a left-deep tree.
Note that, on this query, plain Spark SQL times out after 20 minutes, 
while Spark SQL + \ourSystem{} successfully completes after $\approx$ 9.6 seconds.

\begin{figure}[h]
    \centering
    \begin{lstlisting}[language=SQL]
SELECT min(language.id)
FROM language, work_language, release, work, 
     work_alias, release_status, script
WHERE language.id = work_language.language
    AND language.id = release.language
    AND work_language.work = work.id
    AND work.id = work_alias.work
    AND release.status = release_status.id
    AND release.script = script.id;

    \end{lstlisting}
    \caption{Query 07al (aggregation)}
    \label{fig:query_07al}
\end{figure}

\begin{figure*}
    \centering
    \begin{Verbatim}[fontsize=\small]
Aggregate [min(id) AS min(id)]
+- Project [id]
   +- Join Inner, (script = id)
      :- Project [id, script]
      :  +- Join Inner, (status = id)
      :     :- Project [id, status, script]
      :     :  +- Join Inner, (id = work)
      :     :     :- Project [id, status, script, id]
      :     :     :  +- Join Inner, (work = id)
      :     :     :     :- Project [id, work, status, script]
      :     :     :     :  +- Join Inner, (id = language)
      :     :     :     :     :- Project [id, work]
      :     :     :     :     :  +- Join Inner, (id = language)
      :     :     :     :     :     :- Project [id]
      :     :     :     :     :     :  +- Filter isnotnull(id)
      :     :     :     :     :     :     +- Relation [id,iso_code_2t,iso_code_2b,iso_code_1,name,frequency,iso_code_3]
      :     :     :     :     :     +- Project [work, language]
      :     :     :     :     :        +- Filter (isnotnull(language) AND isnotnull(work))
      :     :     :     :     :           +- Relation [work,language,edits_pending,created]
      :     :     :     :     +- Project [status, language, script]
      :     :     :     :        +- Filter (isnotnull(language) AND (isnotnull(status) AND isnotnull(script)))
      :     :     :     :           +- Relation [id,gid,name,artist_credit,release_group,status,packaging,language, ... ]
      :     :     :     +- Project [id]
      :     :     :        +- Filter isnotnull(id)
      :     :     :           +- Relation [id,gid,name,type,comment,edits_pending,last_updated]
      :     :     +- Project [work]
      :     :        +- Filter isnotnull(work)
      :     :           +- Relation [id,work,name,locale,edits_pending,last_updated,type,sort_name,begin_date_year, ... ] 
      :     +- Project [id]
      :        +- Filter isnotnull(id)
      :           +- Relation [id,name,parent,child_order,description,gid]
      +- Project [id]
         +- Filter isnotnull(id)
            +- Relation [id,iso_code,iso_number,name,frequency]
    \end{Verbatim}
    \caption{Plain Spark SQL: optimized logical plan of query 07al (aggregation)}
    \label{fig:query07al_execution_baseline}
\end{figure*}

\begin{figure*}
    \centering
    \begin{Verbatim}[fontsize=\small]
Aggregate [min(v0) AS min(v0)]
+- Project [id AS v0]
   +- Join LeftSemi, (id = v0)
      :- Project [id]
      :  +- Relation [id,iso_code_2t,iso_code_2b,iso_code_1,name,frequency,iso_code_3]
      +- Project [language AS v0]
         +- Join LeftSemi, (language = v0)
            :- Project [language]
            :  +- Join LeftSemi, (script = v8)
            :     :- Project [language, script]
            :     :  +- Join LeftSemi, (status = v6)
            :     :     :- Project [status, language, script]
            :     :     :  +- Relation [id,gid,name,artist_credit,release_group,status,packaging,language,script, ...]
            :     :     +- Project [id AS v6]
            :     :        +- Relation [id,name,parent,child_order,description,gid]
            :     +- Project [id AS v8]
            :        +- Relation [id,iso_code,iso_number,name,frequency]
            +- Project [language AS v0]
               +- Join LeftSemi, (work = v3)
                  :- Project [work, language]
                  :  +- Relation [work,language,edits_pending,created]
                  +- Project [work AS v3]
                     +- Join LeftSemi, (work = v3)
                        :- Project [work]
                        :  +- Relation [id,work,name,locale,edits_pending,last_updated,type,sort_name, ...]
                        +- Project [id AS v3]
                           +- Relation [id,gid,name,type,comment,edits_pending,last_updated]

    \end{Verbatim}
    \caption{\ourSystem{} on Spark SQL: optimized logical plan of query 07al (aggregation)}
    \label{fig:query07al_execution_yannakakis}
\end{figure*}

\paragraph{Detailed runtime metrics of Spark SQL}
We now present the runtime metrics collected by Spark during the execution of the benchmarks: 
shuffle writes (i.e., data exchanged between nodes in the cluster) and memory consumption.
These results show that \ourSystem{} indeed addresses 
some of the underlying causes of long-running queries.
We mention that with all the measurements reported in this section,
we have made no distinction between acyclic and cyclic queries. That is, for the cyclic queries, we computed a generalized hypertree decomposition (GHD), turned it into a join tree and then applied \ourSystem. So, in principle, we report on 
measurements obtained with all queries from the benchmark of 
\cite{DBLP:conf/icde/MageirakosMKCA22}. Moreover, we present the results obtained 
with theses queries both for the {\em full enumeration} variant 
(Figures~\ref{fig:shuffle_bytes_enum} -- \ref{fig:memory_enum_hist})
and for 
the minimum {\em aggregation} variant
(Figures~\ref{fig:shuffle_bytes_agg} --  \ref{fig:memory_hist_agg_to}), 
which ensures the 0MA property.
However, as will be mentioned 
explicitly below, we sometimes only show the results for those queries which terminated both for plain Spark SQL and \ourSystem. Likewise, we will mention
explicitly below if the results 
include also the queries that timed out on one of the two systems.  

Spark is able to run in \textit{local mode} or in \textit{cluster mode}, where applications can be deployed on multiple cluster managers, for example Hadoop YARN, Kubernetes or the Spark \textit{standalone} cluster manager. In our experiments, the Spark applications were submitted on a YARN cluster. In Spark's local mode, which achieves parallelism only via multi-threading, \ourSystem{} performed similarly. However, we focus here on the cluster environment, 
as it tends to be the main target for real-world applications. Spark \textit{applications} set up multiple \textit{executors}, which persist as long as the applications do. \textit{Jobs} are (parallel) computations, which consist of multiple \textit{stages}, which are again collections of \textit{tasks} that depend on each other, resulting in a \textit{shuffle} operation, which re-distributes data among the partitions. Individual tasks are executed by the executors.

As a measure of communication cost in the Spark application, the total number of shuffle write records (i.e., rows) and shuffle write bytes were collected from the monitoring REST API\footnote{\url{https://spark.apache.org/docs/latest/monitoring.html}}.
Due to the shuffle boundaries at each of the stages, Spark records the shuffle bytes / records at each stage. We thus compute, as a measure of total communication cost, the sum of shuffle bytes / records. 
In Figure~\ref{fig:shuffle_bytes_enum}, the sum of shuffle bytes, for the baseline 
(plain Spark SQL) and for the \ourSystem rewriting is given, where each point represents these two measurements for a single query. Figure \ref{fig:shuffle_records_enum} presents a slightly different perspective, where only the count of records is considered, not their size. However,  the absolute numbers are very similar due to the generally low record size.
The orange points labelled as top 10\% represent the cases where the baseline measurement and/or the \ourSystem measurement belongs to the 10\% of the largest values, i.e., 
those queries which 
constitute the most challenging decile when considering both approaches.
The $45^\circ$ line indicates those points where the baseline and \ourSystem yield the
same result. Points below this line represent queries where \ourSystem caused higher 
communication cost and points above this line represent queries where  
plain Spark SQL performed worse. 
From this data, we can see that, although the overhead of \ourSystem{} is visible on the simpler instances, \ourSystem{} is increasingly competitive on the more challenging instances. 
It is to be noted that we have only considered queries here where both the baseline and \ourSystem{} terminate. We will later see that the results are even more in favour of 
\ourSystem when we also include the queries that timed out either for plain Spark SQL or for \ourSystem on top of Spark SQL. Both Figures~\ref{fig:shuffle_bytes_enum}
and~\ref{fig:shuffle_records_enum}  refer to the {\em full enumeration} 
variants of the 
benchmark queries.

For each stage, Spark also records the peak execution memory, i.e., the peak memory consumption of one of the executors involved. 
Figure~\ref{fig:memory_enum} shows the maximum over all peak execution memory values, 
i.e., the global maximum over all executors over all stages, of the baseline execution, in relation to the \ourSystem{} execution. We can observe here that the memory consumption is highly skewed towards very low numbers in the cases of the easy-to-solve queries, but becomes very large (when seen in relation to the 256 GB available on each node) in the hard cases. 
Figure~\ref{fig:memory_enum_hist} presents the skewed distribution even clearer.
Moreover, there is a significant difference between the measurements for plain Spark SQL and 
\ourSystem: while \ourSystem reaches a peak memory consumption of 50GB only in very rare cases, 
the memory consumption of plain Spark SQL surpasses this value in many cases and may even go as 
high as the total 256 GB memory available on each node.
Both Figures~~\ref{fig:memory_enum}
and~\ref{fig:memory_enum_hist}  
refer to the {\em full enumeration} 
variants of the 
benchmark queries.

We now have a closer look at the {\em aggregation} variants of the 
benchmark queries. 
Figures \ref{fig:shuffle_bytes_agg} and 
\ref{fig:shuffle_records_agg} 
show the sum of shuffle bytes and records, respectively, 
in the same way as previously shown in the full enumeration case. 
Additionally, Figures \ref{fig:shuffle_bytes_agg_to} and  \ref{fig:shuffle_records_agg_to}
also include the instances where timeouts occurred. 
Not surprisingly, it can be seen that the queries with timeouts tend to have a significantly increased communication cost. The effect of \ourSystem is even more pronounced in this case, 
especially for the shuffle write records shown in Figure~\ref{fig:shuffle_records_agg_to}: 
here the majority of the very large intermediate results, and even all of the top 10\%, have a higher value for the baseline execution than for \ourSystem.
Another interesting observation which we can make from these results, through the low deviation from the $45^\circ$ line, is that the overhead of \ourSystem is 
lower on the aggregation (i.e., 0MA) queries than on the full enumeration queries. 

The results of the peak memory consumption measurements 
presented in Figures~\ref{fig:memory_agg}  and \ref{fig:memory_hist_agg} 
(for the queries that terminate both with plain Spark SQL and \ourSystem)
are similar to the case of full enumeration. Actually, here \ourSystem{} performs even better, 
in that the memory consumption is now similar to the baseline on easier instances.
The measurements shown in Figures~\ref{fig:memory_agg_to} 
and~\ref{fig:memory_hist_agg_to} -- now including also those queries which 
caused a timeout with plain Spark SQL or \ourSystem -- draw a similar picture.

\clearpage

\begin{figure*}
\centering
\begin{minipage}{.5\textwidth}
  \centering
  \includegraphics[width=.95\linewidth]{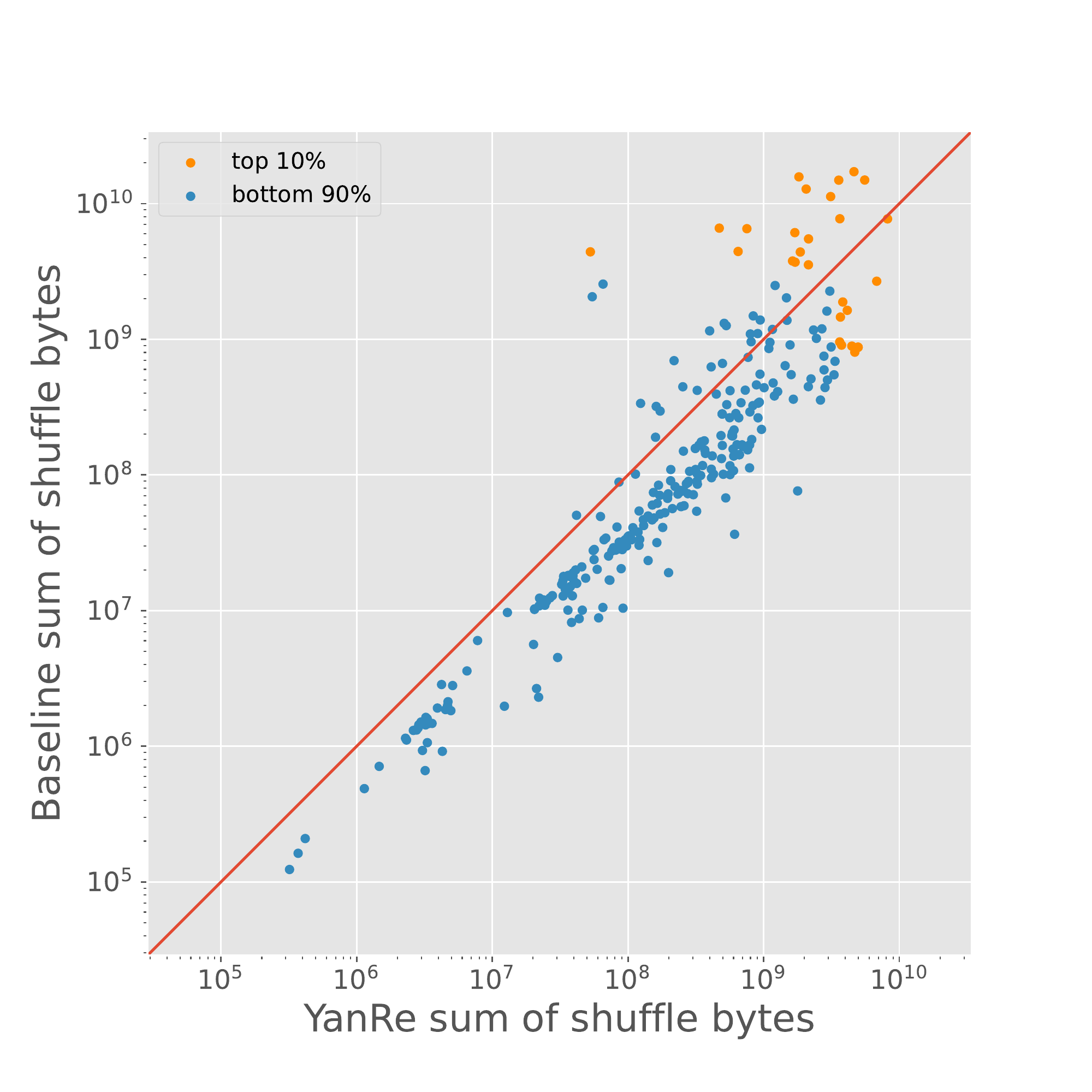}
  \caption{Sum of shuffle write bytes \newline (full enumeration, \ourSystem{} and plain Spark SQL; \newline
  only queries that terminated on both systems)}
  \label{fig:shuffle_bytes_enum}
\end{minipage}%
\begin{minipage}{.5\textwidth}
  \centering
  \includegraphics[width=.95\linewidth]{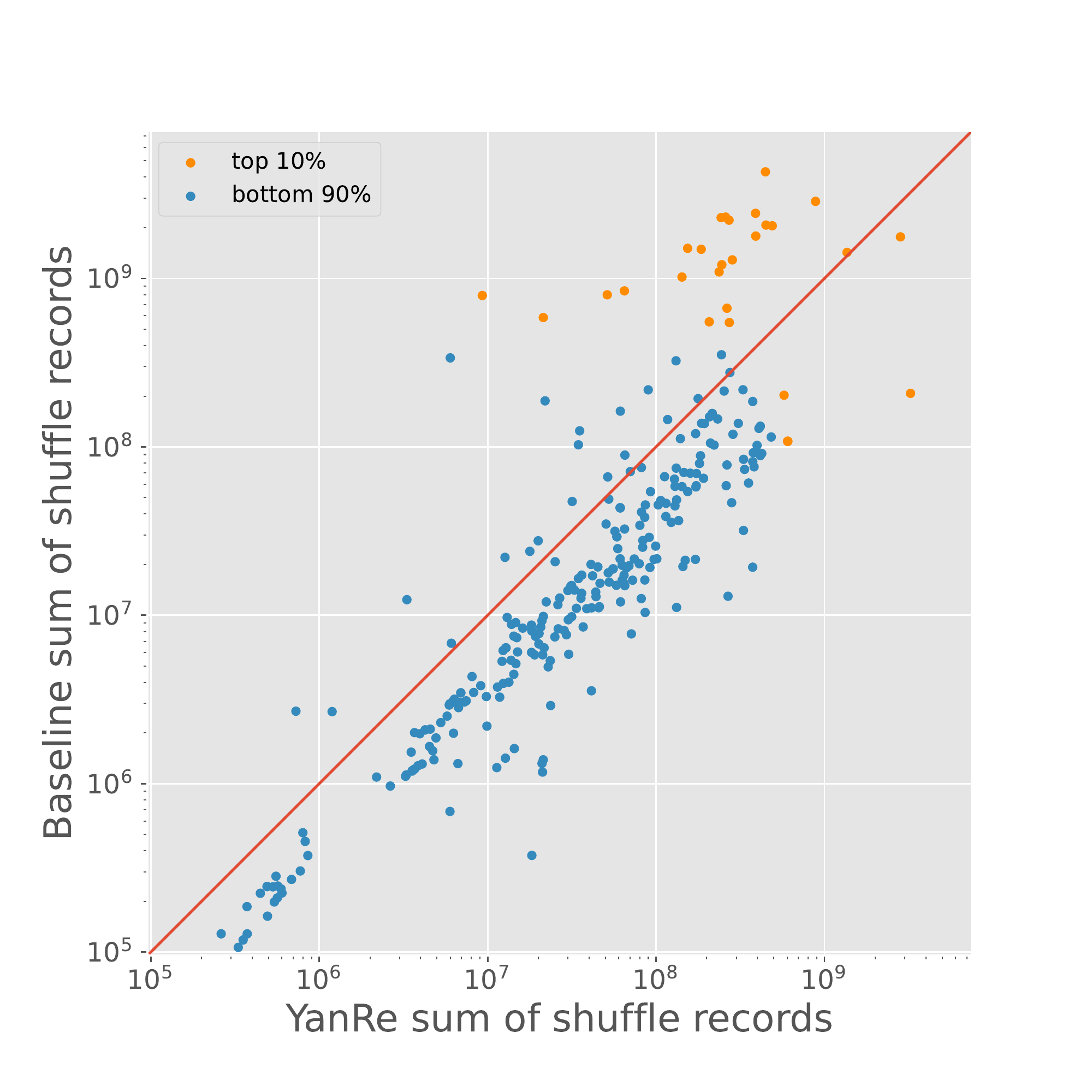}
  \caption{Sum of shuffle write records \newline 
  (full enumeration, \ourSystem{} and plain Spark SQL \newline
  only queries that terminated on both systems))}
  \label{fig:shuffle_records_enum}
\end{minipage}
\end{figure*}

\begin{figure*}
\centering
\begin{minipage}{.5\textwidth}
  \centering
  \includegraphics[width=.95\linewidth]{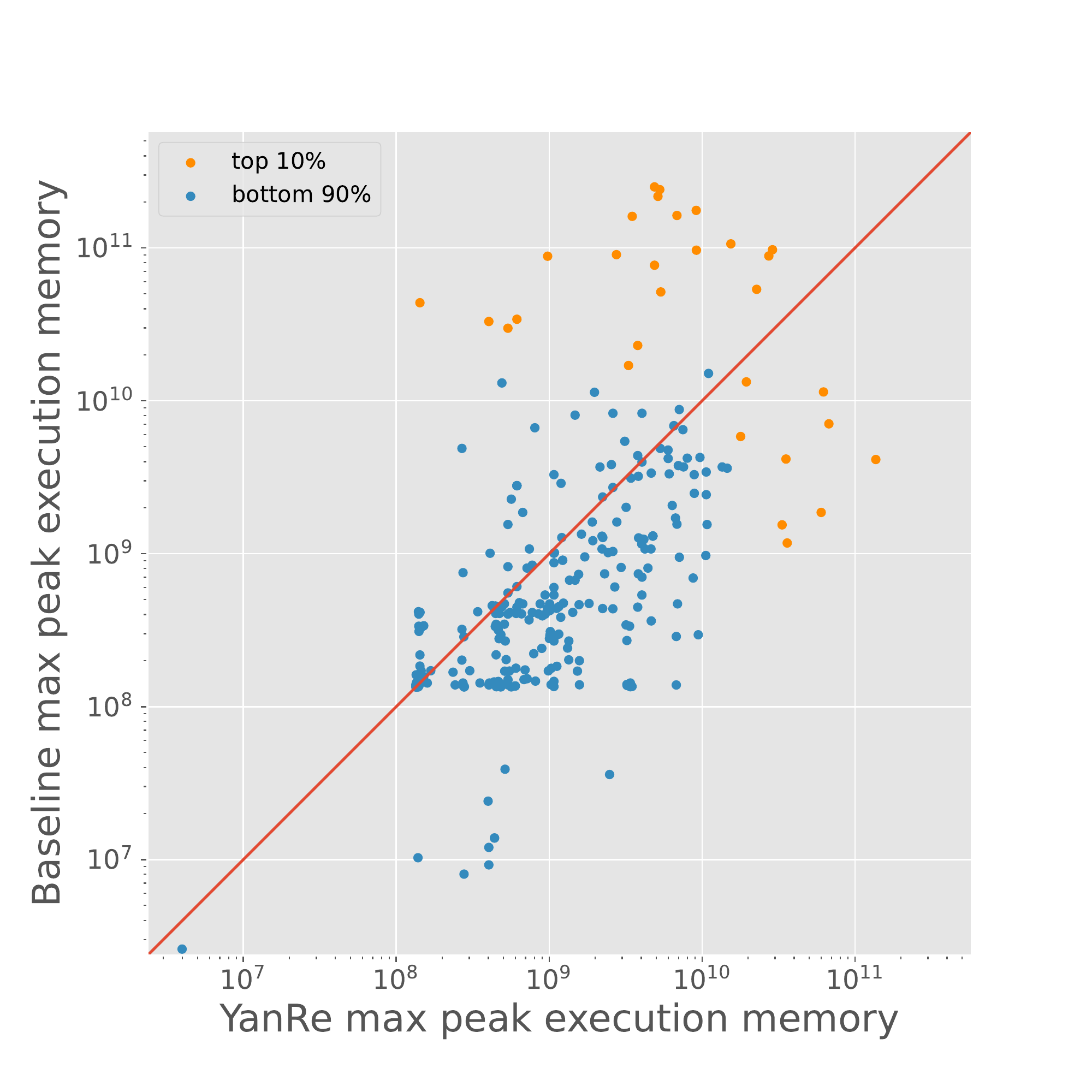}
  \caption{Peak memory consumption \newline 
  (full enumeration, \ourSystem{} and plain Spark SQL; \newline
  only queries that terminated on both systems)}
  \label{fig:memory_enum}
\end{minipage}%
\begin{minipage}{.5\textwidth}
  \centering
  \includegraphics[width=.95\linewidth]{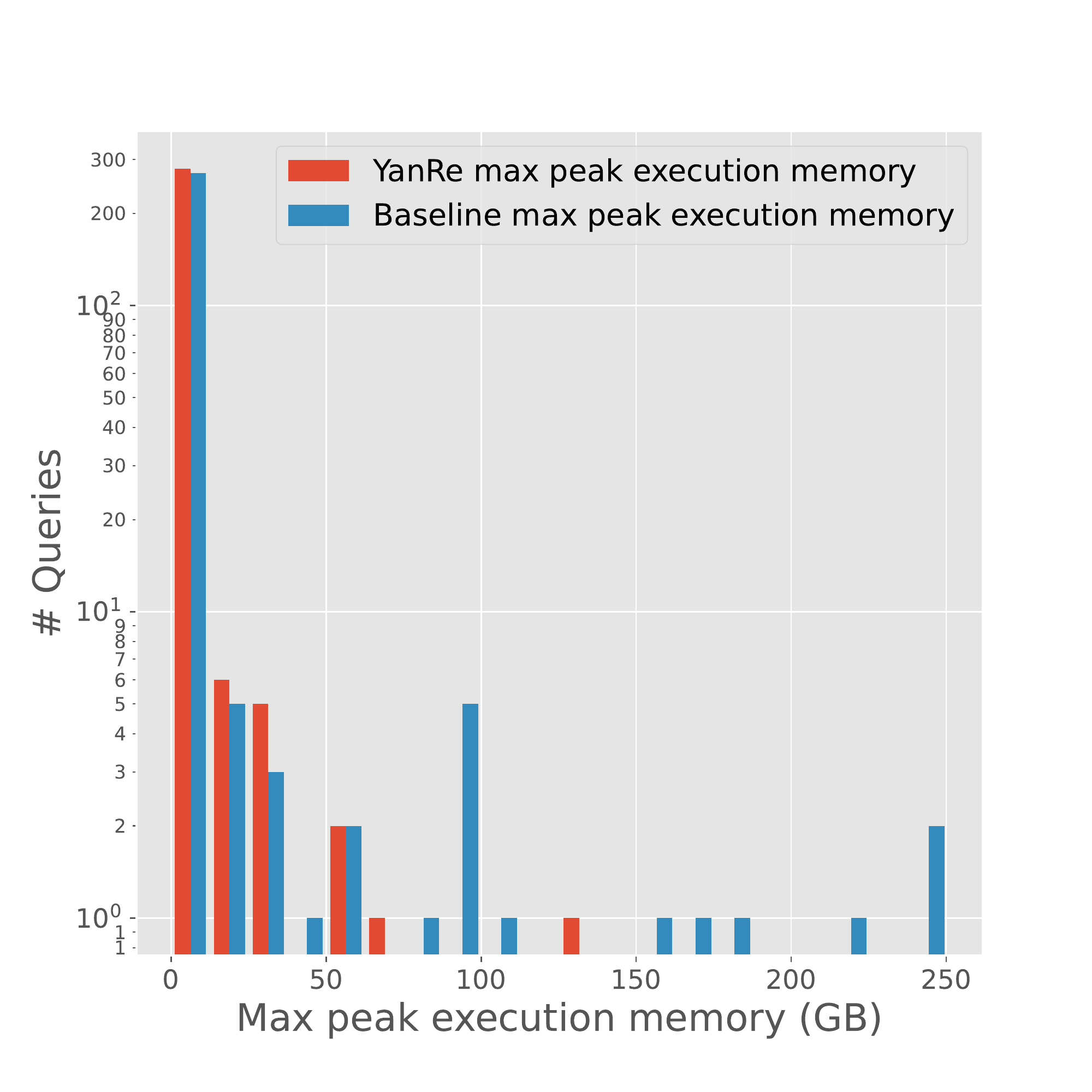}
  \caption{Peak memory consumption \newline
  (full enumeration, \ourSystem{} and plain Spark SQL,  \newline
  only queries that terminated on both systems)}
  \label{fig:memory_enum_hist}
\end{minipage}
\end{figure*}

\begin{figure*}
\centering
\begin{minipage}{.48\textwidth}
  \centering
  \includegraphics[width=.95\linewidth]{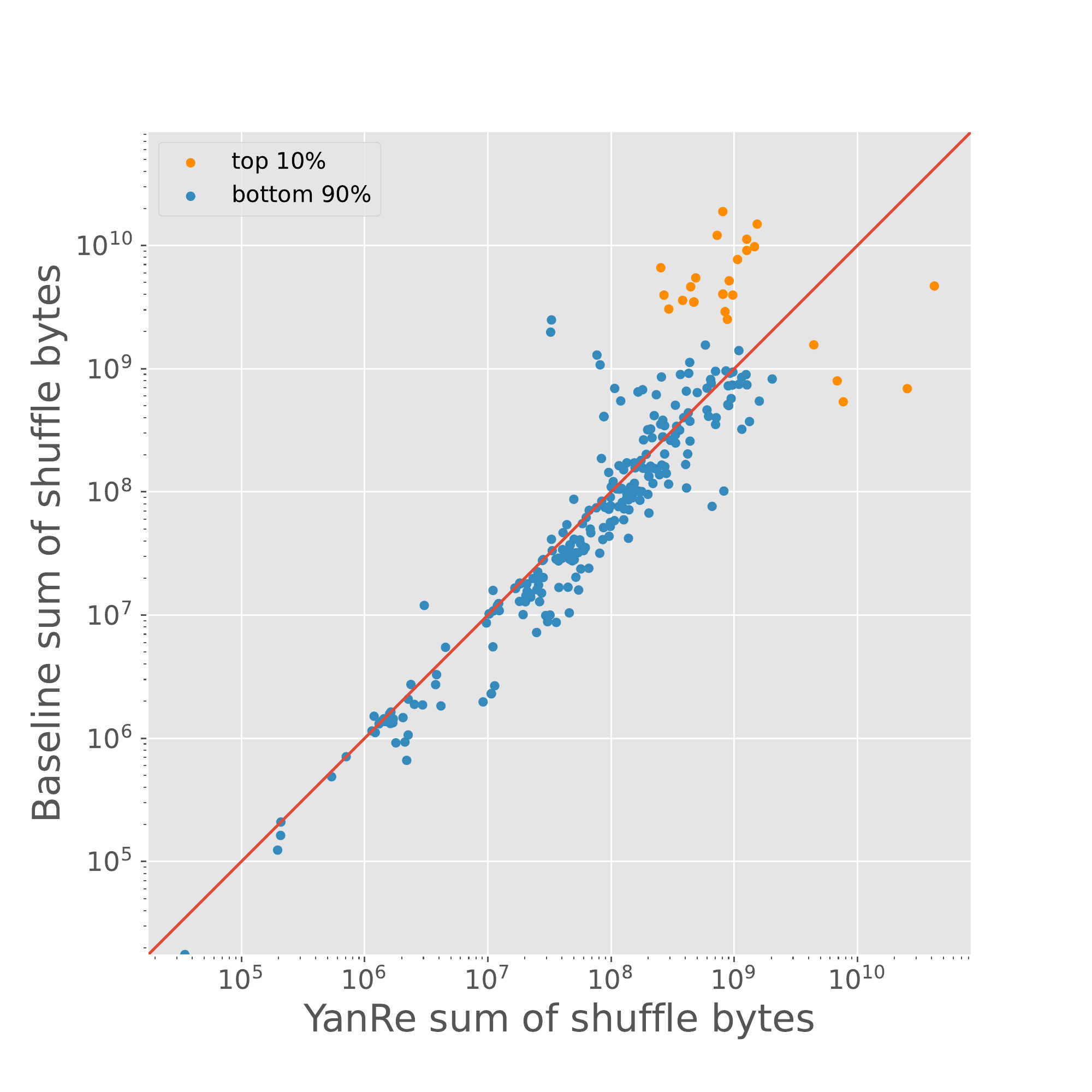}
  \caption{Sum of shuffle write bytes \newline 
  (aggregation, \ourSystem{} and plain Spark SQL; \newline
  only queries that terminated on both systems)}
  \label{fig:shuffle_bytes_agg}
\end{minipage}
\begin{minipage}{.5\textwidth}
  \centering
  \includegraphics[width=.95\linewidth]{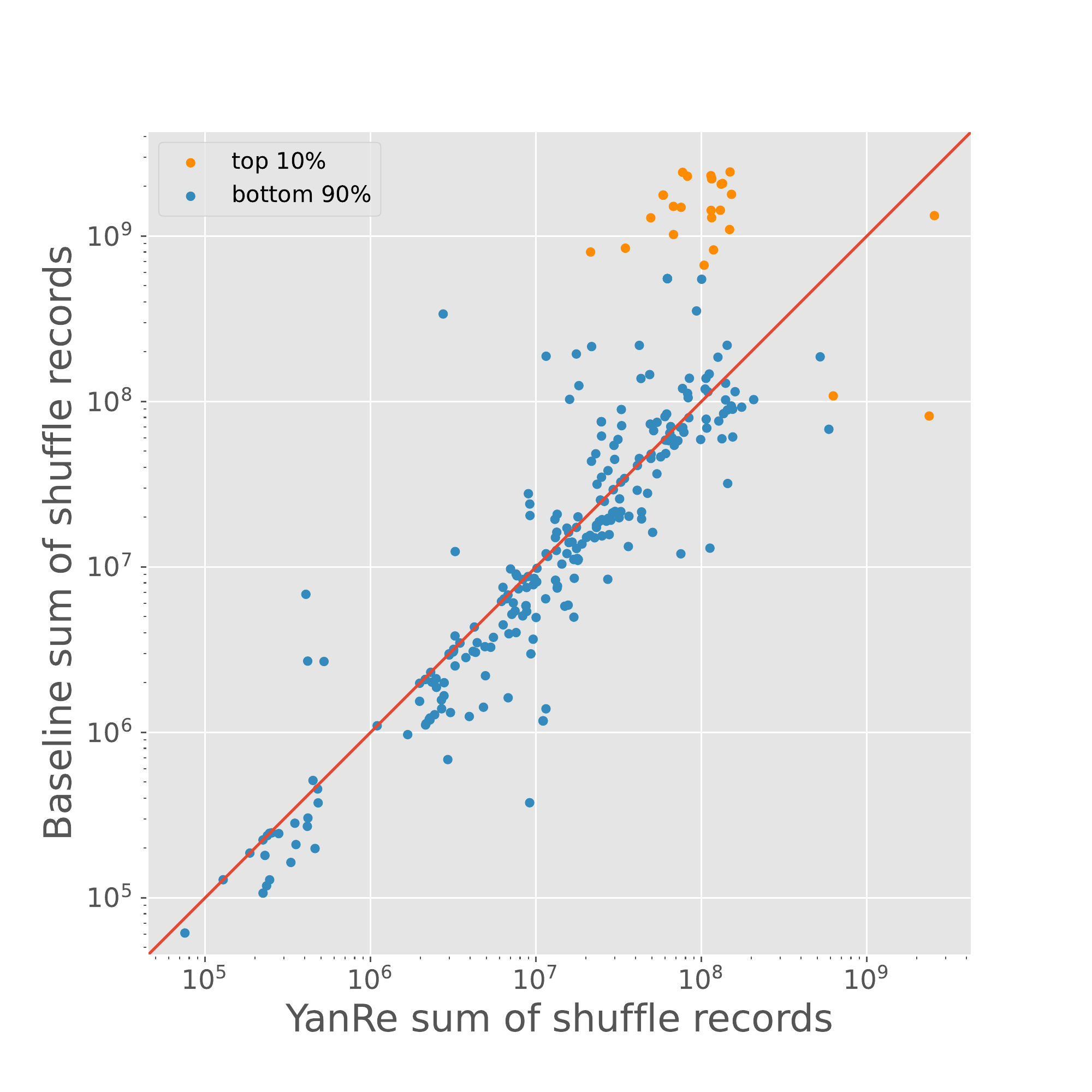}
  \caption{Sum of shuffle write records \newline 
  (aggregation, \ourSystem{} and plain Spark SQL; \newline
  only queries that terminated on both systems)}
  \label{fig:shuffle_records_agg}
\end{minipage}%
\end{figure*}

\begin{figure*}
\centering
\begin{minipage}{.48\textwidth}
  \centering
  \includegraphics[width=.95\linewidth]{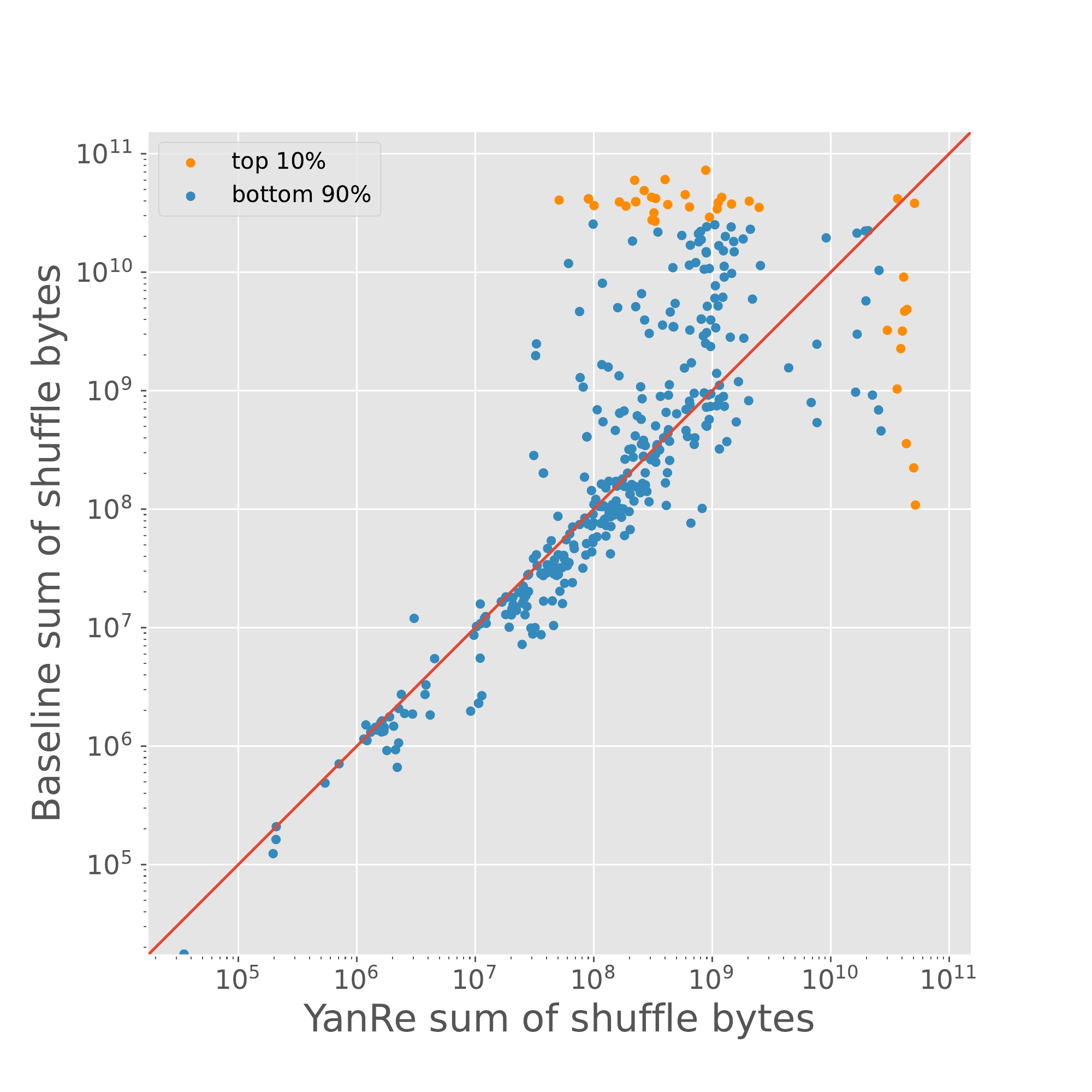}
  \caption{Sum of shuffle write bytes \newline 
  (aggregation, \ourSystem{} and Spark SQL; \newline
  including queries with timeout)}
  \label{fig:shuffle_bytes_agg_to}
\end{minipage}
\begin{minipage}{.5\textwidth}
  \centering
  \includegraphics[width=.95\linewidth]{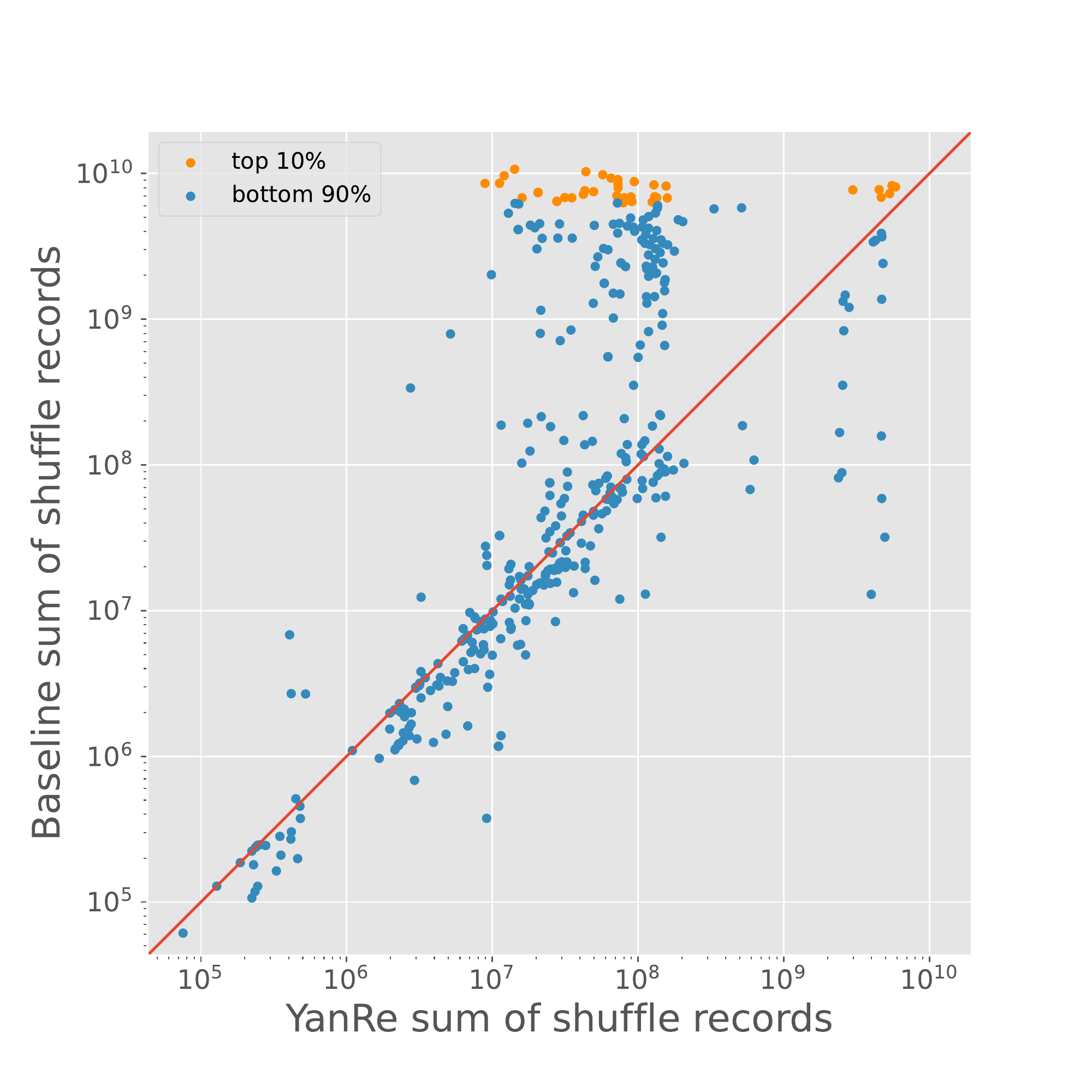}
  \caption{Sum of shuffle write records \newline 
  (aggregation, \ourSystem{} and Spark SQL; \newline
  including queries with timeout)}
  \label{fig:shuffle_records_agg_to}
\end{minipage}
\end{figure*}

\begin{figure*}
\centering
\begin{minipage}{.5\textwidth}
  \centering
  \includegraphics[width=.95\linewidth]{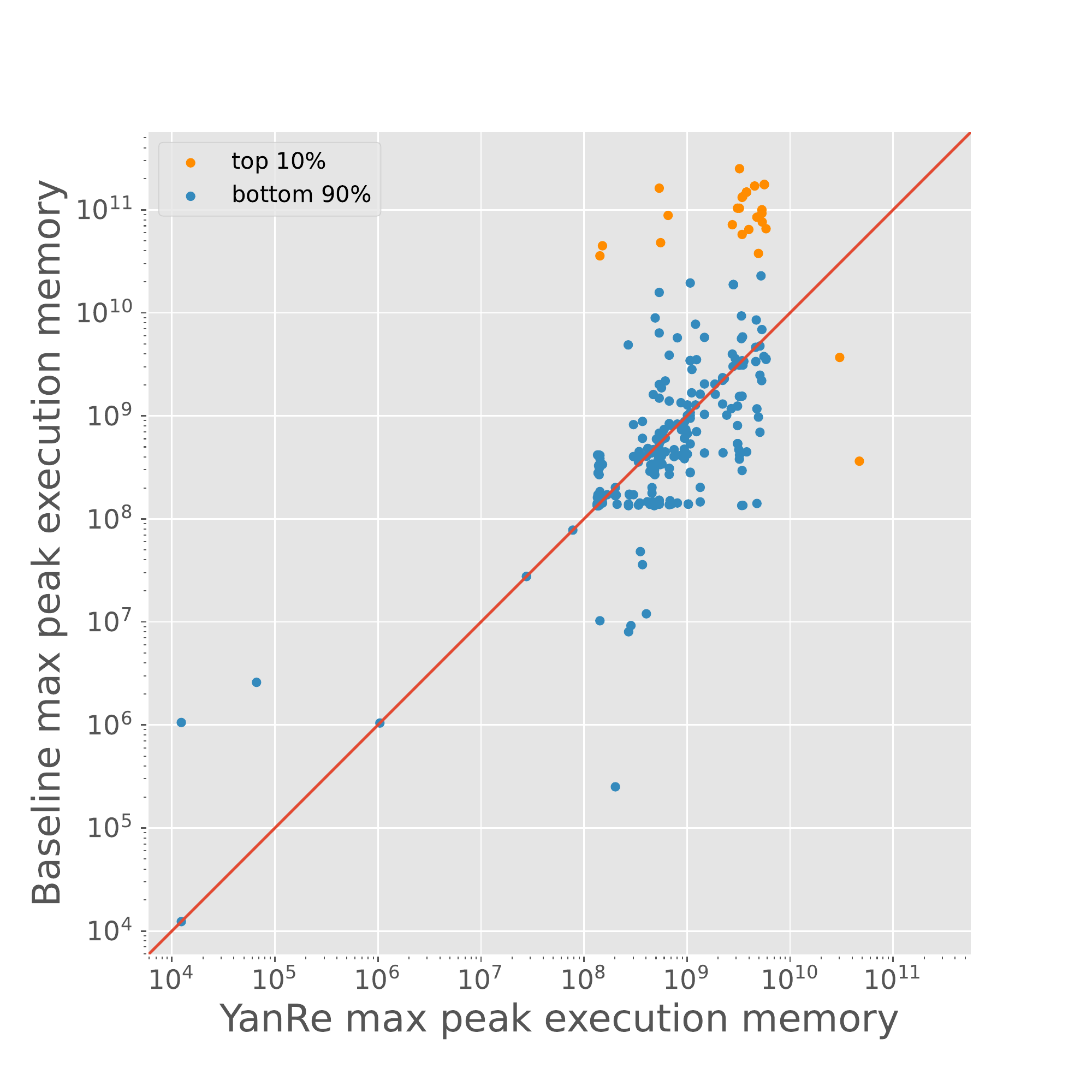}
  \caption{Peak memory consumption \newline 
  (aggregation, \ourSystem{} and Spark SQL; \newline
  only queries that terminated on both systems)}
  \label{fig:memory_agg}
\end{minipage}%
\begin{minipage}{.5\textwidth}
  \centering
  \includegraphics[width=.95\linewidth]{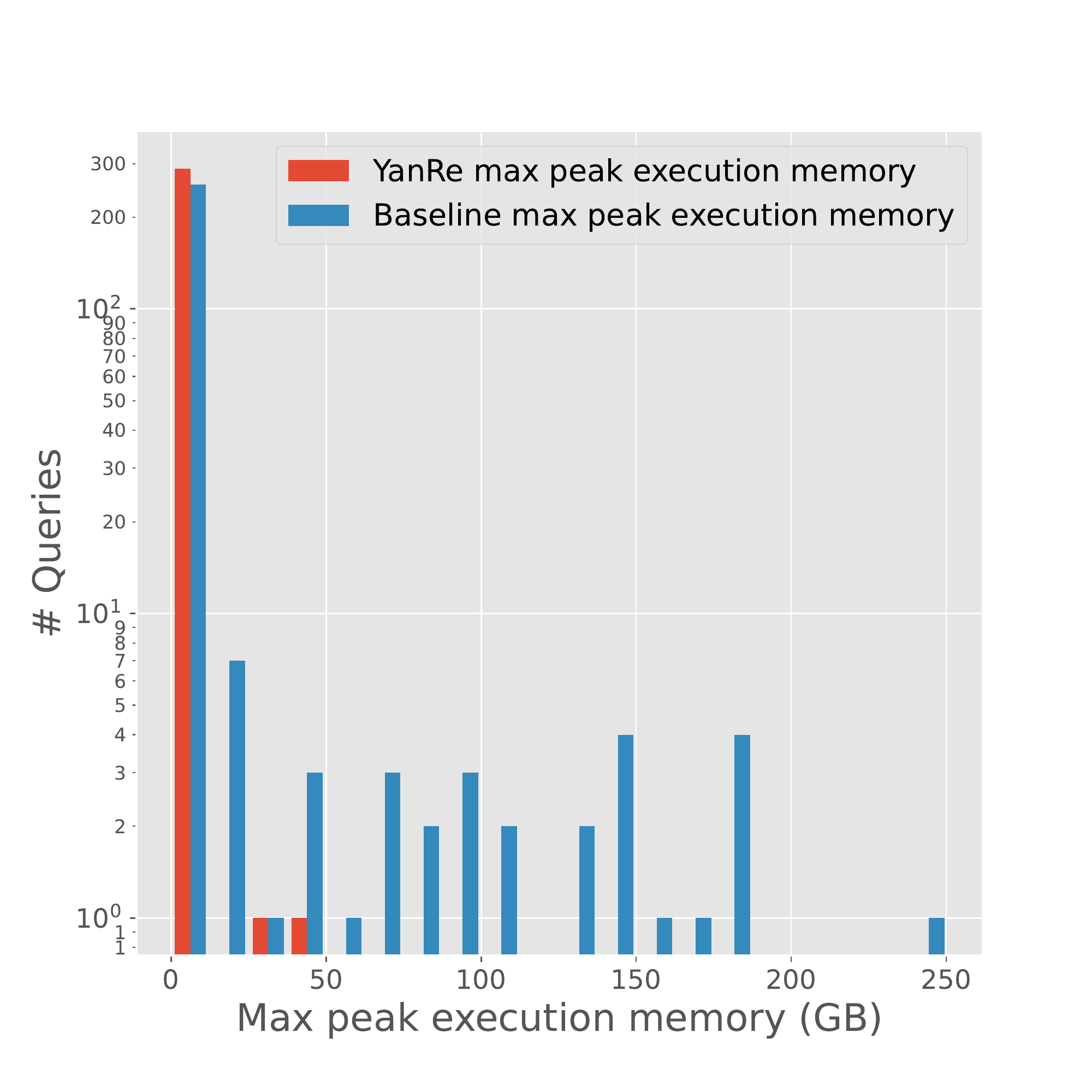}
  \caption{Peak memory consumption \newline
  (aggregation, \ourSystem{} and Spark SQL; \newline
  only queries that terminated on both systems)}
  \label{fig:memory_hist_agg}
\end{minipage}
\end{figure*}

\begin{figure*}
\centering
\begin{minipage}{.5\textwidth}
  \centering
  \includegraphics[width=.95\linewidth]{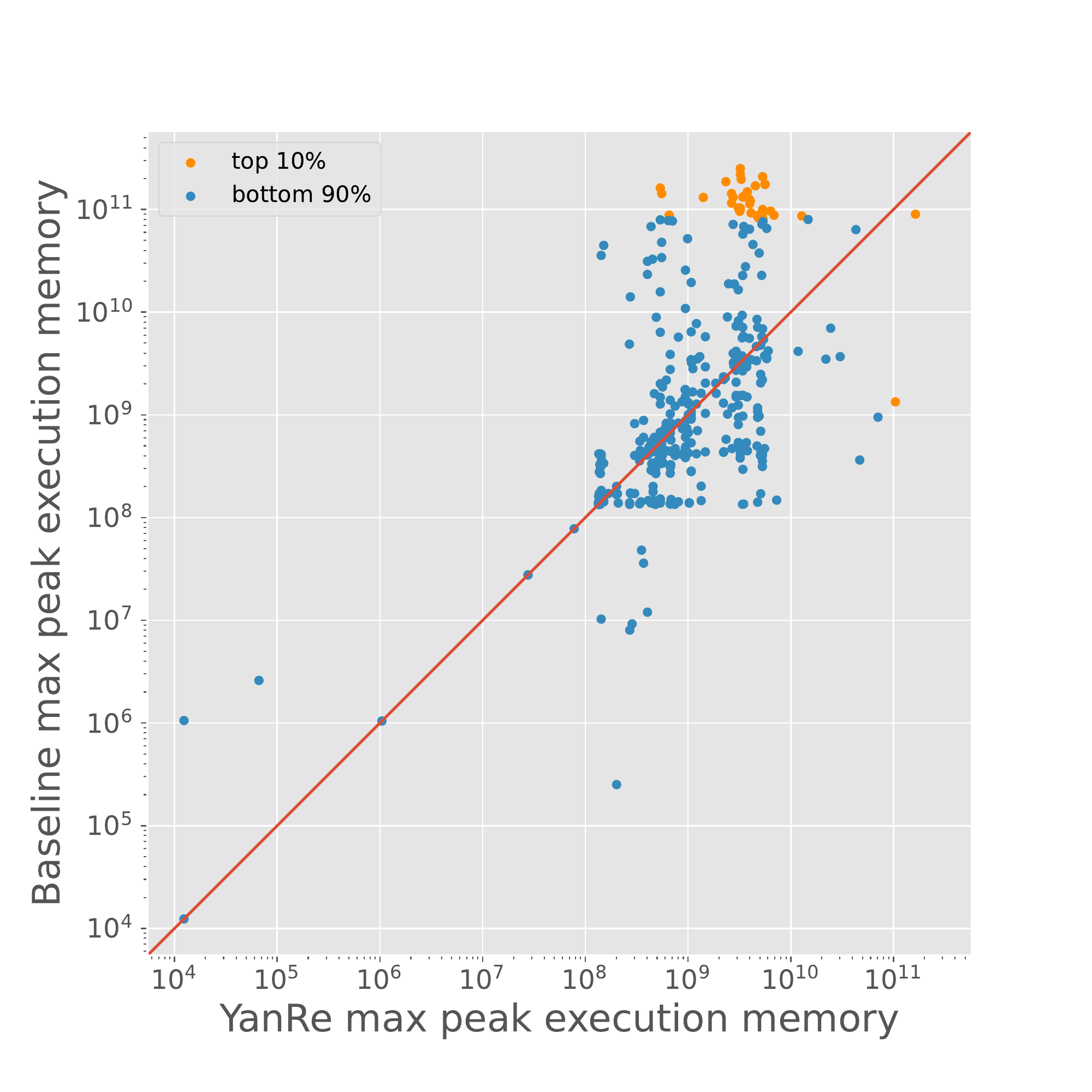}
  \caption{Peak memory consumption \newline 
  (aggregation, \ourSystem{} and Spark SQL; \newline
  including queries with timeout)}
  \label{fig:memory_agg_to}
\end{minipage}%
\begin{minipage}{.5\textwidth}
  \centering
  \includegraphics[width=.95\linewidth]{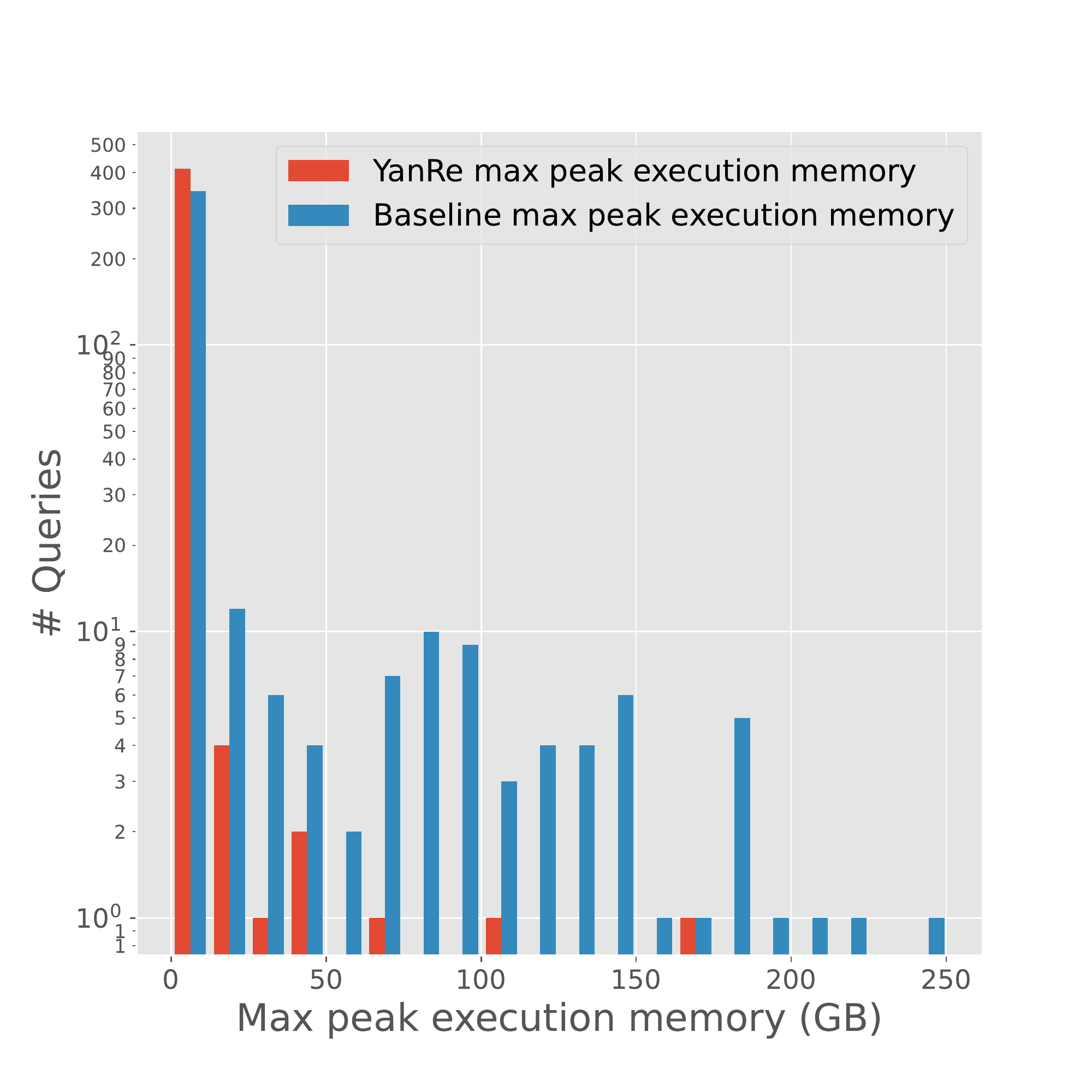}
  \caption{Peak memory consumption \newline
  (aggregation, \ourSystem{} and Spark SQL; \newline
  including queries with timeout)}
  \label{fig:memory_hist_agg_to}
\end{minipage}
\end{figure*}

\clearpage

\section{Cyclic Queries}
\label{sec:cycles}

In this section, we briefly discuss the additional challenges that need to be tackled to obtain similar improvements as reported in the acyclic case also for cyclic queries. We have already reported in 
Section~\ref{sec:expeval}
on some very preliminary
results with cyclic 
queries from~\cite{DBLP:conf/sigmod/ManciniKCMA22}.
The results shown in Table~\ref{tab:ghdqueries} were obtained by first 
computing different generalized hypertree decompositions (GHDs) for three
cyclic queries (09ac, 11ag, and 11al)
from the benchmark. These GHDs were constructed by repeated execution of the decomposition programme BalancedGo~\cite{DBLP:journals/constraints/GottlobOP22} with randomised search order. For each of the distinct GHDs computed in this way, 
we then proceed as in the acyclic case, with the only difference being that the initial relation associated with a tree node $u$ may now either 
be a base relation or a view obtained by joining the relations of the 
edge cover labelling the node of the GHD.

As could be seen in Table~\ref{tab:ghdqueries},
the 
effort of structure-guided query evaluation 
via GHDs can vary heavily, depending on the chosen GHD and, 
in particular, 
on the joins required to turn the GHD into a join tree.
Importantly, even small hypergraphs can have a relatively large number of different GHDs of minimal width. We are therefore confronted with another optimisation problem of finding the GHD with the most efficient reduction to the acyclic case.

We further illustrate this by taking a closer look at one of the cyclic queries thus studied, namely query \textsf{09ac}, which we recall in full in Figure~\ref{fig:queryNineAC}.
On the left-hand side of Figure~\ref{fig:09acex}, we have the hypergraph of this query.
For our purposes, only the structure of the hypergraph is relevant and not the precise names of the attributes. For the sake of better readability, 
we have therefore abbreviated the attribute names to 
a,b,c,d,e,f. Moreover,  attributes irrelevant to the query
have been omitted altogether. 
The correspondence between these abbreviations and 
the true attribute names 
is shown in Table~\ref{tab:abbreviations}. In this table, we have omitted 
the relations which only occur with a single attribute in the query. The correspondence between 
abbreviation and true name is obvious in these cases: 
artist\_credit.id (abbreviated to a), release\_country.release (abbreviated to c), 
release\_group\_secondary\_type\_join.release\_group (abbreviated to b), and 
release\_group\_prior\_type.id (abbreviated to e). 
Note that we have omitted unary edges (which correspond to relations with a 
single (relevant) attribute) from the hypergraph since they have no effect on the acyclicity of a query.
Of course, in the GHDs, the unary relations have to be reintroduced. However, the join with a 
unary relation trivially degenerates to a semi-join. Hence, they can never lead to a blow-up of intermediate results.

On the right-hand side of Figure~\ref{fig:09acex}, 
we have three of the different GHDs generated for this query in our experiments together with the overall execution time of DuckDB + \ourSystem to answer the query. For space reasons, the labels of the nodes contain abbreviations of relation names. 
The correspondence between these abbreviations and 
the true relation names 
are shown in Table~\ref{tab:abbreviationsRelations}. 
We can observe clear structural differences between the GHDs, with decomposition \textbf{Fast} branching only to at most 3 children, while decomposition \textbf{Timeout} is flat and very wide. More importantly, the joins needed to turn the GHDs into join trees are markedly different. Decomposition \textbf{Timeout} induces the costly cross product between  \textsf{medium} and \textsf{release\_group}, while decomposition \textbf{Fast} avoids such views. The third decomposition \textbf{Fast-2} shows a third GHD for which execution is even faster than for \textbf{Fast}. Notably, \textbf{Fast-2} requires only 2 joins to turn
the GHD into a join tree -- in contrast to the 5 joins needed in \textbf{Fast}. 
For reference, ``plain'' DuckDB (i.e., without the rewriting done by \ourSystem) times out on this query and PostgreSQL solves it in 85 seconds.

We conclude our discussion of cyclic CQs with a note on the complexity of 
computing decompositions. 
Until recently, computing GHDs quickly would have presented a
further challenge for cyclic queries. However, with significant recent advancements in
decomposition
algorithms~\cite{DBLP:conf/pods/GottlobLOP22,DBLP:journals/constraints/GottlobOP22,DBLP:journals/jea/FischlGLP21,DBLP:conf/ijcai/SchidlerS21},
the time to compute a width 2 decomposition for the CQs of
the benchmark used here has become trivial. All decompositions used in the
experiments of this section were computed in under 5 milliseconds using off-the-shelf hardware
(in comparison, PostgreSQL takes over 100 milliseconds to create the
query plan for query \textsf{09ac} of the benchmark system).

\begin{figure*}[h]
\small
    \begin{lstlisting}[language=SQL]
SELECT   track.recording, track.medium, medium.release, artist_credit.id, release.release_group, 
         release_group.type
FROM     artist_credit, recording, release_group, release_group_secondary_type_join,
         release_group_primary_type, track, release, medium, release_country
WHERE    artist_credit.id = recording.artist_credit AND release.id = medium.release
         AND artist_credit.id = release_group.artist_credit AND track.medium = medium.id
         AND release_group.id = release_group_secondary_type_join.release_group 
         AND release.id = release_country.release;
         AND release_group.type = release_group_primary_type.id
         AND artist_credit.id = track.artist_credit
         AND recording.id = track.recording
         AND artist_credit.id = release.artist_credit
         AND release_group.id = release.release_group
  
  
    \end{lstlisting}
    
\vspace{-4mm}
    \caption{Query 09ac (full enumeration)}
    \label{fig:queryNineAC}
\end{figure*}

\begin{figure*}[h]
\mbox{}
\bigskip

  \centering
\begin{minipage}{0.3\textwidth}
\centering
\includegraphics[width=\linewidth]{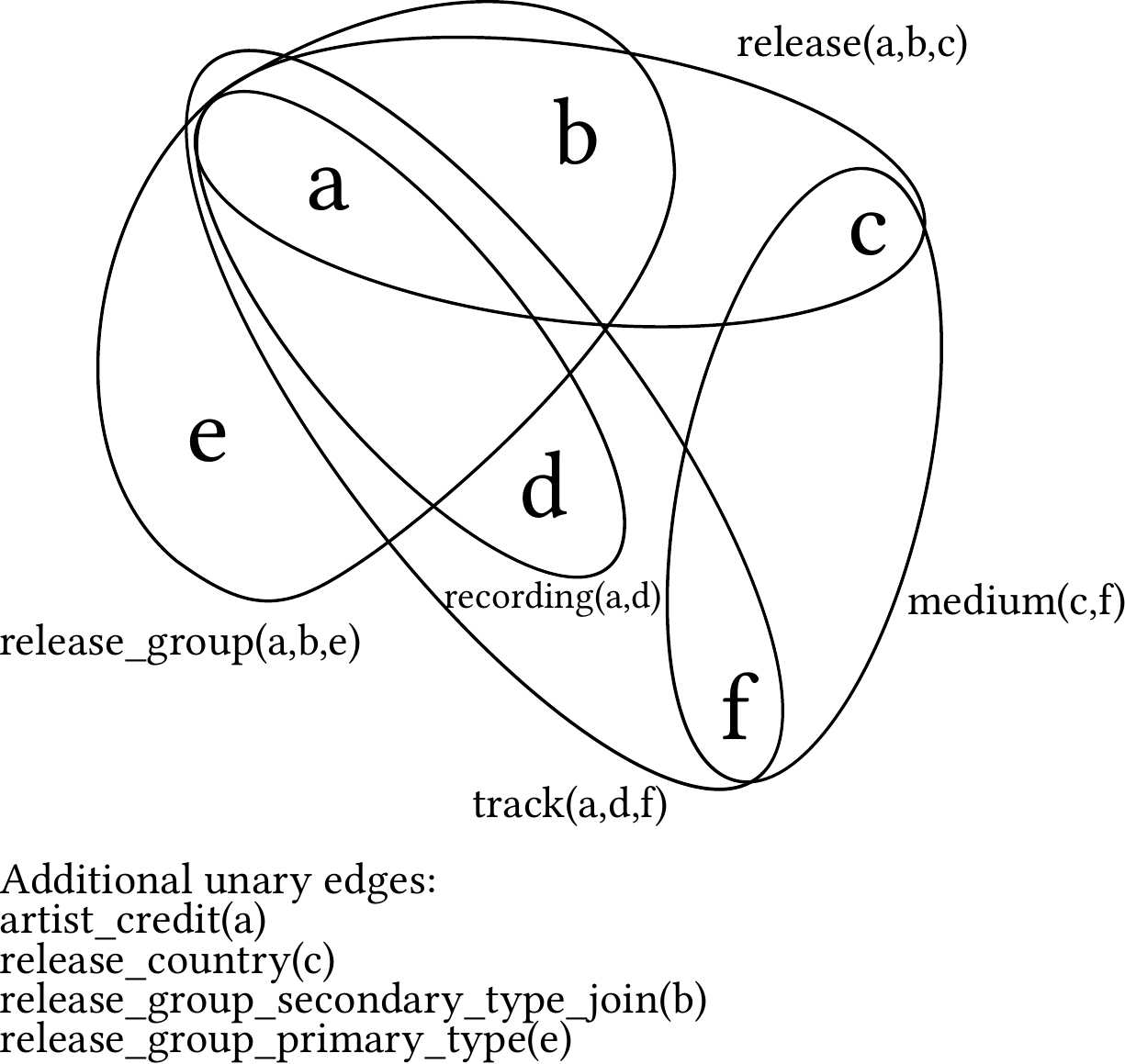}
\end{minipage}
\hfill
\begin{minipage}[c]{0.66\textwidth}
\begin{forest}
  my tree style2
  [, phantom, s sep = 8mm
  [ac $\bowtie$ rec, label={[label distance=1em]above:{\textbf{Fast} (16s)}}
    [r $\bowtie$ m
        [ t
            [rc]
        ]
    ]
    [   r 
        [r $\bowtie$ rg
            [rg  $\bowtie$ rgp]
        ]
    ]
    [r $\bowtie$ rgs]
  ]
  [m $\bowtie$ rg, label={[label distance=1em]above:{\textbf{Timeout (> 20min)}}}
    [ac, tier=A]
    [[rc $\bowtie$ t, tier=B]]        
    [r,tier=A]
    [[rec,tier=B]]
    [rgp, tier=A]
    [[rgs, tier=B]]
  ]  
    [r , label={[label distance=1em]above:{\textbf{Fast-2} (10s)}}
    [rg $\bowtie$ rgp]
    [[rgs,  tier=C]]
    [rec]
    [[ac, tier=C]]
    [m $\bowtie$ r
        [t
            [rc]
        ]
    ]
  ]
  ]
\end{forest}
\end{minipage}
  \caption{Hypergraph and different GHDs of the cyclic query 09ac}
  \label{fig:09acex}
\end{figure*}

\begin{table}[h]
    \caption{Abbreviations of attribute names in query q09ac}
    \centering
    	\vspace{-3mm}
    \begin{tabular}{l|l}
    \toprule
        relation plus schema & true attribute names (in the same order)\\
     \midrule
         medium(c,f) &  release, id \\
         recording(a,d) &  artist\_credit, id \\
         release(a,b,c) &  artist\_credit, release\_group, id \\
         release\_group (a,b,e) &   artist\_credit, id, type \\
         track(a,d,f) &  artist\_credit, recording, medium\\
    \bottomrule
    \end{tabular}
    \label{tab:abbreviations}
\end{table}

\begin{table}[h]
    \caption{Abbreviations of relation names in query q09ac}
    \centering
    	\vspace{-3mm}
    \begin{tabular}{l|l}
        \toprule
        abbreviation & true relation name\\
        \hline
         ac  &  artist\_credit \\
         m  &  medium \\
         r &  release \\
         rc & release\_country \\
         rec &  recording \\
         rg &   release\_group \\
         rgp &   release\_group\_primary\_type \\
         rgs &   release\_group\_secondary\_type\_join \\
         t &  track\\
    \bottomrule
    \end{tabular}

    \label{tab:abbreviationsRelations}
\end{table}
 
\else

\fi

\end{document}